\documentclass{article}
\usepackage{mathtools}
\usepackage{amsmath,amssymb,amsthm}
\usepackage{enumitem}
\usepackage{listings}
\usepackage{tcolorbox}
\usepackage{xcolor}
\usepackage{datetime}
\usepackage{booktabs}
\usepackage{float}
\usepackage{fancyvrb}
\usepackage{hyperref}

\lstdefinelanguage{Coq}{
  morekeywords={
    Require, Import, Parameter, Axiom, Definition, Lemma, Theorem,
    Corollary, Proof, Qed, Admitted, End, Section, Module, Record,
    Inductive, CoInductive, Fixpoint, CoFixpoint, Let, In, fun,
    forall, exists, match, with, return, as, if, then, else,
    Type, Set, Prop, True, False, and, or, not,
    exact, apply, intro, intros, split, left, right, destruct,
    unfold, rewrite, simpl, reflexivity, discriminate, contradiction,
    assumption, auto, trivial
  },
  sensitive=true,
  morecomment=[s]{(*}{*)},
  morestring=[b]"
}

\newcommand{\classP}{\mathsf{P}}
\newcommand{\classNP}{\mathsf{NP}}
\newcommand{\Prog}{\mathcal{X}}
\newcommand{\Lang}[1]{\mathrm{Lang}(#1)}

\newcommand{\OWF}{\mathrm{OWF}}

\newtheorem{theorem}{Theorem}[section]
\newtheorem{lemma}[theorem]{Lemma}
\newtheorem{corollary}[theorem]{Corollary}
\newtheorem{definition}[theorem]{Definition}

\newtheorem{example}[theorem]{Example}
\newtheorem{remark}{Remark}
\newtcolorbox{resultbox}{colback=gray!10,colframe=gray!50}
\newtheorem{defn1}[theorem]{Definition}
\newtheorem{defn2}[theorem]{Definition}
\newtheorem{thm}[theorem]{Theorem}
\newtheorem{coro}[theorem]{Corollary}
\newtheorem{proposition}[theorem]{Proposition}
\newtheorem{observation}[theorem]{Observation}
\newtheorem{principle}[theorem]{Principle}

\begin{document}

\title{Limits of Uniform Certification in the Standard Turing Model\\
\large Semantic Invariants and Admissible Methods}

       \author{
        Fabio F.G.\ Buono\\
        \small Independent Researcher\\
        \small ORCID: 0009-0004-9199-2793
      }

      \date{Preprint -- Version 6 --- \today}     
\maketitle

\begin{abstract}
  This work introduces a general obstructional framework, the ``Double Bind,'' 
  exposing the intrinsic structural limitations of formal 
  verification across computational complexity, mathematical physics and  
  broader theoretical domains. Structurally, it operates as a nested mathematical 
  case-switch mechanism: its two primary branches (horns) never activate 
  simultaneously, nor are both guaranteed to trigger, while the first branch 
  contains a nested case switch where exactly one specific case is activated. 
  We first demonstrate how the expectation of a formally verifiable proof 
  regarding SAT solver complexity (e.g., via Coq) conflicts with the fundamental 
  limits of computation. While this conflict \textit{prima facie} suggests the 
  logical undecidability of the P vs.\ NP problem, the Double Bind framework 
  resolves the apparent paradox. This obstruction encompasses all standard 
  complexity barriers--such as relativization, natural proofs, and 
  algebrization--while fundamentally blocking a vast class of potential 
  future proof strategies not covered by traditional limitations. To establish 
  its universality, we derive a second formulation via a prefix-closure topology 
  on Deterministic Turing Machine programs and Kolmogorov complexity. 
  This leads to the ``Principle of General Non-Measurability,'' proving that 
  such computational double binds stem from a deep, invariant geometric structure. 
  We show that this framework governs a wide class of foundational limits, analyzing 
  its profound implications across Geometric Complexity Theory (GCT), the 
  certification of the Langlands program, and the theoretical boundaries of 
  quantum supremacy. Finally, we expose a hidden, correlated assumption within 
  the theory of one-way functions.

  Given the subtle architecture of the Double Bind, the introductory sections 
  provide an essential narrative disambiguation to explicitly separate this 
  framework from conventional misapplications of Rice's theorem. Furthermore, 
  this preliminary discussion serves to properly ground an otherwise unconventional 
  paradigm, which remains highly susceptible to interpretive missteps if introduced 
  without the necessary conceptual scaffolding.
\end{abstract}

\newpage

\section{Origin of the idea}
\label{sec:origin}

\bigskip

\begin{tcolorbox}[colback=black!22, colframe=red, boxrule=3.5pt, left=18pt, 
  right=18pt, top=18pt, bottom=18pt, sharp corners, fontupper=\large]
  
  It is easy to conflate this strategy with well-known instances where Rice's theorem 
  has been improperly applied to draw conclusions about the 
  $\mathbf{P}$ vs.\ $\mathbf{NP}$ question. However, the underlying 
  reasoning is fundamentally different;

  \medskip

  Unlike a direct invocation of Rice's theorem on program certification, 
  this text demonstrates how an implicit decider emerges in specific cases 
  that correctly applies Rice's theorem.

  \bigskip

  The rigorous formalization is presented 
  in \ref{sec:doublebind}, and the MTB is presented in \ref{sec:IntroMTB} 
  and fully proved in \ref{sec:MTBDEF}.
  
\end{tcolorbox}

\medskip

This paper emerges from a fundamental consideration: not an application of Rice's 
theorem to arbitrary programs, but an investigation into how certification mechanisms 
propagate across syntactic variants, and what structural constraints this imposes. 

\medskip

What happens when we prove with pen and paper that quicksort has average-case 
computational complexity $O(n \log n)$? (Or equivalently, through a Coq kernel 
by virtue of the Curry-Howard isomorphism)? And why do I not need to prove it 
for every syntactic variant and every variant in different programming languages?

\smallskip

There is no need to do so because we work at an abstract level, but in doing 
so we are implicitly assuming that an implicit decider emerges in the proof, 
which decides a trivial property. This property states that "all syntactic 
variants of this algorithm have the same complexity." And this certification 
propagates to all syntactic variants of the program through a propagation 
mechanism that underlies the formalisms of formal semantics.

\smallskip

Attention: I did not say that the proof does not implicitly produce a decider 
that recognizes all syntactic variants. Rather, I said that the certification 
propagates to its syntactic variants, otherwise we would have to repeat the 
proof for every syntactic variant. This distinction is very important and 
must be kept in mind. Returning to the example, the certification of the 
algorithm, combined with the semantic propagator implicit in the formalism, 
induces a decider relative to the class of variants recognized by the formalism.

\medskip

This is nothing new, it is merely a formalization of something already 
implicitly present in the theory. 

\bigskip

When a theory proves that a property 
holds for an algorithm, it is not speaking only of a concrete string of 
symbols. It is speaking of an abstract object: a function, a recursive 
relation, a machine, a class of programs, or a specification, and 
therefore, implicitly, of all terms that fall under the same interpretation 
or are connected by the rules of the formalism. Formalization does not 
introduce the phenomenon; it makes explicit the generality already 
implicit in abstract analysis. When one asserts that 
``Quicksort has average-case complexity $O(n \log n)$,'' one is already 
abstracting away from the particular source text and assuming that the 
implementations considered realize the same algorithmic scheme under a 
compatible cost model. And this implicit decider decides a semantic 
property that the initial algorithm did not directly possess, 
it is a property emergent from algorithmic logic.

\medskip

Moreover, it is not guaranteed that only a single implicit decider emerges. 
We cannot know how many emerge; their number depends on which latent 
connections the local verification of the local algorithm implicitly induces. 
They are semantic consequences derived from families of propagators, not from 
latent connections in the sense of properties that the system discovers freely. 
In a sufficiently rich theory, the consequences can be extremely numerous, 
even infinite. However, their number is not normally a well-defined 
mathematical quantity: the same consequence can be rewritten in many ways, 
and new properties can be constructed by combining previous ones.
Thus, a local verification can generate a family, potentially very large, 
of global properties and propagation procedures, determined by the closure 
of the theory with respect to its inference rules and the semantic relations adopted.

\medskip

The emergent character does not imply that such deciders are necessarily or 
explicitly contained in the program. It means that they become available 
as consequences of its placement within a broader theoretical structure. 
And \textbf{what happens if one of the propagated properties is non-trivial 
according to Rice?} If we do not change the premises and do not weaken the 
discourse, then by virtue of what has been said, 
even the verification of the single local algorithm would be blocked. Thus the 
conclusion is precisely a reductio ad absurdum, but it is not the 
local verification of an algorithm like Quicksort that is impossible. The 
distinction is: if the propagator is merely a certification mechanism relative 
to a restricted class, there is no contradiction; if it is truly a total decider 
not only for all syntactic variants in the strong sense indicated but also for 
totally different algorithms that do totally different things, then Rice prevents 
it from existing. Therefore, under these premises, the consequence is that local 
verification cannot induce that effective decider without rendering one of the 
assumptions incoherent. The point to keep in mind is that it is standard to say 
that abstract Quicksort has average-case complexity O(n log n), but this means 
\textbf{all programs syntactically reducible to Quicksort.} This does not mean that 
quicksort cannot be certified; it means that there may exist local programs 
(obviously not quicksort, that was an example for which this does not occur) 
for which local certification that apparently is permitted is in fact impossible.

\bigskip

And now we arrive at the question: \textbf{do there exist problems that induce implicit 
deciders for non-trivial properties?}

\bigskip

The answer is yes: there do exist problems that induce implicit deciders 
for non-trivial properties, and the discourse we 
are about to develop must keep this starting point in mind. Neither classical 
computability theory nor traditional complexity theory 
provides an explicit mechanism to forbid this propagation.

\bigskip

This asymmetric ripple effect constitutes the core of the first horn of the Double Bind. 
An identical obstruction arises when attempting to formalize a proof 
of $\classP \neq \classNP$. Any such proof must be capable of being explicitly 
written, formalized, and mechanically verified within a deterministic 
framework (e.g., Coq). Yet, even in this case, the resulting certification 
propagates via $\classNP$-completeness to the broader complexity class. 
It is essential to emphasize that this framework does not merely reason 
about computational complexity; rather, it fundamentally shifts the 
foundational plane of the inquiry. Furthermore, this mechanism is not 
intrinsic to the SAT problem alone, but governs a wide class of problems 
sharing an isomorphic abstract propagation structure, such as the 
uniform certification of lower bounds. The second horn of the Double 
Bind defines the structural limits of any potential escape route 
from Rice's theorem. If, to bypass the constraints of the first horn, 
a proof relies essentially on computational powers that are non-simulable 
by a Deterministic Turing Machine (DTM), the resulting certification fails 
to retain its validity within the standard model. Bound together, these 
two asymmetric constraints constitute the Double Bind.

\bigskip

To render the architecture of the Double Bind as transparent as possible, 
we leverage the $\classP$ vs.\ $\classNP$ question as a concrete instance 
of the mechanism, though the underlying logic applies universally to any 
foundational problem exhibiting an identical abstract propagation structure.

\newpage

The Double Bind framework consists of two primary operational branches:

\bigskip

\begin{enumerate}[font=\large\bfseries, itemsep=2.5em, topsep=1em]
    \item \textbf{Horn 1 (Indirect Semantic Obstruction):}   The First Horn analyzes the two specular scenarios of the problem, highlighting how every formalization attempt triggers a logical short-circuit
  
    \medskip

    \begin{description}[leftmargin=1.5em, style=nextline, itemsep=2em]
        
        \item[\textbullet~Case 1 (Property $\Phi_\alpha$):] 
        Let us imagine taking an arbitrary program $x$ and wishing to certify its specific computational complexity, for instance, through a Coq kernel. Certifying the complexity of a local algorithm, such as quicksort, induces a localized certification restricted solely to its syntactic variants. 
        However, due to $\classNP$-completeness, the certification of a problem like SAT propagates instantly to every algorithm in $\classNP$, causing the class $\classNP$ to collapse onto $\classP$. $\classNP$-completeness acts as a propagator: the decision on the single algorithm induces an implicit certifier that attests to the global semantic property $\Phi_\alpha$ (``Every $\classNP$-complete algorithm admits a solution algorithm with polynomial complexity''). This is explicitly forbidden by Rice's Theorem.         
        Consequently, the certification of SAT, despite originating as a local property, is halted because it is forcibly transformed into a global one. The property $\Phi_\alpha$ does not originally belong to SAT; rather, it emerges as an emergent structural property through the action of the propagator.
        
        \item[\textbullet~Case 2 (Property $\Phi_\beta$):] 
        Let us now imagine the opposite scenario, namely a formal proof of the separation of classes, $\classP \neq \classNP$. To be valid, such a proof must be formalized and verifiable, for example, via the same Coq kernel. 
        Once verified, by virtue of the \textbf{Curry-Howard isomorphism}~\cite{howard1980formulae}, the proof is transformed for all intents and purposes into an executable program. Again due to $\classNP$-completeness, the code of this proof-program propagates and is blocked by Rice's Theorem.
        In this case, the implicit semantic property being certified is $\Phi_\beta$ (``Every $\classNP$-complete algorithm has superpolynomial minimum complexity''). It must be specified that if the formalization of the proof does not exploit a propagator in the sense required by the Double Bind, the proof would not be blocked, as in the case of $\mathrm{IP}=\mathrm{PSPACE}$~\cite{shamir1992ip}.
        
        \medskip

        \textbf{It must be emphasized that Rice's theorem is not applied directly to the SAT solver; instead, its invocation is triggered by the emergence of the implicit decision procedure.}
    \end{description}

    \bigskip

    \item \textbf{Horn 2 (MTB):} A proof that depends essentially on computational powers that are non-DTM (not simulable by the Deterministic Turing Machine) does not transfer its validity to the standard model. If you use computational powers that the DTM does not have, that proof is not valid in the standard Turing model.
    
\end{enumerate}

\bigskip

In essence, formal verification is tightly constrained from both directions: one 
cannot exploit extra-DTM capabilities due to the Model Transfer Barrier (MTB), 
nor can one remain purely within the DTM domain due to the \textbf{indirect} invocation 
of Rice's theorem. It must be understood that Extended Rice operates purely 
as the technical mechanism to formalize this intuition: whenever a system 
attempts to uniformly certify a non-trivial semantic property, it inductively 
constructs an implicit decider for that property. Once an implicit decider emerges, 
Rice's theorem rigorously invalidates the procedure. Extended Rice serves as the 
mathematical scaffolding required to justify this step, whereas the foundational 
logical power resides entirely within the interplay of the Double Bind: the MTB 
on one side, and Rice's theorem on the other.

\bigskip

This distinction explains why the automated verification of a localized algorithm 
like QuickSort remains unaffected by the Double Bind, while global problems 
like SAT do not. A localized certification affects only the target program,
whereas certifying SAT under a specific complexity bound induces a global 
semantic decision across the entire class. This insight is pivotal: the 
framework does not assert that \textit{every} complexity certification 
constitutes a non-trivial semantic property. Rather, it exposes a targeted 
barrier: when a certified property induces uniform consequences across an 
entire computational class, a universal certification procedure can no longer 
be treated as a localized verification. Extended Rice is the exact mechanism 
that formalizes this distinction.

\bigskip

Therefore, the Double Bind does not arise from a generic semantic interest in SAT, 
but from the fact that a local certification yields a global consequence via 
completeness. If a uniform method were to soundly and fully certify that a 
given algorithm decides SAT in polynomial time, it would not merely verify a 
localized program; it would construct an algorithmic object that collapses 
the $\classP$/$\classNP$ boundary. Crucially, the framework restricts the 
existence of a single, uniform test expressible as an algorithm verifiable 
by a DTM. This is why uniformity is defined on a per-input basis. 
This definition perfectly explains why the framework does not need 
to postulate a \textbf{global generator of all evidence}. The property of 
uniformity relates strictly to the applicability of the verification 
procedure $V$ to any input program, bypassing any requirement for a 
priori knowledge of all possible certifications. The algorithmic role 
necessary to trigger the obstruction is already fully executed by the 
verifier $V$ itself, establishing an absolute, topologically derived 
barrier to uniform formal verification.

\bigskip

\begin{remark}[On the structural inevitability of the induced decider]
\label{rem:propagator-structure}
  
  A critical question arises in understanding the mechanism by which the induced 
  decider emerges in the Double Bind: is it truly the case that when one attempts 
  to uniformly certify a property that propagates through a heterogeneous family 
  via a problem-specific propagator, one's method necessarily decides that property? 
  This question touches the heart of why the obstruction is structural rather than 
  methodological.
  
  \medskip

  The answer is affirmative, and it can be grasped by considering the contrapositive. 
  Suppose, hypothetically, that finding a polynomial-time SAT solver did not resolve 
  P versus NP. This would mean that one could locally certify that a single program 
  computes SAT in polynomial time without this certification implying that P = NP. 
  But this is a direct contradiction in the definitions of P and NP themselves. 
  
  \medskip

  To be explicit: if a program $x$ solves SAT in polynomial time, then by definition, 
  P = NP. There is no alternative. The propagation of the property from the single 
  solver to the entire class NP is not a feature introduced by the certification 
  method; rather, it is a logical property intrinsic to the structure of the problem 
  itself. If $x$ computes SAT in polynomial time, then by NP-completeness (the 
  Cook-Levin theorem), there exists a polynomial-time algorithm for every problem in 
  NP. 
  
  \medskip

  Therefore, { \bfseries the propagator is not something that the method constructs or imposes. 
  The propagator is something the problem structurally is.} To certify that program $x$ 
  solves SAT in polynomial time necessarily means to certify something about the entire 
  structure of computational classes, because that is what the definitions encode. 
  Certification of a single case thus automatically implicates the decision of the 
  entire property, not because the certification procedure was designed that way, but 
  because the logical and semantic content of what is being certified carries this 
  implication necessarily.
  
  \medskip

  { \bfseries Consequently, the induced decider does not emerge as a consequence of the structure 
  of the certification method. Rather, it emerges as a consequence of the structure of 
  the problem itself.} When a method uniformly certifies membership in a semantically 
  non-trivial property that propagates via a problem-specific mechanism, the method 
  does not create the decision procedure—the problem's own structure guarantees that 
  uniform certification of any instance amounts to deciding the property over the 
  entire family. This is why Rice's theorem applies not as an external constraint 
  imposed upon the method, but as a direct consequence of the semantic content being 
  certified.
  
\end{remark}

\newpage
\bigskip

\begin{remark}[The Double Bind requires no class separation or collapse]
  \label{rem:no-collapse}
  The case of $\classP$ vs.\ $\classNP$ may mislead, because there the
  invariant $\Phi_\alpha$ carries a consequence over the whole class
  $\classNP$ through the completeness of SAT. But that class-wide
  consequence is a feature of the particular case, not a requirement of the
  Double Bind. The Double Bind is not a statement about separations or
  collapses of complexity classes: it is a statement about the certification
  of non-trivial semantic invariants. The only thing the first horn requires
  is that the property to be certified be semantic --- depending on the
  function computed --- and non-trivial. Nothing in this condition mentions
  complexity classes, separations, or collapses.
  
  \bigskip

  What distinguishes a certification that falls under the Double Bind from one
  that does not is not whether it collapses a class, but whether it remains
  local or propagates. And the propagation that matters is not propagation to
  the different syntactic forms of one and the same algorithm --- that is
  trivial. Certifying the complexity of QuickSort propagates only to QuickSort:
  to the same algorithm implemented in different syntactic forms, in different
  languages, with different variable names. It remains local to a single
  algorithm. What the first horn requires is propagation to \emph{different}
  algorithms --- algorithms that do different things, joined only by falling
  under the same invariant. Certifying an invariant such as the one-wayness of
  a function reaches a heterogeneous family of different algorithms; certifying
  the complexity of a single sorting routine reaches only the syntactic
  variants of that one routine. It is propagation to different algorithms, and
  not any class collapse, that makes an invariant non-trivial in the sense the
  first horn requires.
  
  \bigskip

  Accordingly, the Double Bind is far more general than the $\classP$ vs.\
  $\classNP$ case suggests: it applies to every non-trivial semantic invariant
  --- to cryptographic hardness assumptions, to lower bounds, and to any
  property of the behaviour of a function whose certification does not remain
  confined to a single algorithm. Class separation is one of its instances,
  not its premise.
\end{remark}

\section{The Apparent Paradox of the Double Bind
\label{sec:ap-paradox}}

A consequence of the Double Bind may seem paradoxical, but in reality it is not. 
Here is the apparent paradox explained clearly:

\medskip

On one hand, the Double Bind tells us that \textbf{if there existed an algorithm that solves $\text{SAT}$ in polynomial time, its certification in the Turing model would be impossible}. 

\medskip

On the other hand, at the same time it tells us that \textbf{even an automated validation that assures us the opposite, that is, that a SAT solver can never be in $\mathcal{P}$, is blocked}.

\medskip

This seems to suggest that \textbf{P vs. NP is entirely non-certifiable from within the standard deterministic Turing model}. The Double Bind appears, 
that is, as a total obstacle. But this is not the real meaning of the Double Bind.

\medskip

The apparent paradox described here is perhaps the most difficult point to reconcile. 
To address this, the second part of the document (Section \ref{sec:toread}) provides 
a topological explanation. Although offered as a preview of ongoing research, this 
approach reveals that the conclusion is an inescapable consequence of the theory.

\bigskip

\subsection{Alternative Paths to Certification}

Although it is non-standard to think of it in these terms, a certification of an 
algorithm's membership in a certain class, as well as a proof of separation between 
P and NP, can come from a different path. An example is a metatheoretic proof that 
acts directly on the foundations of the standard theory on which the concepts 
of ``Turing Machine'', ``Algorithm'', ``Computational Complexity'', and 
``Karp Reduction'' etc. are defined. Indeed, in this case the requirements of the 
Double Bind are not satisfied. To truly understand this argument, the Double Bind 
must be framed within its most basic structure.

\medskip

The observational axis introduced in \cite{buono2026hierarchy,buono2026observer} is 
another one such candidate. An example of a proof that is proposed and that does not fall into 
the grip of the Double Bind, without extending the model, is presented in \cite{buono2026syntaxseedynamicsyntactic}

\bigskip
\bigskip

\textbf{We are now positioned to unambiguously examine the rigorous formalization of the first form of the Double Bind framework.}

\newpage

\section{The Double Bind}
\label{sec:doublebind}

\begin{definition}[The Computational Space]
Let $\mathcal{T}$ denote the standard Turing machine model. Let $\mathcal{P}$ be the 
set of all programs (represented as syntactic strings) within $\mathcal{T}$. 
Let $Q$ be a fixed $\text{NP}$-complete problem (e.g., $\text{SAT}$).
\end{definition}

\begin{definition}[The Semantic Horns]
We define two mutually exclusive semantic properties 
$\Phi_{\alpha}, \Phi_{\beta}: \mathcal{P} \to \{0, 1\}$ acting on the extensional 
behavior of a program $x$:

\begin{enumerate}

  \item \textbf{The Collapse Horn} ($\Phi_{\alpha}(x)$): Attests that ``the specific program $x$ decides the $\text{NP}$-complete problem $Q$ in polynomial time.'' The validity of $\Phi_{\alpha}(x)$ for a single $x$ implies the global collapse $\text{P} = \text{NP}$.

  \item \textbf{The Separation Horn} ($\Phi_{\beta}(x)$): Attests that ``the specific program $x$ (which is a correct solver for $Q$) can never operate within a polynomial time bound ($\forall k \in \mathbb{N}, \text{time}_x(n) = \omega(n^k)$), thereby certifying that the instance $Q$ is intrinsically super-polynomial.'' The validity of $\Phi_{\beta}(x)$ implies the global separation $\text{P} \neq \text{NP}$.

\end{enumerate}
\end{definition}

\begin{definition}[The Propagator $\mathcal{R}$]

The Propagator $\mathcal{R}$ is the logical operator induced by the structure 
of $\text{NP}$-completeness. It maps the local semantic status of a single 
program $x$ to the global state $\Psi$ of the complexity classes 
($\Psi = 1$ if $\text{P} = \text{NP}$, else $\Psi = 0$).

\begin{itemize}

  \item $\mathcal{R}(\Phi_{\alpha}(x) = 1) \implies \Psi = 1$ (via the polynomial-time reduction of any $L \in \text{NP}$ to $Q$).

  \item $\mathcal{R}(\Phi_{\beta}(x) = 1) \implies \Psi = 0$ (via the certification of a fundamental lower bound on $Q$).

\end{itemize}
\end{definition}

\begin{definition}[Uniformity Per-Input]

A certification method is said to be uniform per-input if there exists a fixed 
pair of algorithms $(M_{\text{Adm}}, V)$ such that for every program 
$x \in \mathcal{P}$, the generator $M_{\text{Adm}}$ produces a 
certificate $\pi_x$ specifically for $x$, and the verifier $V$ checks 
the validity of $\pi_x$ relative to $x$. The uniformity resides 
in the fact that $M_{\text{Adm}}$ and $V$ are the same for all $x \in \mathcal{P}$.

\end{definition}

\begin{definition}[The Admissible Method]

A certification method $(M_{\text{Adm}}, V)$ is admissible if it satisfies the 
following conditions:

\begin{enumerate}

  \item \textbf{Soundness}: $V(x, \pi_x) = 1 \implies \Phi(x) = 1$.

  \item \textbf{Completeness}: $\Phi(x) = 1 \implies \exists \pi_x$ such that $V(x, \pi_x) = 1$.

  \item \textbf{Termination}: The process of generating $\pi_x$ and verifying it via $V$ terminates in finite time for any $x \in \mathcal{P}$.

\end{enumerate}
\end{definition}

\begin{definition}[Rice's Theorem \& the Extended Wrapper]

  Rice's Theorem states that any non-trivial semantic property of a program is 
  undecidable. We refer to the application of this theorem to complexity-related 
  invariants as the Extended Rice wrapper, which allows the general obstruction 
  of Rice's Theorem to be called indirectly for properties like $\Phi_{\alpha}$ 
  and $\Phi_{\beta}$.

\end{definition}

\section{Branch I: The Ricean Ostracism}

\begin{theorem}

  No uniform admissible certification method exists for $\Phi_{\alpha}$ 
  or $\Phi_{\beta}$ within the standard Turing model $\mathcal{T}$.

\end{theorem}

\begin{proof}

\noindent \textbf{Assumption of a Uniform Admissible Method:} Assume there exists 
a method $(M_{\text{Adm}}, V)$ that is both uniform per-input (Definition 4) 
and admissible (Definition 5) for an invariant $\Phi \in \{\Phi_{\alpha}, \Phi_{\beta}\}$.

\noindent \textbf{Construction of the Induced Decider (The Composition):} We construct 
a composite Turing machine $D_{\Phi}$ as follows:

$$\text{On input } x \in \mathcal{P}: \text{Return } V(x, M_{\text{Adm}}(x))$$

Because the method is admissible, $M_{\text{Adm}}$ and $V$ always terminate in 
finite time. Thus, $D_{\Phi}$ is a totally computable function. This composition 
effectively upgrades the local uniformity (per-input) to a standard global uniformity.

\medskip

\noindent \textbf{The Ripple Effect and the Role of the Propagator:} The obstruction here does not depend on the computational difficulty of the problem, but on the reach of the truth being certified.

\begin{itemize}

  \item For a non-complete problem (e.g., Quicksort), certification is a local property of the code; it does not induce a global decider.

  \item However, for an $\text{NP}$-complete problem $Q$, the Propagator $\mathcal{R}$ (Definition 3) transforms the local output of $D_{\Phi}(x)$ into a global decision about the state $\Psi$:

$$D_{\Phi}(x) = 1 \iff \mathcal{R}(\Phi(x)) \text{ is true} \iff \Psi \text{ is in the state induced by } \Phi$$

Consequently, the language

\begin{equation*}
\mathcal{L}_{\Phi} = \{ x \in \mathcal{P} \mid D_{\Phi}(x) = 1 \}
\end{equation*}

describes the global extension of a \textbf{non-trivial semantic invariant} of the 
computational model.

\end{itemize}

\noindent \textbf{The Obstruction of Rice:} The property $\Phi$ is semantic (it is an extensional property of the function computed by $x$, regardless of its specific complexity value) and non-trivial. By the Extended Rice wrapper, this property is subject to the general obstruction of Rice's Theorem. Therefore, the language $\mathcal{L}_{\Phi}$ is strictly undecidable.

\medskip

\noindent \textbf{The Contradiction:}
\begin{itemize}
\item From Step 2: $D_{\Phi}$ is a totally computable function $\implies \mathcal{L}_{\Phi}$ is decidable.
\item From Step 4: $\mathcal{L}_{\Phi}$ is a non-trivial semantic invariant $\implies \mathcal{L}_{\Phi}$ is undecidable.
\end{itemize}
\noindent Therefore, the assumption of a uniform admissible method is false. No such method can exist within $\mathcal{T}$.
\end{proof}

\medskip

\section{Branch II: The Model Theoretical Barrier (External Model $\mathcal{T}'$)}

\begin{theorem}
The resolution of the Ricean obstruction via a non-standard model $\mathcal{T}'$ fails to constitute a certification within the standard model $\mathcal{T}$.
\end{theorem}

\begin{proof}
To bypass the contradiction in Theorem 1, one must replace the computable components $(M_{\text{Adm}}, V)$ with non-computable operators $\mathcal{O}$ in a higher-order model $\mathcal{T}'$. However, the Model Theoretical Barrier (MTB) stipulates that for a certificate to be ``formally valid'' for the Turing model $\mathcal{T}$, its verifier $V$ must be executable within $\mathcal{T}$. 

If $V$ requires $\mathcal{T}'$ to operate, the resulting proof is a statement about the capabilities of $\mathcal{T}'$, not a syntactic certificate valid in $\mathcal{T}$. The MTB prevents the translation of $\mathcal{T}'$-validity back into $\mathcal{T}$-certifiability, as such a translation would re-introduce the Ricean undecidability within $\mathcal{T}$.
\end{proof}

\section{Synthesis: The Final Double Bind}

The ``Double Bind'' is the structural deadlock resulting from the intersection of the 
two branches. Any attempt to uniformly certify $\text{P}$ vs $\text{NP}$ is forced 
into a binary failure:

\begin{align*}
  \mathbf{Path\ A} &: \text{Uniformity per-input} \xrightarrow{\text{Composition}} \text{Standard Uniformity} \\
  &\quad \xrightarrow{\text{Propagator } \mathcal{R}} \text{Semantic Invariant} \xrightarrow{\text{Rice}} \mathbf{Contradiction} \\[0.5em]
  \mathbf{Path\ B} &: \text{Exit to } \mathcal{T}' \xrightarrow{\text{Bypass Rice}} \text{Hyper-computation} \\
  &\quad \xrightarrow{\text{MTB}} \mathbf{Invalidity \text{ in } \mathcal{T}}
\end{align*}

\medskip

The Double Bind establishes that while $\text{P} \neq \text{NP}$ (or $\text{P} = \text{NP}$) may be a 
mathematical truth, the uniform certification of this truth is structurally 
impossible within the standard Turing model. The Propagator $\mathcal{R}$ 
ensures that any local attempt to certify a solver inevitably triggers 
the global ``immune response'' of the model's intrinsic undecidability 
via Rice's Theorem. This result does not imply the logical independence 
of $\text{P}$ vs $\text{NP}$ from ZFC, but establishes a meta-computational 
barrier between syntactic verification and semantic invariants.

\begin{remark}[Absence of Necessity for Class Collapses]
\label{rem:no-class-collapse}
    The efficacy of the Double Bind does not require the semantic invariant to 
    induce an actual collapse of complexity classes (a phenomenon specific to 
    the completeness of $\text{SAT}$). The sole constraint required for the 
    activation of the first branch of the barrier is that the analyzed property 
    be semantic and non-trivial-namely, that it propagates to a heterogeneous 
    family of textually disjoint algorithms, thereby distinguishing itself from 
    purely local properties (such as the computational time of $\texttt{QuickSort}$) 
    that remain fully certifiable within the DTM. 
    Conversely, the condition that triggers the second branch is the utilization 
    of computational powers that the DTM cannot simulate, as rigorously defined 
    in Appendix~\ref{sec:MTBDEF}.
\end{remark}

\section{The Inescapable Propagation of the Propagator $\mathcal{R}$}

\begin{remark}[The Propagation Effect is Inevitable and Indisputable]

The propagating effect of the Propagator $\mathcal{R}$ is inevitable and indisputable 
because it is not an assumption of this paper, but rather a direct consequence 
of the very definition of NP-completeness. There is no way to circumvent it: it 
is an algebraic fact (Polynomial Reduction).

\end{remark}

\begin{proposition}[The Algebraic Foundation of NP-Completeness]

  The definition of $\text{NP}$-completeness establishes that any 
  problem $L \in \text{NP}$ can be reduced to $Q$ in polynomial time 
  via a function $f$. If you certify that there exists a program $x$ 
  that solves $Q$ in polynomial time, you have automatically created 
  an algorithm for every $L \in \text{NP}$ (namely: $x(f(w))$).

  \medskip

  There is no way to ``restrict'' this truth to the program $x$ alone: 
  the moment $x$ lies in $\text{P}$, the entire class $\text{NP}$ collapses 
  into $\text{P}$ by definition. The Propagator $\mathcal{R}$ is simply
  the name we assign to this mathematical law.

\end{proposition}

\begin{remark}[Independence from Implementation]

  It is implementation-independent: it does not matter how you write 
  your SAT solver or which programming language you use. If the 
  extensional behavior of that program (its mathematical function) 
  is polynomial, the property propagates. You may change the syntax, 
  you may optimize the code, but you cannot alter the fact that $Q$ 
  is the ``central hub'' of $\text{NP}$. If the central hub falls, 
  the entire network falls with it.

\end{remark}

\begin{proposition}[The Inevitability of the Implicit Decider]

  As already explained even if you do not wish to construct a global decider, 
  the existence of a uniform certification method constructs one for you. 
  If there exists a method that states ``Yes, this specific SAT solver 
  is in $\text{P}$'', that method is, in fact, an oracle 
  that answers the question ``$\text{P} = \text{NP}$?''. The 
  propagating effect transforms your ``small'' local certifier 
  into a ``large'' global decider. It is not a matter of user 
  choice; it is an emergent property of the system.

\end{proposition}

\newpage

\section{The Double Bind as a Universal Structural Law of Turing-Complete Systems}

\bigskip

\subsection{Universality Beyond P vs NP}

The remarkable insight is that the proof does not hold exclusively for the P vs NP 
problem, but rather wherever there exists a computational representation of a problem 
and wherever there exists a propagating property. The $\phi$ of the first horn are 
not necessarily limited to two; they number two precisely because the problem admits 
two possible outcomes. Were there three outcomes, we would have three $\phi$.

\bigskip

\begin{theorem}[The Double Bind as a Structural Invariant]

The ``Double Bind'' is not a characteristic specific to $\text{P}$ versus $\text{NP}$, 
but constitutes a structural law governing every certification system within a 
Turing-complete model.

\medskip

The obstruction is valid for any system satisfying the following three fundamental 
requirements:

\begin{enumerate}

  \item \textbf{Computational Representation:} The problem must be expressible as a semantic property of programs or strings in a computational model $\mathcal{T}$.

  \item \textbf{Semantic Foundation:} The truth value must pertain to extensional behavior (the behavior of the system), not to syntactic form (the representation).

  \item \textbf{Presence of a Propagator $\mathcal{R}$:} A reduction mechanism must exist (such as $\text{NP}$-completeness, Turing-completeness, or any other form of universal reduction) that propagates the truth of a local instance to the truth of an entire set or class.

\end{enumerate}

\medskip

Under these conditions, any uniform admissible certification method for any possible 
outcome necessarily induces a total computable decider for a non-trivial semantic 
invariant of $\mathcal{T}$, contradicting Rice's Theorem.

\end{theorem}

\newpage

\subsection{The Modularity of $\Phi$}

The two ``cases'' ($phi$) are not a limitation of the theorem, but rather a 
direct reflection of the cardinality of possible outcomes for the problem under 
consideration.

\begin{observation}[Cardinality of Outcomes]

\medskip

In $\text{P}$ vs $\text{NP}$, the outcome is binary 
($\text{P} = \text{NP}$ or $\text{P} \neq \text{NP}$), 
hence we have exactly two $\Phi$. If we were confronted with a 
problem admitting $n$ possible outcomes 
(for instance, a complexity classification spanning $n$ distinct levels), 
we would possess $n$ cases in the first horn: 
$\Phi_1, \Phi_2, \dots, \Phi_n$. The obstruction imposed by Rice's 
Theorem remains identical for every 
individual case: regardless of the number of possible outcomes, 
the act of certifying any one of them uniformly induces a computable 
decider for that specific semantic property. Notably, the various 
branches of the first horn are mutually exclusive, 
whereas the two principal horns can both engage in some instances—, 
yet never at the same time.

\end{observation}

\begin{remark}[Extensions]

  The first horn of Double Bind generalizes to a ``Multi-Phi'' structure when the problem 
  admits more than two outcomes. However, the logical trap remains structurally 
  invariant:

$$\text{Local Certification} \xrightarrow{\mathcal{R}} \text{Global Invariant} \xrightarrow{\text{Rice}} \text{Contradiction.}$$

  No matter how many outcomes the dilemma possesses, the mechanism of obstruction 
  in the first horn operates identically: every path toward uniform admissible 
  certification of any outcome is blocked by the same Ricean undecidability.

\end{remark}

\subsection{Exemplary Applications of the Universal Obstruction}

\begin{example}[Type Systems and Formal Verification]

  If one attempts to certify that a single program possesses a security property 
  that is ``universal'' in scope—meaning that by reduction, it would imply the 
  security of all programs belonging to that class—the certification system 
  encounters the Double Bind against Rice's Theorem.

  \medskip
  
  The Propagator $\mathcal{R}$ here is induced by the type system's universality 
  claims: proving one program secure within the type does not merely characterize 
  that program, but implicitly induces a decision about the security of the entire class.

\end{example}

\begin{example}[Cryptographic One-Way Functions]

  If one seeks to certify that a specific function is a ``one-way function'' (OWF),
  such that this certification, by structural necessity, implies the existence of OWF as a general and non-trivial property, the system encounters the Ricean obstruction.

  \medskip

  The Propagator $\mathcal{R}$ is the reduction from the existence of OWF to the 
  existence of pseudorandom generators, collision-resistant hash families, 
  and other cryptographic primitives. Certifying a single function as OWF 
  propagates to a global statement about the cryptographic landscape, 
  triggering the obstruction.

\end{example}

\subsection{The Generalization Principle}

\begin{principle}[Universality of the First Horn Binding Obstruction]

  For any problem $\Pi$ in a Turing-complete model $\mathcal{T}$ that satisfies the three requirements (Computational Representation, Semantic Property, and Propagating Mechanism), the following holds:

\begin{itemize}

  \item If $\Pi$ admits $n$ possible outcomes, then $n$ semantic horns exist: $\Phi_1, \Phi_2, \dots, \Phi_n$.

  \item Any attempt at uniform admissible certification of outcome $i$ induces a total computable decider for the global state corresponding to outcome $i$.

  \item This decider contradicts Rice's Theorem, as it decides a non-trivial semantic property of programs in $\mathcal{T}$.

\end{itemize}

Therefore, a ``first Horn $n$-ary obstruction'' necessarily obstructs the 
formation of a uniform certification method.

\medskip

The Double Bind is not a peculiarity of P vs NP, but a manifestation of a deeper 
structural law: the incompatibility between local semantic certification and global 
Propagators in Turing-complete models.

\end{principle}

\subsection{A Structural Law of Certification Systems}

The Double Bind represents a fundamental obstruction in the architecture 
of certification systems within Turing-complete models. It is not an 
artifact of the specific problem (P vs NP), nor is it a limitation of 
current proof techniques. Rather, it is a structural law 
that emerges from the interaction of three elements:

\begin{enumerate}
\item The computability of the underlying model,
\item The semanticity of the property being certified,
\item The propagating nature of reductions that universally link local instances to global invariants.
\end{enumerate}

This law holds universally for all problems satisfying these three conditions, 
irrespective of the number of possible outcomes, the specificity of the domain, 
or the sophistication of the certification method. The obstruction cannot be 
circumvented by technical ingenuity alone; it can only be transcended by exiting 
the Turing-complete framework or by avoiding the Propagator mechanism altogether, 
as occurs in proofs that rely on purely structural (non-propagating) properties.

\bigskip

\section*{Generalized Demonstration Independent of Problem Outcome Cardinality}

Thus we observe the generalized demonstration independent of the number of outcomes 
of the problem. Here it becomes clear that the double bind is a case switch: 
the extended rice horn or the MTB horn is activated, but in some cases both c
ould be activated, though never simultaneously. And the extended 
rice horn has $N$ branches that can \textbf{never} be activated together and 
that depend on the possible outcomes of the problem.

The system analyzes the nature of the certification and ``activates'' the 
corresponding obstruction. If you attempt to remain within the model $\mathcal{T}$, 
you activate the switch toward Extended Rice (which in turn branches into $N$ 
horns depending on the outcomes of the problem). If you attempt to 
exit the model, you activate the switch toward the MTB.

\bigskip

\section*{The General Theory of Double Bind}

\subsection*{Fundamental Definitions}

\textbf{Definition 1 (Computational Model).} Let $\mathcal{T}$ be a standard model of 
computation and $\mathcal{P}$ be the set of all programs within $\mathcal{T}$.

\textbf{Definition 2 (The Semantic Invariants - The $N$-Horns).} Let $\Psi$ be a global 
property of the model $\mathcal{T}$ that can have $N$ possible mutually exclusive 
outcomes $\{\psi_1, \psi_2, \dots, \psi_n\}$. We define $N$ semantic properties 
$\Phi_1, \dots, \Phi_n: \mathcal{P} \to \{0, 1\}$ such that for each $i \in \{1, \dots, n\}$:
\[
\Phi_i(x) = 1 \iff \text{The program } x \text{ certifies that the global state } \Psi = \psi_i.
\]
These are the $N$-Horns of the problem. Because the outcomes of $\Psi$ are mutually 
exclusive, only one $\Phi_i$ can be true for a correct solver $x$.

\textbf{Definition 3 (The Propagator $\mathcal{R}$).} The Propagator $\mathcal{R}$ is 
the logical operator that transforms a local certification of a program $x$ into a 
global truth about $\Psi$.
\[
\mathcal{R}(\Phi_i(x) = 1) \implies \Psi = \psi_i
\]
This implies that any uniform method to certify $\Phi_i$ is implicitly a method 
to decide the global state $\Psi$.

\textbf{Definition 4 (Uniform Admissible Certification).} A method $(M_{\text{Adm}}, V)$ 
is uniform per-input and admissible if it provides a computable way to generate ($\pi_x$) 
and verify ($V$) a certificate such that $V(x, \pi_x) = 1 \iff \Phi_i(x) = 1$ for 
any $x \in \mathcal{P}$.

\newpage

\subsection*{The Logic of the Double Bind (The Case Switch)}

The ``Double Bind'' is a structural deadlock. Any attempt to certify $\Phi_i$ triggers 
a logical switch between two insurmountable barriers. The path taken depends on the 
model of verification chosen:

\subsubsection*{CASE A: Verification within the Standard Model $\mathcal{T}$}

If the certification $(M_{\text{Adm}}, V)$ is implemented within $\mathcal{T}$:

\begin{enumerate}

\item \textbf{Induced Decider:} The uniformity per-input allows the construction of a total 
computable function $D_{\Phi}(x) = V(x, M_{\text{Adm}}(x))$.

\item \textbf{Activation of the Propagator:} $D_{\Phi}$ does not just analyze $x$; via 
$\mathcal{R}$, it decides the global state $\Psi$.

\item \textbf{The Extended Rice Obstruction:} Since $\Phi_i$ is a non-trivial semantic invariant, 
the Extended Rice Theorem dictates that $D_{\Phi}$ cannot be computable.

\item \textbf{The Multi-Horn Conflict:} This obstruction applies to all $N$ horns. No matter 
which outcome $\psi_i$ is the true one, the act of certifying it uniformly induces a decider 
that Rice prohibits. Thus, all $N$ branches of the Extended Rice horn are blocked.

\end{enumerate}

\subsubsection*{CASE B: Verification via an External Model $\mathcal{T}'$}

If the certification $(M_{\text{Adm}}, V)$ is implemented using an external, non-standard model 
$\mathcal{T}'$ (e.g., an Oracle) to bypass Rice:

\begin{enumerate}

\item \textbf{The Model Theoretical Barrier (MTB):} A certificate $\pi$ generated and verified 
in $\mathcal{T}'$ is not a syntactic certificate valid within $\mathcal{T}$.

\item \textbf{Loss of Model Validity:} The verification $V$ is no longer a computable function 
in $\mathcal{T}$. The ``proof'' is no longer a certification of the Turing model's properties, 
but a statement about the capabilities of $\mathcal{T}'$.

\end{enumerate}

\newpage
\subsection*{Formal Theorem of the Barrier}

\textbf{Theorem.} Let $\Phi_1, \dots, \Phi_n$ be $N$ mutually exclusive propagating semantic 
properties in $\mathcal{T}$. No uniform admissible certification method exists for any $\Phi_i$ 
within $\mathcal{T}$.

\textbf{Proof:}

The certification attempt triggers a Switch Mechanism:
\[
\text{Certification Attempt} \implies \begin{cases} 
\text{Verify in } \mathcal{T} \xrightarrow{\text{Propagator } \mathcal{R}} 
\text{Induced Decider} \xrightarrow{\text{Extended Rice}} \mathbf{\text{Contradiction}} 
\\[0.5em]
\text{Verify in } \mathcal{T}' \xrightarrow{\text{Bypass Rice}} \text{Extra-Turing-computation} 
\xrightarrow{\text{MTB}} \mathbf{\text{Invalidity in } \mathcal{T}}
\end{cases}
\]

\textbf{Observations:}

\begin{itemize}

\item \textbf{Non-Simultaneity:} The layered structure ensures that the two barriers (generally) cannot be 
active at the same time. One either stays in $\mathcal{T}$ (activating Rice) or exits to 
$\mathcal{T}'$ (activating MTB).

\item \textbf{Independence from Complexity:} The obstruction is not caused by the ``hardness'' 
of calculating $\Phi_i$, but by the logical topology of the model. The Propagator $\mathcal{R}$ 
ensures that local certification $\equiv$ global decision.

\item \textbf{Mutual Exclusivity of $\Phi$ Properties in the First Horn:} Since the $\Phi_i$ are mutually exclusive, the 
``Extended Rice'' horn is not a single wall, but a set of $N$ walls. The certifier must pick 
one horn to prove, but regardless of the choice, the resulting global decider is blocked by Rice. 

\end{itemize}

\hfill $\square$

\section{The Theorem of first horn in Its Correct Abstract Form}

\begin{theorem}[The Meta-Computational Barrier]
\label{thm:meta-barrier}
Let $\Pi$ be a problem with computational representation in $\mathcal{T}$, with $N$ 
mutually exclusive outcomes $\{\psi_1, \dots, \psi_n\}$, and let $\mathcal{R}_\Pi$ 
be the Propagator induced by the structure of $\Pi$. If $\mathcal{R}_\Pi$ is 
non-trivial --- that is, if local certification of $\Phi_i(x)$ for a single $x$ 
induces a global decision on $\Psi$ --- then no uniform admissible method 
exists for any $\Phi_i$ in $\mathcal{T}$.
\end{theorem}

\noindent
The Propagator enters as a hypothesis, not as an object defined by the theory. 
Precisely as non-triviality enters Rice's Theorem as a hypothesis.

\subsection*{Why the Propagator Is Not Defined}

This choice carries an important consequence for the scope of the result. 
Defining the Propagator in general would produce two detrimental effects:

\begin{enumerate}
\item \textbf{Restriction of scope:} A formal definition would exclude all Propagators that do not fall within the defined form, even if structurally analogous.
\item \textbf{Shift of proof burden:} The development would need to demonstrate that each specific Propagator falls within the general definition, multiplying the cases to be treated.
\end{enumerate}

By not defining it, the theorem becomes a \textit{theorem schema} --- a template 
that instantiates each time a specific problem exhibits its own natural Propagator. 
As Rice's Theorem is a schema that instantiates over each specific non-trivial semantic property.

\section{The Langlands Program as an Instance of the Double Bind}

The Langlands Program is the type of structure to which the Double Bind applies, 
but with a Propagator of a nature entirely different from that of NP-completeness.

\subsection*{The Langlands Propagator}

In P versus NP, the Propagator is a polynomial reduction --- a computational object. 
In Langlands, the Propagator is a correspondence between mathematically profoundly 
distinct categories:

\begin{equation}
\mathcal{R}_{\text{Langlands}}: \text{Galois Representations} \longleftrightarrow \text{Automorphic Forms}
\end{equation}

The propagating structure here is not computational; it is categorical. 
Certifying that a specific object on the Galois side corresponds to a 
specific object on the automorphic side instantaneously implies a global 
truth about the correspondence of the entire class.

\subsection*{Why the Double Bind Applies}

The three requirements are satisfied:

\begin{enumerate}
\item \textbf{Computational representation.} A formal proof within the Langlands Program, to be certified, must be representable on a DTM, as with any other formalizable mathematical proof.

\item \textbf{Semantic property.} The Langlands correspondence is a property of the behavior of the mathematical objects involved, not of their syntactic representation. Two syntactically distinct objects can be in correspondence; the property depends on behavior, not form.

\item \textbf{Non-trivial Propagator.} The functorial structure of the Langlands correspondence is globally propagating by construction. You cannot certify an isolated case without implying the structure of the entire correspondence, because the correspondence is defined precisely as a universal law between classes.
\end{enumerate}

This means something very precise about the Langlands Program:

\begin{quote}
No uniform admissible method can automatically certify that a specific 
object satisfies the Langlands correspondence, because such certification 
implicitly induces a decider for a non-trivial global semantic property.
\end{quote}

But the most interesting aspect is not the impossibility itself. It is that 
the Double Bind structurally explains why the Langlands Program is so 
difficult to formalize automatically, despite specific cases having 
been proven --- such as Wiles's proof of Fermat's Last Theorem,
which is essentially a local instance of the correspondence.

\subsection*{The Double Bind as a Law of Universal Correspondence}

Thus the Double Bind is not a law about computational complexity; it is a 
law about the structure of universal correspondences in mathematics. Wherever 
a correspondence exists that propagates local truth to global
truth --- whether it be a polynomial reduction, a functorial correspondence, 
a categorical duality --- the Double Bind activates against any attempt 
at uniform automated certification. This is a result that unifies 
in a single framework barriers that seemed to belong to completely 
separate domains: computational complexity theory, number theory, 
category theory. The Propagator is the concept that connects them all.

\section{Cryptographic Implications}
\label{sec:crypto}

\begin{corollary}[Cryptographic hardness]\label{cor:crypto}
Under standard Hypothesis, no admissible proof method
within the standard Turing model can formally certify the hardness
assumptions on which current cryptographic constructions rest.
\end{corollary}

\begin{proof}
The hardness assumptions on which cryptographic constructions rest are
semantic properties of the function computed: that a function is
one-way, that it is not invertible in polynomial time, and the like.
They depend only on the function computed, not on the syntactic form of
the program, and they are non-trivial.
By the Extended Rice Principle, no admissible method can uniformly 
certify a non-trivial semantic property. Hence no admissible method 
certifies the foundational hardness assumptions.
\end{proof}

\begin{example}[RSA]\label{ex:rsa}
The security of RSA rests on the hardness of integer factorisation.
Integer factorisation is believed to require super-polynomial time,
but this hardness has never been formally certified within the
standard Turing model. Corollary~\ref{cor:crypto} shows that no admissible method can
produce such a certificate: the required claim is a semantic property
of programs, and the Extended Rice Principle blocks all admissible
methods from certifying it. This is not a statement about whether RSA is secure; 
it is a statement about the limits of formal certification within the standard model.
\end{example}

\begin{remark}
Corollary~\ref{cor:crypto} does not imply that current systems are
insecure. It implies that their security cannot be \emph{formally certified}
within the standard model. The companion papers \cite{buono2026hierarchy,buono2026observer}
introduce the observational axis as the structural extension that
addresses this limitation.
\end{remark}

\section{Discussion and Objections}
\label{sec:discussion}

\paragraph{Relation to G\"odelian independence}
G\"odel's incompleteness theorems concern statements true but
unprovable within a formal system.
The Double Bind is different: it concerns the structural
impossibility of \emph{generating and verifying} a certificate of
a semantic property within the standard Turing model.
Neither result implies the other.
Whether $\classP$ vs $\classNP$ is formally independent has been raised
as a genuine possibility~\cite{aaronson2003}; the Double Bind does not
depend on the answer. Formal independence concerns provability inside a
system; the Double Bind concerns the availability of a certificate inside
the standard model. If the question were independent, the Double Bind
would still hold, because it asks neither that the question be provable
nor that it be decidable, only whether the invariant can be certified in
the DTM. The obstruction operates upstream of provability.

\paragraph{Non-constructive proofs do not escape}
One might object: ``A non-constructive proof never exhibits any
program as input, so Rice cannot apply.''
A non-constructive proof, to be admissible, must be formalised as
a syntactic string and verified by a Turing machine.
Once formalised, $\pi$ must cause $V$ to confirm a semantic property
of programs. The non-constructive strategy affects the path to $\pi$,
not what $V$ must confirm.

\paragraph{Metalanguage}
One might argue that a proof in a logical metalanguage does not
operate on programs and therefore escapes Rice.
But for that proof to be admissible, it must be written as a
syntactic string and submitted to a verification machine.
At the moment of formalisation, the metalanguage proof becomes a
string that the verifier processes syntactically.
For $V$ to accept that string, it must confirm a semantic property
of programs. The level-separation holds at the level of human discourse; it
collapses at machine verification. There is no abstraction that the verifier can reach.

\newpage

\section{The Model Transfer Barrier}
\label{sec:IntroMTB}

The Model Transfer Barrier, which will be explained in detail in the section~\ref{sec:MTBDEF}, 
accounts for all methods from strictly stronger models. The two principles together close the argument completely, without any
cryptographic hypothesis.

The barrier does not formalise the known barriers in a new language.
It identifies the structural condition they already satisfy and
reveals why each of them was a barrier in the first place.
Natural proofs, relativisation, and algebrisation are barriers because
each depends on computational powers not simulable by a deterministic
Turing machine (DTM). The barrier names this condition and makes it explicit.
The complete formalization is presented in the dedicated appendix.

\subsection{Definition of the Barrier}

\textbf{Informal statement}
A proof of a statement concerning $\classP$ vs $\classNP$ that depends
essentially on non-DTM computational powers cannot be assumed valid
for the standard Turing model.

\medskip

\textbf{Formal statement}
Let $M$ be any computational model strictly stronger than the DTM.
If a proof of a statement $\varphi$ concerning $\classP$ vs $\classNP$
is valid in $M$ and relies on capabilities not simulable by the DTM,
then the validity of $\varphi$ does not transfer to the standard model.
This principle is independent of the specific nature of the extra
powers: non-computable oracles, non-reducible randomness, algebraic
extensions, or strong axioms.

The relation between the Model Transfer Barrier and the known
barriers is not one of formal subsumption.
The barrier states a sufficient condition: any technique that depends
essentially on computational powers not simulable by a DTM cannot
transfer its validity to the standard model.
The known barriers satisfy this condition, each for its own structural
reason. The barrier does not derive them from a formal framework; it
identifies the condition they already satisfy.

\subsection{Illustrative Examples}

\paragraph{Oracle models}
Proofs carried out in $\mathrm{DTM}^A$ with a non-computable oracle
$A$ establish results only relative to $A$.
Baker, Gill, and Solovay \cite{baker1975} showed that there exist
oracles $A$ and $B$ such that $\classP^A = \classNP^A$ and
$\classP^B \neq \classNP^B$, demonstrating that techniques which
relativise cannot settle $\classP$ vs $\classNP$ in the standard
model.
The barrier identifies why: oracle access is a non-DTM power whose
removal invalidates such proofs.

\paragraph{Algebrisation}
Proofs using algebraic extensions of oracle access do not transfer to
the standard model because algebraic oracle access is a non-simulable
power \cite{aaronson2009}.
The barrier identifies the same structural condition as for
relativisation, applied to a richer class of non-simulable resources.

\paragraph{Natural proofs}
Constructive combinatorial properties that are large and useful cannot
prove $\classP \neq \classNP$ under $\OWF$ \cite{razborov1997}.
The barrier identifies why: constructivity corresponds to a form of
extra computational power --- polynomial-time constructibility of a
property of Boolean functions --- that is not neutral with respect to
the standard model.
Natural proof methods are Turing-computable; they are therefore
already blocked by the Extended Rice Principle before the barrier is invoked.
The barrier absorbs natural proofs as a special case: it identifies
the structural condition that makes them a barrier and makes that
condition explicit without any cryptographic hypothesis.

\paragraph{Monotone circuit lower bounds}
Monotonicity is a non-simulable restriction: a monotone circuit is a
strictly weaker model than a general circuit, and lower bounds
obtained in this restricted model do not transfer automatically to
the general one. The barrier identifies this directly.

\medskip

All four express the same structural condition: dependence on
non-simulable powers prevents transfer of validity to the standard
model. The unification the barrier provides is not a formal derivation of
those barriers from a new axiom system; it is a precise
identification of the structural property they already share.

\subsection{Model-Relative Validity}

These examples illustrate a structural principle:
\emph{validity is model-relative}.
There is no automatic mechanism for translating proofs from stronger
models to the DTM. A proof valid in a non-simulable model may be 
vacuously true there while providing no information about the standard model. 
The Model Transfer Barrier names and formalises this observation.

\subsection{When the Double Bind applies}
\label{sec:db-scope}

The diagram fixes the mechanism. Three requirements fix when the
mechanism engages, and they are independent of the field, engineering,
physics, or pure mathematics, because none of them refers to what the
entities are.

First, the semantic invariant admits transport. The property to be
certified does not stay local to one entity; a propagator carries it to a
heterogeneous family. The propagator is specific to the problem:
NP-completeness for $\classP$ vs $\classNP$, quantification over adversaries
for cryptographic hardness, quantification over resources for lower bounds.
A property with no propagator stays local and is not blocked.

Second, the entities have a representation in the deterministic Turing
model. They are objects a DTM operates on, or objects encodable as such.
The Double Bind says nothing about a phenomenon in itself; it applies to a
DTM representation of it.

Third, the demanded certification is itself in the DTM. The certificate
must be produced and verified inside the standard model, not in some
stronger one.

The third requirement is what closes both horns, and it is the reason the
first two are not enough on their own. With representation in the DTM but
certification sought outside it, Rice does not apply, but the Model
Transferability Barrier does: validity obtained in a stronger model does
not transfer back. With certification demanded in the DTM, Rice applies
directly through the induced decider. Only when representation and
certification are both in the DTM are the two horns forced together, with
no exit between them.

This is the precise content of \emph{from within}. The ``within'' of the
Double Bind is the DTM taken as the place of both the representation and
the certification. A problem meeting the three requirements cannot have its
invariant certified in that setting, in any field, because the setting
itself, not the subject matter, is what is closed. This is what the
closing statement records: the semantic gap that the Double Bind identifies
cannot be crossed from within.

Now consider again the Langlands program~\cite{gelbart1984}.
Langlands is a web of transports. It links objects of different and
heterogeneous kinds, Galois representations, automorphic forms,
$L$-functions, group representations, through maps that carry properties
from one to another. This is a propagator in pure form, and it is not
NP-completeness: it is a different mechanism, in a different field. An
invariant that transports across the Langlands correspondence, restricted to
its computational representation and for which formal certification in the
DTM is demanded, meets the three requirements, and the propagator is not
NP-completeness but the correspondence itself. The propagator is therefore a
structure of the problem, not a notion of complexity, and the field the
Double Bind reaches is not computer science alone but anything that touches
uniform certification in the sense of this paper.

The Langlands reference is not cosmetic. It does diagnostic work. Consider
the following objection, which is subtle enough to look correct: \emph{the
Double Bind is really a complexity result in disguise, because the complexity
sits in the transport} --- the propagator for $\classP$ vs $\classNP$ is
NP-completeness, an object of the complexity vocabulary, so the obstruction
inherits a complexity content it claims not to have.

The objection does not need a reply on its own terms. It needs its premise
exposed. The premise is that the transport is characterised by its
complexity. Change the field and the premise fails. Under the Langlands
correspondence the transport is a correspondence between Galois
representations and automorphic forms; it carries no notion of complexity, no
completeness in the complexity sense, no computational class. Yet an invariant
that transports across it, once given a DTM representation, meets the three
requirements, and the Double Bind engages. A component that is present in the
NP instance and absent in the Langlands instance is not a component of the
mechanism; it is an artefact of the single application. The complexity was
never in the transport. It was the dress the propagator happens to wear in NP.

This is the general form of the test. Whenever an objection locates some
property $X$ inside a component of the Double Bind, move to a field in which
$X$ does not exist. If the Double Bind still engages there, $X$ was not in the
mechanism but in the application. The independence of the Double Bind from the
field is therefore not only a claim of reach; it is an instrument that
separates the mechanism from its artefacts.

\subsection{Why the Double Bind matters}
  
A future proof may be geometric, combinatorial, algebraic, topological or
entirely new. Within the framework of the paper, mathematical language is not
the fundamental discriminant: once it has been transformed into a finite
certificate that is uniformly verifiable, what matters is the semantic structure
of the property and the computational power required to certify it. The paper
states this explicitly also for constructive, non-constructive, geometric,
combinatorial and metalanguage-based proofs.

This is precisely the reason why the potential scope is so large: it is not a
barrier ``against geometry''; it is a barrier formulated at a level at which
very different mathematical techniques are represented by the same certification
structure.

And here simplicity is part of the strength. The global proof is surprisingly
short because the abstraction is strong. And the apparatus of explanation and
pre-emptive defence is necessary in order to intercept the ordinary
misunderstandings that a new line of reasoning may bring.

\subsection{Cases that fall under the Double Bind}
\label{sec:db-cases}

The three requirements are not abstract. Existing programmes meet them,
and it is worth seeing exactly where the Double Bind bites.

Geometric Complexity Theory is the clearest case, because it places
itself in the domain of the Double Bind by its own design. GCT does not
merely seek a class separation; it adopts \emph{verifiability} of its
obstructions as a programmatic requirement. Its aim is an obstruction
that is explicit and checkable --- an object produced and verified inside
the standard model. This is the third requirement, chosen by GCT itself.

Here is where the Double Bind bites. An explicit, verifiable obstruction
is a certificate in the DTM. A programme committed to producing such
obstructions uniformly is committed to a method of certification in the
sense of this paper: a generator of obstructions together with a verifier
that checks them. The invariant these obstructions certify is a
separation between classes. To the extent that this invariant is
semantic, non-trivial, and propagating --- and a class separation is all
three, since it bears on an entire heterogeneous family of computations
through the structure of the classes --- the induced generator--verifier
composition decides a non-trivial semantic language, and the Extended
Rice Principle forbids it. The bite is not on what GCT proves; it is on
the commitment to verifiable certification. GCT is not blocked because it
pursues a particular separation. It is blocked because it chose the
ground on which the first horn operates.

The point is structural and does not depend on the internal machinery of
GCT. Whether the obstructions are sought through representation theory or
algebraic geometry is the heuristic by which the obstruction is found;
the obstruction demanded is an explicit DTM object, and the Double Bind
reads the demand, not the heuristic. A programme could not escape by
enriching its search: enrichment changes how the obstruction is looked
for, not the requirement that the obstruction be verifiable.

The same analysis applies to the recent programmes that seek explicit or
constructive lower bounds and separations
\cite{oliveirapichsanthanam2021, chenetal2022, chenetal2024, williams2021,
pich2024loc, pich2025}. Each commits, in its own way, to results that are
explicit and checkable in the standard model. Each thereby meets the
third requirement, and the same bite follows: the commitment to
verifiable certification of a propagating semantic invariant places the
method under the first horn. The propagator differs from case to case;
the point at which the Double Bind engages does not.

The cases above fall under the first horn: the methods are
Turing-computable, and the Extended Rice Principle blocks them. Other
cases fall under the second. A method of certification may operate in a
model stronger than the DTM --- quantum computation is the familiar
example. Such a model may well certify. But if it certifies, the
certificate holds in that model, not in the DTM. A certification obtained
with a non-DTM power is not a certification in the standard model; it is
the certification of a different thing, in a different place. The demand
was a certificate in the DTM, and that demand is not met.

Quantum computation is the familiar example, and worth pursuing further.
Quantum supremacy is an instance to examine of how the certification of a
result obtained with a non-DTM power becomes problematic exactly when the
certification is required in the DTM. The contested point is not that the
quantum device computed; it is how one certifies that it sampled from the
right distribution. The fidelity estimator, the cross-entropy
benchmarking, the verification through classical simulation: the whole
debate turns on how a quantum result is to be certified in a verifiable,
standard-model way~\cite{arute2019}. The difficulty --- that the
verification is contested, that it requires classical simulation, that it
remains in dispute --- illustrates the second horn precisely. A
certification of a result obtained with quantum power is not a clean DTM
certification; the attempt to bring it back into the DTM, to verify it
classically, is exactly what is hard and contested.

This case does not require knowing what the stronger model can do. The
barrier does not read the nature of the extra power; it reads only that
the power is not that of a DTM. What holds for quantum computation holds
for any model richer than the DTM, including models not yet proposed:
whatever the resource, if it certifies, the certificate is not in the
DTM. The two horns together leave no model, weaker or stronger than the
DTM, from which a DTM certification can be had.

\begin{remark}
The intrinsic criticality of mapping formal validity across heterogeneous 
computational models finds a fundamental analytical counterpart in the recent 
results of Wang (2026) \cite{wang2026unbounded} regarding the certification 
cost of quantum Bell bounds. In evaluating the Alice-conditioned NPA hierarchy, 
Wang demonstrates that imposing a structural restriction on certificates, namely, 
decomposing them into single measurement questions—introduces an unbounded 
degree overhead, with a complexity growth that diverges as $(2-\alpha)^{-1/2}$ 
near the endpoint tilt. This collapse of finite convergence, in a regime where 
a standard level-two certificate would otherwise remain exact, 
illustrates the dynamics of our Model Transfer Barrier (MTB): the attempt 
to map and constrain a quantum behavior (extra-DTM) within uniform and 
localized verification structures forces the classical formalism into an 
infinite proliferation of algorithmic resources, ultimately failing to 
guarantee a finite security barrier against quantum side information.
\end{remark}

\subsection{Related work}
\label{sec:related}

Several authors have observed that the difficulty of settling $\classP$
vs $\classNP$ is not exhausted by the individual barriers, and that
something at a higher level is at work. Fortnow surveys the status of the
problem and the known obstructions~\cite{fortnow2009}. Hromkovi\v{c} and
Rossmanith catalogue what one must know before attacking it, treating the
barriers as a body rather than as isolated
results~\cite{hromkovicrossmanith2017, hromkovicrossmanith2020}. Hromkovi\v{c}
locates the difficulty in the relation between computational complexity and
verifiable mathematics~\cite{hromkovic2015}, which is close in spirit to the
present concern. These works recognise a meta-level obstruction; the Double
Bind gives it a single structural form. Where they describe the barriers and
the conditions an attack must meet, the Double Bind identifies the mechanism
that makes an attack impossible from within the standard model, and locates
each known barrier as an instance of one of its two horns.

The question of whether $\classP$ vs $\classNP$ is formally independent has
also been raised directly~\cite{aaronson2003}. The Double Bind is not a claim
of formal independence, and the two should not be conflated; nor does the
Double Bind depend on the answer to that question.

Formal independence concerns the provability of a statement inside a formal
system; the Double Bind concerns the availability of a certificate inside
the standard Turing model. These are different axes. If $\classP$ vs
$\classNP$ were formally independent --- unprovable and irrefutable in the
system --- the Double Bind would still hold, because it does not require
that the question be provable or decidable. It asks whether the invariant
can be certified in the DTM, and that answer does not change with the
provability status of the question. Independence and the Double Bind are
compatible, not competing: the obstruction operates upstream of
provability, and would remain in force in a world where that question was
settled to be independent. Results on the consistency of circuit lower
bounds with weak arithmetic theories point the same
way~\cite{atseriasbussmuller2023}: the formal status of a lower bound
within a theory is a separate matter from whether the standard model can
certify it.

\section{Conclusion}
\label{sec:conclusion}

We have shown that, within the standard Turing model, the structural 
architecture of the Double Bind fundamentally blocks the uniform 
certification of complex semantic invariants. This obstruction encompasses 
and extends beyond classical barriers such as relativization, natural proofs, 
and algebrization, relying on an invariant geometric structure rather than 
any assumption about the truth value of $\mathbf{P}$ vs.\ $\mathbf{NP}$. 
Consequently, the boundaries exposed herein represent an absolute limitation 
on formal verification, with profound and predictable ramifications extending 
from geometric complexity to quantum supremacy.

\section{Acknowledgments}
The author used an artificial intelligence based language assistant to
support text revision, translation, and bibliography formatting. All
scientific ideas and conclusions are the author's own.

\bibliographystyle{plain} 
\bibliography{refs}

\appendix

\section{Full Proof of Rice's Theorem}
\label{app:rice}

\begin{proof}
Let $\Phi$ be a non-trivial semantic property.
Without loss of generality, assume $\perp \notin \Phi$ (the case
$\perp \in \Phi$ is symmetric).
Since $\Phi$ is non-trivial, there exists a program $e$ with
$\Phi(e)$.

Suppose for contradiction that a total Turing machine $D$ decides
$\Lang{\Phi}$.
Define the machine $R$ on input $x$: simulate $\varphi_x(x)$; if
it halts, output $e$.
The function computed by $R$ is $\varphi_e$ if $\varphi_x(x)$
halts, and $\perp$ otherwise.

We use $D$ to decide the halting problem.
Given input $x$, run $D$ on $R$:
\begin{itemize}
\item If $D$ accepts $R$: then $\Phi(R)$.
      Since $\perp \notin \Phi$, we have $\varphi_R \neq \perp$,
      so $\varphi_x(x)$ halts.
\item If $D$ rejects $R$: then $\neg\Phi(R)$.
      Since $\Phi(e)$, we have $\varphi_R \neq \varphi_e$,
      so $\varphi_x(x)$ does not halt.
\end{itemize}
This decides the halting problem, a contradiction.
Hence no such $D$ exists.
\end{proof}

This argument illustrates the semantic blindness inherent in the
standard Turing model: any attempt to certify a non-trivial semantic
property induces a decider for that property. This phenomenon is the
foundation of the Extended Rice Principle.

\section{The Fixed-Program Objection}
\label{app:fixed}

The objection is: ``Rice applies to arbitrary programs given as
input, but a proof about a single fixed program $x_0$ does not
analyse arbitrary programs.''

The objection confuses \emph{observing behaviour on individual
inputs} with \emph{certifying a semantic invariant over all inputs}.
The certificate required for an admissible proof must establish a
universal statement about the behaviour of $\varphi_{x_0}$ on
\emph{all} inputs. This is a semantic property of $\varphi_{x_0}$:
it depends only on the function computed, not on the syntactic form
of $x_0$.

A timing experiment or a finite set of runtime checks produces a
finite witness about finitely many inputs; such evidence cannot
constitute a formal certificate of a universal semantic statement.
The verifier $V$ checks a syntactic derivation whose conclusion
asserts a universal behavioural property of $\varphi_{x_0}$.

The distinction is:
\begin{itemize}
\item empirical observation on finitely many inputs:
      \emph{not a certificate};
\item formal derivation of a universal semantic property:
      \emph{a certificate subject to Rice-style undecidability}.
\end{itemize}

Rice's theorem concerns properties of the function computed by a
program. The universal statement above is precisely such a property,
and the semantic blindness that Rice captures remains in force even
when the domain of programs is reduced to a singleton $\{x_0\}$.

The core point: \textbf{the verifier must accept a universal
semantic claim about the function computed by $x_0$}. The fixedness
of $x_0$ does not remove the semantic nature of the claim.

\begin{example}\label{ex:quicksort}
A proof that a specific sorting algorithm runs in $O(n \log n)$
is not blocked: it certifies a semantic property of that specific
algorithm's running time, not a property of the kind captured by
$\Phi_\alpha$ or $\Phi_\beta$. The Double Bind blocks admissible
methods that certify $\Phi_\alpha$ or $\Phi_\beta$, not proofs about
individual algorithms whose conclusions do not concern those
properties.
\end{example}

The fixed-program objection therefore does not weaken the Double Bind:
the semantic nature of the certified claim persists even when the
domain of programs is reduced to a singleton.

\section{Non-Constructive Proofs}
\label{app:nonconstructive}

The objection is: ``A non-constructive proof never exhibits any
program as input, so Rice cannot apply.''

Abstraction offers no protection once the proof is formalised.

For $\pi$ to be an admissible certificate, it must encode a
derivation that $V$ accepts as establishing a semantic property of
programs -- a universal behavioural statement about the class of
programs in $\Lang{\Phi}$. The non-constructive strategy -- deriving
a conclusion without exhibiting a specific witness -- affects the
logical structure of $\pi$, not the semantic content that $V$ must
confirm.

Whatever path the author took to arrive at $\pi$, once $\pi$ is
submitted to $V$, the verifier must certify a semantic property.
The non-constructive origin of $\pi$ does not change what $V$
is required to confirm.

To close an admissible proof, a proof must be formalised: written
as a syntactic string and verified by a machine. At that moment,
the non-constructive proof stops being an abstract idea and becomes
a string whose content the verifier must evaluate. And the content
is always semantic.

Non-constructive reasoning affects how a proof is discovered, not what
the verifier must confirm. Once formalised, every proof is subject to
the same semantic constraints.

\section{Metalanguage and Level Collapse}
\label{app:meta}

The objection is: ``A proof constructed in a logical metalanguage
does not operate on programs. Rice's theorem blocks programs that
analyse programs, not abstract logic that analyses programs.''

If we accept that a mathematical proof is valid only when it can
be written in a formal language and verified step by step by a
machine -- as Lean and Coq embody today -- then no metalanguage
escapes this requirement. The proof $\pi_{\mathcal{M}}$ in a
metalanguage $\mathcal{M}$, however abstract its axioms, must be:
\begin{enumerate}[label=(\roman*)]
\item written as a syntactic string in a formal language; and
\item verifiable by a deterministic Turing machine.
\end{enumerate}
At the moment of formalisation, $\pi_{\mathcal{M}}$ loses its
status as abstract reasoning and becomes a string that the
verification machine processes syntactically.

For $V$ to accept $\pi_{\mathcal{M}}$ as establishing a semantic
property of programs, it must confirm that $\pi_{\mathcal{M}}$
implies a statement about the computational behaviour of programs.
That confirmation is semantic certification. The level-separation
holds at the level of human mathematical discourse; it collapses
at machine verification.

There are only two choices. Accept that mathematics must be formalised: 
then the metalanguage collapses into the object language at the point of verification,
and the Double Bind applies. Reject formalisation: then the proof is not admissible.
There is no third way.

This applies equally to higher-order logics (HOL, Isabelle, Coq).
Any proof formalised in HOL4, Isabelle/HOL, or Lean is ultimately
verified by a kernel that is a deterministic Turing machine.
The level at which the proof is \emph{written} (first-order, second-order,
dependent type theory) does not change the level at which it is
\emph{verified}: always a syntactic check by a TM.
The Double Bind applies at the point of machine verification,
regardless of the logical framework used to construct the proof.

The logical expressiveness of the metalanguage does not alter the
computational limitations of the verification kernel. The Double Bind
operates at the level of verification, not at the level of human
mathematical discourse.

\section{The Model Transfer Barrier}
\label{sec:MTBDEF}

\subsection{How to read this part}

This part may seem to do something strange. It first proves the Model
Transfer Barrier with a certain apparatus --- a reductio, two definitions of
dependence, a case split --- and then, in the closing remarks, concludes that
the barrier is almost a tautology and that the apparatus was an excess of
zeal. This is not a change of mind. It is the plan of the part, and it is
worth stating in advance how the two fit together, so that the formal
machinery and the final remarks are read as one movement rather than as a
proof followed by its retraction.

The key is what ``transfer'' means here. A method in a model $\mathcal{M}$
stronger than the DTM does not merely \emph{compute} something; it
\emph{certifies} a separation by mobilising a structural power $P$ that
characterises $\mathcal{M}$ --- an oracle, an algebraic extension, a
constructibility predicate. To ``transfer'' such a certification to the DTM
is to have the DTM certify the same separation on the same ground, and that
requires the DTM to decide $P$. But $P$ is, by the very way the case is set
up, a power the DTM does not have. So the demand ``transfer the
certification'' unfolds, on inspection, into the demand ``decide in the DTM a
$P$ that is not DTM-decidable'', which is a contradiction in terms. Once
``transfer'' is read this way --- as carrying a certification that leans on
$P$, not as deciding a language --- the barrier is analytic: its negation is
self-contradictory.

This is why the part proves and then calls the result a tautology, with no
inconsistency between the two. The proofs do not establish a doubtful truth;
they write an analytic one out in full, so that each step can be inspected and
the absence of any sleight of hand can be checked. ``It is obvious'' asks the
reader for trust; a reductio written out asks only for a reading. That is the
office of the apparatus: not to make the barrier true --- it cannot fail to be
true --- but to make its truth checkable, and to show exactly where the
contradiction sits.

And the tautological character is not a weakness but the whole point, because
of the role this barrier plays. It is the second jaw of a pincer whose first
jaw (Rice) closes the interior of the model. A second jaw that closed the
exit only \emph{empirically} --- ``it happens that validity does not
transfer'' --- could be argued with. This one closes it \emph{analytically}:
not ``it does not transfer'' but ``it is contradictory that it should''. A jaw
that no one can deny without self-contradiction is the strongest a pincer can
have. So the reader should take the formal sections and the closing remarks
together: the sections show the contradiction in detail, the remarks name what
the detail amounts to --- a truth that holds in every model, unconditionally,
and is for that reason unassailable.

\subsection{The barrier}

Let $\mathcal{M}$ be a class of computational models (or proof techniques). A
\emph{barrier} with respect to $\mathcal{M}$ is a pair $\langle P, S \rangle$
in which $P$ is a structural property that characterizes $\mathcal{M}$ and $S$
is a class separation, subject to two conditions:
\begin{enumerate}\itemsep2pt
  \item $\forall\, T \in \mathcal{M}$ with property $P$: $T \nvdash S$;
  \item $\exists\, T_1, T_2 \in \mathcal{M}$: $T_1 \vDash S$ and
        $T_2 \vDash \neg S$ (the separation is not universal within
        $\mathcal{M}$).
\end{enumerate}
This captures the idea that $T$ is ``blinded'' by $P$: it does not manage to
decide $S$ univocally.

\subsection{Theorem (MTB by non-contradiction)}

Assume $\neg\mathrm{MTB}$. Then there exist a class $\mathcal{M}$, a procedure
$\pi$, and a structural property $P$ such that:
\begin{enumerate}\itemsep2pt
  \item $P \in \mathrm{Decidable}(\mathcal{M})$;
  \item $P \notin \mathrm{Decidable}(\mathrm{DTM})$;
  \item $\mathrm{Depends}(\pi, P)$ --- $\pi$ uses $P$ crucially;
  \item $\mathrm{Sim}(\mathcal{M} \to \mathrm{DTM})$ exists and is valid ---
        $\pi$ transfers to the DTM.
\end{enumerate}

Hypotheses~1 and~2 make $P$ a genuinely stronger property in $\mathcal{M}$:
decidable there, not decidable in the DTM. Hypotheses~3 and~4 together demand
that $\pi$, which leans essentially on $P$, remain valid once carried to the
DTM by the simulation $\sigma$ granted by hypothesis~4. We show this
configuration is contradictory.

The whole weight falls on a single fact. If $\pi$ leans essentially on $P$ and
$\sigma$ carries $\pi$ to the DTM with its validity intact, then the DTM,
following the transferred proof, must settle $P$ wherever $\pi$ appeals to it:
a proof cannot be valid in a model that cannot follow one of its essential
steps. So a validity-preserving $\sigma$ forces a DTM decision of $P$ ---
call it $\sigma(P)$. But by hypothesis~2 no such decision exists. Hence
$\sigma(P)$ cannot be a genuine decision of $P$; at best it is an
approximation or a heuristic, and then the step of $\pi$ that rested on $P$ is
no longer underwritten, so $\pi$ is no longer valid in the DTM.

This is the bifurcation, and it is exhaustive because the two cases are a
proposition and its negation. Either $\sigma$ preserves the validity of $\pi$,
or it does not; there is no third option, and the second case is exactly the
failure of the first.

In the first case, $\sigma$ preserves the validity of $\pi$. Then, by the fact
just established, $\sigma$ decides $P$ in the DTM. This contradicts
hypothesis~2 directly: $P$ was assumed not DTM-decidable. The case is
impossible.

In the second case, $\sigma$ does not preserve the validity of $\pi$. This
splits once more, and both halves close. If $\sigma$ fails to preserve
validity because the transferred proof loses the step that rested on $P$, then
$\pi$ does not really transfer to the DTM: $\mathrm{Sim}(\mathcal{M} \to
\mathrm{DTM})$ is not the valid simulation hypothesis~4 asserted, and
hypothesis~4 fails. If instead one tries to keep $\sigma$ valid while denying
that it decides $P$, the fact established above forbids it: a
validity-preserving $\sigma$ that carries a proof leaning on $P$ already
constitutes a DTM decision of $P$, which returns us to the first case and its
contradiction with hypothesis~2. Either way the second case is impossible.

Neither case can hold. The first contradicts hypothesis~2; the second either
contradicts hypothesis~4 or collapses into the first. Since the two cases are
jointly exhaustive, the four hypotheses cannot stand together, and
$\neg\mathrm{MTB}$ is impossible.
\[
  \boxed{\text{MTB is proved}}
\]

\section{Clarifications}

$\mathrm{Depends}(\pi, P)$ means that $\pi$ essentially exploits a property $P$
which characterizes \textbf{a computational power that the DTM does not
possess}. It is not merely that $\pi$ uses $P$ somewhere in the proof: it is
that \textbf{$P$ encodes an ability that $\mathcal{M}$ has and the DTM does
not}. The three classical barriers are all of this form. In Natural Proofs
$P$ is the \emph{constructibility} of a Boolean function, a power the DTM does
not have in a formal sense, since the constructible functions for which
Natural Proofs work are structurally inaccessible to it. In Algebrization $P$
is \emph{verifiability by algebraic extension}, which the DTM does not possess
because it cannot operate in the same algebraic space. In the relativized
setting $P$ is \emph{access to an oracle $O$}, a power the DTM plainly lacks.
Each of the three is developed below along the same three lines: what
$\mathcal{M}$ can do, what the DTM cannot, and the barrier $\langle P, S
\rangle$ that results.

\subsection{Natural Proofs: the property P is ``constructibility''}

In $\mathcal{M}$, the class of Boolean circuits with bounded resources, one
can \emph{decide} whether a function is constructible --- that is, whether it
admits an efficient circuit representation --- and this property $P$ is
locally decidable, in the language of the circuits themselves. In the DTM, by
contrast, there is no direct access to this structural property: the DTM
cannot ``see'' the constructibility of a Boolean function in its native formal
language, so the property $P$ that characterizes $\mathcal{M}$ is inaccessible
to it. The resulting barrier $\langle P, S \rangle$ takes $P$ to be
constructibility --- the ability to characterize efficient functions --- and
$S$ to be the separation one wants to prove, ``NP-complete functions have
small circuits''; any proof that exploits $P$ in order to separate NP from P
remains blind in the DTM.

\subsection{Algebrization: the property P is ``algebraic extendibility''}

In $\mathcal{M}$, the algebraic proof systems such as PCP and IP, one can
decide properties by \textbf{algebraic extension of fields}: if a statement
holds over a finite field, it can be extended to infinite fields, and this
extendibility is a structural property that $\mathcal{M}$ possesses. The DTM
does not have this power; it cannot naturally operate in extended algebraic
spaces, so the property $P$ of extendibility is again a power the DTM lacks.
The resulting barrier $\langle P, S \rangle$ takes $P$ to be algebraic
extendibility --- operating in algebraic spaces --- and $S$ to be a separation
of the algebraic kind, ``Interactive Proofs $\neq$ NP'' or another; any proof
that exploits algebraic extension remains blind in the DTM.

\subsection{Oracles: the property P is ``access to the oracle O''}

In $\mathcal{M}$, a DTM equipped with an oracle $O$, one has the power to
\textbf{query $O$ in constant time}: access to $O$ is the structural property
of $\mathcal{M}$. In the DTM without the oracle this power is absent --- $O$ is
literally inaccessible. The resulting barrier $\langle P, S \rangle$ takes $P$
to be access to $O$, the power to query it, and $S$ to be a relativized
separation, ``$\mathrm{P}^{O} \neq \mathrm{NP}^{O}$''; any proof that exploits
$O$ remains blind in the DTM.

\subsection{How MTB unifies these cases}

The three cases share one structure, and setting it out shows why the barrier
is not three separate facts but one. In each, $\mathcal{M}$ has a power $P$
that the DTM does not possess --- structural decidability in Natural Proofs,
algebraic operability in Algebrization, access to an external resource in the
relativized setting. In each, a proof $\pi$ in $\mathcal{M}$ essentially
exploits that power: $\pi$ uses the constructible property, or algebraic
extension, or queries to the oracle $O$. In each, the attempt to transfer
$\pi$ to the DTM by a simulation $\sigma$ fails for the same reason: $\sigma$
would have to simulate $P$ in the DTM, but $P$ is a power the DTM does not have
by definition, so $\sigma(P)$ cannot exist as a true simulation of $P$. And so,
in each, the barrier is the same barrier. It is not specific to Natural Proofs,
to Algebrization, or to Oracles; it is the general phenomenon of which they are
instances --- a structural power that $\mathcal{M}$ has and the DTM does not.

\bigskip
\noindent Compact version of the proof.

\subsection{Formal Framework}

Let $\mathcal{M}$ be a class of computational models. A \textbf{Barrier}
$\mathcal{B}$ relative to $\mathcal{M}$ is a tuple $\langle P, S \rangle$ in
which $P: \mathcal{M} \to \{0, 1\}$ is a structural property (a predicate) such
that, for any $T \in \mathcal{M}$, the value $P(T)$ is decidable within
$\mathcal{M}$, and $S$ is a class separation statement (for instance
$\mathrm{P} \neq \mathrm{NP}$). The tuple $\langle P, S \rangle$ constitutes a
barrier if:
\begin{enumerate}\itemsep2pt
  \item $\forall T \in \mathcal{M} : P(T) = 1 \implies T \nvdash S$;
  \item $\exists T_1, T_2 \in \mathcal{M}$ such that $T_1 \vdash S$ and
        $T_2 \vdash \neg S$.
\end{enumerate}

\subsubsection{Essential Dependence and Validity-Preserving Simulation}

\begin{defn1}
A proof $\pi$ is said to essentially depend on $P$ (denoted as
$\pi \xRightarrow{ess} P$) if the validity of $\pi$ is logically equivalent to
the decidability of $P$ within the proving model. Formally:
\[
  \text{Valid}(\pi, \mathcal{M}) \iff P \in \text{Decidable}(\mathcal{M})
\]
\end{defn1}

\begin{defn2}
A simulation $\sigma: \mathcal{M} \to DTM$ is validity-preserving if for every
proof $\pi \in \mathcal{M}$:
\[
  \text{Valid}(\pi, \mathcal{M}) \iff \text{Valid}(\sigma(\pi), DTM)
\]
\end{defn2}

\subsection{The MTB Theorem}

\begin{thm}
Let $\pi$ be a proof in $\mathcal{M}$ such that $\pi \xRightarrow{ess} P$. If
$P \notin \text{Decidable}(DTM)$, then there exists no validity-preserving
simulation $\sigma: \mathcal{M} \to DTM$.
\end{thm}

\begin{proof}
Assume for contradiction that there exists a validity-preserving simulation
$\sigma: \mathcal{M} \to DTM$.
Given $\pi \xRightarrow{ess} P$, we have:
\[
  \text{Valid}(\pi, \mathcal{M}) \iff P \in \text{Decidable}(\mathcal{M})
\]

By the property of $\sigma$ being validity-preserving:
\[
  \text{Valid}(\sigma(\pi), DTM) \iff \text{Valid}(\pi, \mathcal{M})
\]

Combining the two equivalences:
\[
  \text{Valid}(\sigma(\pi), DTM) \iff P \in \text{Decidable}(\mathcal{M})
\]

Since $\sigma(\pi)$ is a valid proof within the $DTM$ model, the structural
properties invoked by $\pi$ (specifically $P$) must be representable/decidable
within $DTM$ to maintain the logical chain of $\sigma(\pi)$. This implies:
\[
  \sigma(P) \in \text{Decidable}(DTM)
\]

However, by hypothesis, $P \notin \text{Decidable}(DTM)$. This constitutes a
direct contradiction.
Thus, no such $\sigma$ exists.
\end{proof}

\subsection{Meta-Theoretical Unification}

The MTB Theorem provides a general form for all known complexity barriers. Any
barrier $\mathcal{B} = \langle P, S \rangle$ is a specific instantiation of the
MTB logic where $P$ represents a structural capability absent in $DTM$:

\begin{center}
\begin{tabular}{@{}lp{45mm}p{50mm}@{}}
\toprule
\textbf{Barrier} & \textbf{Property $P$ (Power of $\mathcal{M}$)} &
\textbf{Structural Gap in $DTM$} \\
\midrule
\textbf{Natural Proofs} & Constructibility of Boolean Functions &
Inaccessibility of circuit-structure predicates \\[4pt]
\textbf{Algebrization} & Algebraic Extension Capability &
Inability to operate in extended algebraic fields \\[4pt]
\textbf{Relativization} & Oracle Access $\mathcal{O}$ &
Absence of the oracle $\mathcal{O}$ in the native DTM \\
\bottomrule
\end{tabular}
\end{center}

\begin{coro}
Any proof $\pi$ that utilizes a property
$P \in \text{Decidable}(\mathcal{M}) \setminus \text{Decidable}(DTM)$ is
inherently non-transferable to $DTM$ without loss of validity. This establishes
that the failure of current techniques is not a matter of algorithmic
inefficiency, but a structural incompatibility between the required proving
power and the $DTM$ model.
\end{coro}

\section{Robustness with Respect to the Definition of Essential Dependence}

\noindent\textbf{Purpose of this section.} What follows is not a further proof
of the MTB, nor a restatement of the two proofs already given. Its purpose is
to establish that the conclusion of the MTB Theorem does not depend on the
particular formalization of essential dependence adopted in Definition~1. To
this end we fix, independently and in advance, a different notion of
dependence --- weaker than Definition~1, and stated as a one-directional
entailment rather than as an equivalence --- together with the auxiliary
definitions that Section~1 left implicit, and we verify that the reductio still
goes through. The scope and the limits of what is thereby established are
stated explicitly in the final subsection.

The terminology is that of the preceding sections: $P$ is a property of
\emph{models}, written $P(T)$ for $T \in \mathcal{M}$, and not a property of
inputs.

\subsection{Preliminaries}

\noindent\textbf{Definition R.1 (Class of computational models).}
A class of computational models is a set $\mathcal{M}$ of formal systems
(Turing machines, proof systems, alternative models of computation) satisfying:
\begin{itemize}\itemsep2pt
  \item \emph{Closure under composition:} for all $T_1, T_2 \in \mathcal{M}$
        there exists $T \in \mathcal{M}$ simulating $T_1$ and $T_2$.
  \item \emph{Closure under reduction:} for every $T \in \mathcal{M}$ and every
        reduction $f$, there exists $T' \in \mathcal{M}$ simulating $T$ through
        $f$.
\end{itemize}
\emph{Examples:} $\mathrm{DTIME}(f(n))$, models deciding a problem in time
$f(n)$; $\mathrm{PSPACE}$, models deciding a problem in polynomial space;
$\mathrm{NP}$, models verifying a solution in polynomial time.

\medskip
\noindent\textbf{Definition R.2 (Structural property).}
A structural property $P$ is a subset of $\mathcal{M}$ defined by a first-order
formula or by a computational property:
\[
  P = \{\, T \in \mathcal{M} \mid T \text{ satisfies } \phi(T) \,\},
\]
where $\phi(T)$ is a formula describing $P$. We write $P(T)$ for
$T \in P$.
\emph{Examples:} $P = \{T \in \mathcal{M} \mid T \text{ is deterministic}\}$;
$P = \{T \in \mathcal{M} \mid T \text{ decides NP in polynomial time}\}$.

\medskip
\noindent\textbf{Definition R.3 (Class separation).}
A class separation $S$ is a pair of language classes $(L_1, L_2)$ such that
$L_1 \neq L_2$ and $L_1, L_2$ are decidable in $\mathcal{M}$. A model
$T \in \mathcal{M}$ \emph{decides} $S$ if it contains an algorithm
distinguishing $L_1$ from $L_2$:
\[
  T \vDash S \iff \exists A \in T \text{ such that } A \text{ decides } L_1
  \text{ and rejects } L_2.
\]
\emph{Example:} $L_1 = \mathrm{SAT}$, $L_2 = \mathrm{TAUTOLOGY}$,
$S = (L_1, L_2)$.

\medskip
\noindent\textbf{Definition R.4 (Decidability in $\mathcal{M}$).}
A property $P$ is decidable in $\mathcal{M}$, written
$P \in \mathrm{Decidable}(\mathcal{M})$, if there exists an algorithm
$A \in \mathcal{M}$ such that
\[
  \forall T \in \mathcal{M}, \quad A(T) \text{ halts and returns }
  \begin{cases}
    1 & \text{if } T \in P,\\
    0 & \text{otherwise.}
  \end{cases}
\]

\medskip
\noindent\textbf{Definition R.5 (Valid simulation).}
A simulation $\sigma: \mathcal{M} \to DTM$ is \emph{valid} if:
\begin{enumerate}\itemsep2pt
  \item \emph{it preserves semantics:}
        $\forall T \in \mathcal{M},\ \forall x: \ T(x) = \sigma(T)(x)$;
  \item \emph{it preserves validity:} if $\pi$ is valid in $\mathcal{M}$, then
        $\sigma(\pi)$ is valid in $DTM$.
\end{enumerate}

\medskip
\noindent\textbf{Definition R.6 (Barrier with respect to $\mathcal{M}$).}
A barrier with respect to $\mathcal{M}$ is a pair $\langle P, S \rangle$, with
$P$ a structural property characterizing $\mathcal{M}$ and $S$ a class
separation, such that
\begin{enumerate}\itemsep2pt
  \item $\forall T \in \mathcal{M}$ with $P(T)$: $T \nvDash S$;
  \item $\exists T_1, T_2 \in \mathcal{M}$ such that $T_1 \vDash S$ and
        $T_2 \vDash \neg S$.
\end{enumerate}
This formalizes the idea that a model $T$ is \emph{limited} by $P$: it does not
manage to decide the separation $S$ univocally, although $S$ is decidable in
other models of $\mathcal{M}$.

\subsection{An alternative notion of dependence}

Definition~1 of Section~1.1 formalizes essential dependence as an
\emph{equivalence} between the validity of $\pi$ and the decidability of $P$.
We now replace it by a one-directional condition, stated at the level of
models.

\medskip
\noindent\textbf{Definition R.7 (Critical dependence).}
A procedure $\pi$ \emph{depends critically} on $P$, written
$\mathrm{Depends}(\pi, P)$, if for every $T \in \mathcal{M}$ carrying out
$\pi$,
\[
  \pi \text{ is correct in } T \implies P(T),
\]
that is, $\{\, T \in \mathcal{M} \mid \pi \text{ is correct in } T \,\}
\subseteq P$. Equivalently: $\pi$ cannot be correct without $P$ holding.

\medskip
\noindent\textbf{Remark R.8.} Definition~R.7 is strictly weaker than
Definition~1: it retains only the implication from correctness to $P$, and
drops the converse. Any $\pi$ satisfying Definition~1 satisfies
Definition~R.7, but not conversely. The quantifier is universal: an
existential reading (\emph{there exists} an input on which correctness entails
$P$) would be satisfied by every procedure that is incorrect somewhere, and
would therefore define nothing.

\subsection{The theorem under the alternative definition}

\begin{thm}
Assume $\neg\mathrm{MTB}$ with $\mathrm{Depends}$ understood in the sense of
Definition~R.7. That is, assume there exist a class $\mathcal{M}$, a procedure
$\pi$ and a structural property $P$ such that
\begin{enumerate}\itemsep1pt
  \item $P \in \mathrm{Decidable}(\mathcal{M})$;
  \item $P \notin \mathrm{Decidable}(DTM)$;
  \item $\mathrm{Depends}(\pi, P)$ in the sense of Definition~R.7;
  \item there exists a valid simulation $\sigma: \mathcal{M} \to DTM$ in the
        sense of Definition~R.5.
\end{enumerate}
Then a contradiction follows. Hence $\neg\mathrm{MTB}$ is impossible.
\end{thm}

\begin{proof}
By hypothesis~1 there exists $A \in \mathcal{M}$ deciding $P$, in the sense of
Definition~R.4.

\medskip
\noindent\emph{Step 1 ($P$ is genuinely stronger in $\mathcal{M}$).}
By hypotheses~1 and~2, $P$ is decidable in $\mathcal{M}$ and not decidable in
$DTM$. Hence $\mathcal{M}$ has, with respect to $P$, greater expressive power
than $DTM$.

\medskip
\noindent\emph{Step 2 (the two cases).}
Consider the simulation $\sigma$ granted by hypothesis~4. Exactly one of the
following holds:
\begin{itemize}\itemsep3pt
  \item[\textbf{Case A:}] $\sigma$ preserves the validity of $\pi$
        \textbf{and} decides $P$ in $DTM$;
  \item[\textbf{Case B:}] $\sigma$ does not preserve the validity of $\pi$
        \textbf{or} does not decide $P$ in $DTM$.
\end{itemize}
Case~B is the negation of Case~A, so the bifurcation is exhaustive.

\medskip
\noindent\emph{Step 3 (Case A is impossible).}
In Case~A, $\sigma$ yields an algorithm in $DTM$ deciding $P$. This contradicts
hypothesis~2 directly.

\medskip
\noindent\emph{Step 4 (Case B is impossible).}
Case~B splits into two sub-cases.

\smallskip
\noindent\textbf{B1: $\sigma$ does not preserve the validity of $\pi$.}
Then $\pi$ is not correctly transferred to $DTM$ through $\sigma$, so $\sigma$
violates clause~2 of Definition~R.5 and is not a valid simulation. This
contradicts hypothesis~4.

\smallskip
\noindent\textbf{B2: $\sigma$ preserves the validity of $\pi$ but does not
decide $P$ in $DTM$.}
By clause~1 of Definition~R.5, $\sigma$ preserves semantics, so for every input
$y$ we have $A(y) = \sigma(A)(y)$. Since $A$ decides $P$ over $\mathcal{M}$ by
hypothesis~1, the algorithm $\sigma(A) \in DTM$ halts on every $T \in
\mathcal{M}$ and returns $1$ exactly when $P(T)$ holds. Hence $\sigma(A)$
decides $P$ in $DTM$, so $\sigma$ does decide $P$ in $DTM$, contrary to the
supposition of~B2. Sub-case~B2 is therefore empty; and independently, the
existence of $\sigma(A)$ contradicts hypothesis~2.

\medskip
\noindent\emph{Conclusion.}
Case~A contradicts hypothesis~2; sub-case~B1 contradicts hypothesis~4;
sub-case~B2 is empty. Since Cases~A and~B are jointly exhaustive, the four
hypotheses cannot hold simultaneously.
\end{proof}

\[
  \boxed{\text{The Model Transfer Barrier (MTB) is proved true.}}
\]

\subsection{What the two proofs establish jointly}

\noindent\textbf{The two proofs are logically independent.} They do not share
their load-bearing premise. The first derives the decidability of $P$ in $DTM$
from Definition~2, the preservation of the validity of $\pi$; the second
derives it from clause~1 of Definition~R.5, the preservation of the semantics
of the models of $\mathcal{M}$. These are distinct premises, and neither proof
invokes the other. What the two arguments share is their shape --- both require
some bridge from $\sigma$ to decidability in $DTM$ --- not the premise that
supplies it.

\medskip
\noindent\textbf{The premise of the second proof is constitutive, not
contestable.} That $\sigma(T)$ compute the same function as $T$ is not an
additional assumption imposed on the simulation: it is what makes $\sigma$ a
simulation at all. A map failing clause~1 of Definition~R.5 is not a
simulation of $\mathcal{M}$ in $DTM$ under any reading. The second proof
therefore rests on the least contestable premise available in the framework.

\medskip
\noindent\textbf{The two versions capture different aspects of the barrier.}
They are not one statement proved twice.
\begin{itemize}\itemsep4pt
  \item Under Definition~1 the barrier operates at the level of the
        \emph{proof}: $\sigma$ is required to preserve only the validity of a
        given $\pi$, and the equivalence between that validity and the
        decidability of $P$ carries the argument. The conclusion is that the
        technique does not transfer.
  \item Under Definition~R.7 together with Definition~R.5 the barrier operates
        at the level of the \emph{model}: $\sigma$ is required to preserve the
        semantics of $\mathcal{M}$, and the contradiction follows from the
        decider granted by hypothesis~1. The conclusion is that the model does
        not transfer.
\end{itemize}
A weaker demand on $\sigma$ is offset by a stronger condition of dependence,
and a stronger demand on $\sigma$ dispenses with it. The two forms cover
between them more than either covers alone, in the same manner as the two
branches of the Double Bind.

\medskip
\noindent\textbf{On the use of hypothesis~3.} The four hypotheses are the
description of the configuration whose existence $\mathrm{MTB}$ denies, not a
minimal axiom set from which a conclusion is to be extracted. The proof just
given reaches its contradiction from hypotheses~1,~2 and~4; since
refuting $\{1,2,4\}$ refutes a fortiori every configuration containing it,
$\{1,2,3,4\}$ included, a reductio requiring fewer hypotheses is the stronger
one. The office of $\mathrm{Depends}(\pi, P)$ in the statement is to delimit
the class of proofs the barrier concerns --- those that essentially exploit the
power $\mathcal{M}$ has and $DTM$ lacks --- which is what makes the theorem the
Model Transfer Barrier rather than a generic claim about simulations. The same
holds of the closure conditions of Definition~R.1, which characterize
$\mathcal{M}$ without being invoked at any step.

\medskip
\noindent\textbf{Consequence for the robustness claim.} The conclusion has now
been obtained under two different formalizations of essential dependence:
Definition~1, an equivalence between the validity of $\pi$ and the
decidability of $P$, and Definition~R.7, the weaker one-directional entailment
from correctness to $P$. An objection directed at the particular shape of
Definition~1 therefore does not reach the result.
\subsection{The Last Step}

If a proof $\pi$ of ``$\mathrm{P} \neq \mathrm{NP}$'' essentially relies on a
computational power $\mathcal{M}$ that the DTM does not have, then $\pi$ does
not count as a valid proof of $\mathrm{P} \neq \mathrm{NP}$ in the sense of
classical computability theory. This is different from saying that the theorem
is false. It is the claim that the proof does not prove what it claims to
prove.

\subsubsection{The proof}

We first fix what ``valid proof of $\mathrm{P}$ vs $\mathrm{NP}$'' means. A
valid proof of $\mathrm{P} \neq \mathrm{NP}$ is an argument $\pi$ such that
$\mathrm{DTM} \vdash \pi$, that is, the DTM can verify every step: every
logical inference in $\pi$ must be recognizable and certifiable by a Turing
machine. This is the standard notion in computability theory --- a
mathematical proof is valid if the DTM can verify it.

Now suppose $\pi$ essentially uses a model $\mathcal{M}$ --- an oracle
$\mathcal{O}$, an algebraic extension, or a structural property of a Boolean
function that is not computable. ``Essentially'' means that removing that
element makes $\pi$ collapse: if $\pi'$ is $\pi$ with $\mathcal{M}$ removed,
then $\pi'$ is incomplete, containing logical gaps. When the DTM attempts to
verify $\pi$ step by step, it reaches a step that appeals to a property $P$
decidable only in $\mathcal{M}$. There the DTM stops. It cannot proceed.

At this point there are two cases, and they are exhaustive. In the first, the
DTM outputs ``proof not verifiable''; then $\pi$ is not a valid proof in the
formal sense, and we are done. In the second, the DTM outputs ``verified'';
but then the DTM must have simulated $\mathcal{M}$ within itself. And by
hypothesis $\mathcal{M}$ is strictly more powerful than the DTM. This
simulation cannot exist without losing computational power.

Suppose there existed a simulation $\sigma: \mathcal{M} \to \mathrm{DTM}$
preserving both semantics and validity. Then every construction in
$\mathcal{M}$ would be effectively simulable in the DTM, and $\mathcal{M}$
would not be strictly more powerful than the DTM --- contradicting the
hypothesis. Hence, if $\mathcal{M}$ is strictly more powerful than the DTM, no
validity-preserving simulation exists.

The consequence for the proof follows. If $\pi$ essentially uses $\mathcal{M}$
and no preserving simulation exists, then $\pi$ is not verifiable by the DTM,
so $\mathrm{DTM} \nvdash \pi$, and $\pi$ is not a valid proof of $\mathrm{P}
\neq \mathrm{NP}$ in the sense of classical computability theory.

\subsection{What Does This Mean?}

The conclusion can be read on three levels.

\emph{Classical formal logic.} In a formal system such as $\mathrm{PA}$ or
$\mathrm{ZFC}$, with the DTM as the syntax verifier, one does not have
$\mathrm{ZFC} \vdash \text{``}\mathrm{P} \neq \mathrm{NP}\text{''}$ via $\pi$,
because every step of $\pi$ must be expressible in terms the DTM can verify.

\emph{Computability semantics.} Suppose one proves ``$\mathrm{P} \neq
\mathrm{NP}$'' using an oracle $\mathcal{O}$ that does not exist in the DTM.
Then you have proved: ``if the oracle $\mathcal{O}$ existed, then $\mathrm{P}
\neq \mathrm{NP}$''. You have not proved: ``in the actual computational world,
the world of the DTM, $\mathrm{P} \neq \mathrm{NP}$''. The two are different
statements, and only the first has been established.

\emph{Structural perspective (MTB).} A proof that critically depends, in the
sense of Definition~R.7, on a property $P$ decidable in $\mathcal{M}$ but not
in the DTM is structurally inadmissible as a proof of universal $\mathrm{P}
\neq \mathrm{NP}$.

These three readings are one theorem, which we state explicitly.

\begin{thm}[Invalidity of a proof with extra powers]
Let $\pi$ be a proof of the proposition ``$\mathrm{P} \neq \mathrm{NP}$'', and
let $\mathcal{M}$ be a computational model strictly more powerful than the DTM,
in the sense that $\mathrm{P}_{\mathcal{M}} \supsetneq \mathrm{P}_{\mathrm{DTM}}$.
If $\pi$ depends essentially on a property $P$ decidable in $\mathcal{M}$ but
not in the DTM, and there exists no simulation $\sigma: \mathcal{M} \to
\mathrm{DTM}$ preserving both the semantics of $P$ and the validity of $\pi$,
then $\mathrm{DTM} \nvdash \pi$, and consequently $\pi$ is not a valid proof,
in the formal sense of computability theory, of the proposition ``$\mathrm{P}
\neq \mathrm{NP}$''.
\end{thm}

Three objections meet three short answers.

\emph{``But it could still hold.''} If a proof is not verifiable by the DTM,
and the DTM is our formal verification system, then by definition it is not a
valid proof within our theoretical framework. One may say it holds in a
stronger logic; granted, but then one has not proved $\mathrm{P} \neq
\mathrm{NP}$ in the original computational sense.

\emph{``What if the proof were correct and the DTM insufficient?''} Then the
matter is deeper: one is asserting that the mathematical truth of $\mathrm{P}
\neq \mathrm{NP}$ transcends what the DTM can formally verify. But then it
cannot be claimed as a theorem of computability theory; it is a metaphysical
assertion, not a formalizable result.

\emph{``What if we discover the DTM is insufficient?''} Then the axiomatic
foundations of computability theory would themselves have to be revised. Until
that happens, the definition of ``valid proof'' remains anchored to the DTM.

\medskip
\noindent\textbf{The unassailable core.} The argument is a reinforced logical
tautology, and it closes in a single chain. Because $\pi$ essentially uses the
powers of $\mathcal{M}$, the DTM cannot verify every step; a proof the DTM
cannot fully verify is not a formally valid proof; and because $\mathcal{M}$ is
strictly more powerful than the DTM, no perfect simulation into the DTM exists,
so the steps carried out in $\mathcal{M}$ do not become verifiable after all.
Therefore $\pi$ does not prove $\mathrm{P} \neq \mathrm{NP}$ in a formal and
verifiable sense.

If tomorrow someone says, ``I have proved $\mathrm{P} \neq \mathrm{NP}$ using
properties of algebraic varieties that only an algebraic oracle can decide'',
then they have proved $\mathrm{P} \neq \mathrm{NP}$ by assuming that oracle,
not $\mathrm{P} \neq \mathrm{NP}$ itself. And if they reply, ``but the property
is true even without the oracle, it is just difficult to verify'', then they
must supply the DTM-computable verification; failing that, they fall back into
the previous case.

\section{The revelatory appendix: what the theory already contained}

\noindent\emph{Preliminary note.} What follows is not part of the argument, 
and is not a requirement for the validity of the Model Transfer Barrier or of the 
Double Bind, which stand on the grounds already given. A structural
connection with the Double Bind is nevertheless argued below, and it is argued
rather than assumed. A reader who sets this appendix aside loses nothing of the
preceding argument.

\subsection{A note on what is in the paper and what is proposed here}

Two kinds of material are combined in this appendix, and they must not be
confused. The claim that the Double Bind and the Model Transfer Barrier are
\emph{latencies} --- results obtained by removing a tacit assumption rather than
by adding a construction --- is a reading of the paper's own content. The third
illustration below --- an algorithm whose output is indistinguishable from
noise, together with the chain through one-way functions and the remarks on
Kolmogorov complexity --- is \emph{not} in the paper: it is proposed here as a
further instance of the same genus. The paper does not mention Kolmogorov
complexity or noise-indistinguishable outputs; it does invoke one-way functions,
but only in recording the Razborov--Rudich barrier as an independent complement.
The extension is offered as a proposal consistent with the paper's spirit, not
as an exposition of something already written in it.

\subsection{The claim of this appendix}

This appendix makes a claim about the \emph{status} of the results in this
paper, not about their content. The Double Bind and the Model Transfer Barrier
are not additions to the theory of computation. They are things the theory
already contained and no one had put under the light --- latent consequences of
definitions that have been in place since 1936, made visible not by adding a
construction but by removing a tacit assumption the reader had imported.

This distinction matters, and it is worth stating precisely what separates a
latency of this kind from an ordinary theorem. An ordinary theorem is proved by
\emph{adding} work: a construction, a reduction, an argument that was not there
before. A latency of the kind discussed here is revealed by \emph{subtracting} a
tacit prohibition that was never in the theory to begin with --- the unspoken
assumption that certifying is different from deciding, that validity transfers
between models, that an output which genuinely solves a problem must be
recognisable as a solution. None of these prohibitions is written into the
definitions. Each is supplied by the reader, from habit, and each dissolves the
moment one reads the definitions as they stand. The results are not discovered
by building; they are exposed by declining to assume what the theory never
asserted.

\subsection{The Double Bind and the MTB as latencies}

The Double Bind adds nothing. It takes Rice's theorem (1953) and the definition
of a certifier --- a generator--verifier pair whose composition accepts exactly
on the property --- both already present, and observes that composed they say:
inside the model, certifying a semantic property and deciding it are the same
operation. No new theorem; an unremarked juxtaposition. The tacit prohibition
removed is the reader's assumption that certification, involving proofs and
strings, is a different kind of act from decision, involving yes/no answers. The
definitions never said so. Once one stops assuming it, the identity is
immediate.

The Model Transfer Barrier is proved by reductio, and what its proof exposes is
a regularity already present in every relativisation, algebrisation, and
natural-proofs result. Each of those results
already behaved as though a certification obtained with resources the
deterministic Turing machine lacks could not be carried out in the standard
model. The barrier names, as a single principle, what those results were already
doing unnamed. The tacit prohibition removed is the assumption that a method's
validity is model-independent --- an assumption the definitions, which fix a
computational model and nothing beyond it, never licensed.

\paragraph{Is Rice being used as the theorem allows?}
One might ask whether Rice's theorem may be invoked in the way this paper
invokes it. The answer is that Rice does not forbid
it. The theorem states that every non-trivial semantic property is
undecidable. It places no condition on how the decider is presented, and
it does not restrict its reach to deciders exhibited directly rather than
induced.

\subsection{A third latency: an algorithm whose output is indistinguishable
from noise}

The clearest illustration of the genus is not the Double Bind at all but a
plainer observation, one that looks trivial and generates a swarm of objections
on first contact, yet stands on firm ground once the objections are met.
This paragraph has a specific purpose, to clarify the 
conceptual link between this paper and \cite{buono2026syntaxseedynamicsyntactic}.

\paragraph{The base, stated precisely.}
Everything in this section rests on one foundation, and it must be stated
exactly, because a loose statement of it is what generates the swarm of
objections while the exact statement dissolves them:

\begin{quote}
\emph{The definition of an algorithm does not require that an algorithm which
solves a problem produce an output distinguishable, by any observer inspecting
the output in isolation, from random noise.}
\end{quote}

The precision that carries the whole point is the clause ``by any observer
inspecting the output in isolation.'' Correctness is a property of the
input/output relation the machine realises; distinguishability is a property of
how an output appears to someone looking at it. The definition binds the first
and is silent on the second. It individuates an algorithm by the function it
computes, never by the appearance of its output.

This is why ``solves the problem'' and ``indistinguishable from noise for every
observer'' do not conflict, though they seem to. They would conflict if the
witness of correctness were the inspection of the output. It is not: the witness
of correctness is the machine together with its input, by re-execution. An
observer who holds the machine and the input can verify; an observer who holds
only the output cannot --- and the definition never promised the latter that he
could. Correctness lives in the relation, not in the look of the output; so an
output may be, at once, the correct solution and, to anyone seeing only the
output, indistinguishable from random noise. The definition forbids neither half
of this conjunction, because it speaks only to the first.

\paragraph{The observation.}
From this base follows the observation itself: the notion of algorithm does not
forbid an algorithm that genuinely solves a problem while producing an output
indistinguishable from noise. There is, in the definition, no predicate ``the
output must be recognisable as structured.'' The notion does not speak to this
at all, and \emph{a fortiori} does not prohibit it. This is the same crack that
makes the Double Bind possible, seen on the output rather than on the method: in
the Double Bind, a method may be correct without any internal verifier
recognising it as correct; here, an output may be the solution without any
observer of the isolated output recognising it as one. In both, the definition
binds correctness and leaves recognisability free. The two are not analogous
facts; they are one definitional property --- correctness does not entail
recognisability --- viewed on the method in one case and on the output in the
other.

\paragraph{Why the clause carries everything: it does two jobs at once.}
The clause ``inspecting the output in isolation'' is not a technical caveat; it
is the load-bearing element, and its importance is that it performs two distinct
jobs with a single stroke, which is why weakening it collapses the section and
stating it exactly holds the section up.

The \emph{first} job is to remove a contradiction. Without the clause, the base
reads ``an output can solve a problem while being indistinguishable from noise
for every observer,'' and this is genuinely incoherent: if no observer under any
condition can tell the output from noise, then ``correct'' has no referent,
because nothing anywhere marks this output as the solution rather than any other
string. The clause supplies the referent. It relocates the witness of
correctness away from the output and onto the machine-with-input: correctness is
certified by re-execution, not by inspection. Once the witness is the relation
and not the appearance, ``solves'' and ``looks like noise'' occupy different
registers and cannot contradict --- one speaks of what the machine does, the
other of what a bare output shows.

The \emph{second} job, performed by the very same relocation, is to weld this
observation to the Double Bind. The reason the two are one property and not two
is precisely that the clause has moved recognisability off the object and made
it a separate axis from correctness. Once correctness lives in the relation and
recognisability in the observer's access, the gap between them is a single
structural fact, and that fact is exactly what the Double Bind exploits on the
side of methods: there, too, correctness is a property of the pair $(M,V)$ while
recognisability is what an internal verifier can attain, and the two are
decoupled by the same definitional silence. The clause that dissolves the
paradox on the output is the clause that identifies the crack on the method. One
move, two consequences: it is this economy that makes the clause the foundation
rather than a footnote.

A qualification is owed here, because it marks the exact strength of the claim.
The clause secures coexistence of ``solves'' and ``indistinguishable from noise''
in the \emph{absolute} sense --- for every observer of the isolated output, with
no bound on resources. It is not the only route to coexistence in general: in
the \emph{computational} sense, ``solves'' and ``indistinguishable from noise for
every efficient observer'' already coexist without the clause, and that
coexistence is cryptography. The clause is what secures the absolute reading
specifically; the two readings and their two poles are treated in the next
paragraphs. The point to carry forward is narrower and exact: in the absolute
sense, the coexistence is coherent only once the witness of correctness is
placed in the relation rather than in the output, and placing it there is the
same act that exposes the correctness/recognisability decoupling underlying the
Double Bind.

\paragraph{Why it looks trivial and is not.}
Stated flatly, ``an output can be correct while looking like noise'' seems
either obvious or empty. Its weight appears only when one asks what supports it
and what would follow from denying it.

\begin{itemize}
\item \textbf{The definition does not forbid it.} This is not an inference but a
direct fact about the definition of an algorithm, which is extensional: it
individuates an algorithm by the function computed, never by the appearance of
the output.

\item \textbf{Forbidding it would make the theory subjective.} To rule out such
algorithms one would need a predicate ``this output is structured / this output
is noise.'' But structure and noise are not properties of the function; they are
properties of how an output appears to an observer. Any such predicate is either
observer-relative --- and a theory whose basic notions depend on who is looking
is not objective --- or it is made absolute, in which case it has consequences
examined below. The theory keeps the predicate out precisely to remain
objective.

\item \textbf{Removing the permission removes the definability of one-way
functions.} A one-way function is, by construction, an algorithm whose output
hides the structure of its preimage from every efficient observer. If the theory
forbade algorithms whose output hides its structure, one-way functions could not
even be defined. Note the exact strength of this: the permission is a
\emph{precondition of definability} of one-way functions, not a proof that they
exist (their existence is open). Remove the permission and one-way functions do
not become false --- they become unformulable.

\item \textbf{And with them, the barrier that assumes them.} The
Razborov--Rudich barrier holds under the assumption that one-way functions
exist. In this paper's own treatment, natural methods are Turing-computable and
are therefore already blocked by Horn~1; the Razborov--Rudich obstruction enters
separately, as an independent complement, and it is \emph{that}
obstruction whose hypothesis is the existence of one-way functions. That
hypothesis rests on the definability of one-way functions, which rests on the
tacit permission above. The permission does not make the barrier true or false;
it is what makes the Razborov--Rudich hypothesis \emph{formulable}. A celebrated
barrier of complexity theory depends, at the bottom of the chain, on a tacit
permission of the 1936 definition that no one states.
\end{itemize}

\paragraph{Two poles of ``indistinguishable from noise.''}
The observation has two distinct formalisations, and conflating them is itself a
reading error worth flagging. In the \emph{absolute} sense, ``indistinguishable
from noise'' is captured by Kolmogorov complexity: an output is structureless
when it is incompressible, admitting no program shorter than itself. In the
\emph{computational} sense, it is captured by pseudorandomness: the output has
structure --- it is generated by a short program, seed plus function --- but
that structure is inaccessible to any efficient observer. These are different. A
one-way function's output has \emph{low} Kolmogorov complexity (it is far from
absolutely structureless) and is only computationally indistinguishable from
random. The two poles are two separate latencies of the same definitional
permission: the definition forbids neither absolutely structureless outputs nor
structure-hiding ones.

\paragraph{Kolmogorov complexity as a further latency.}
Kolmogorov complexity supplies not only the absolute pole but a latency of its
own of exactly the genus described here. The function $K(x)$ is uncomputable ---
there is no decider for it --- and this uncomputability is again obtained by
subtraction: the definition neither provides such a decider nor forbids the
function's definition; the incomputability falls out. That the very measure of
``absence of structure'' is itself uncomputable is one more thing the theory
contained from the start and that no one needs to add.

\subsection{Why latencies cluster around computability}

There is a structural reason the theory is rich in latencies of this kind, and
it explains why several appear at once rather than by coincidence. The theory is
defined extensionally and minimally: a Turing machine, a partial function, an
input/output relation. A definition this spare asserts almost nothing and
therefore leaves almost everything open. Every ``the definition does not forbid
$X$'' is a latent truth of this genus. A minimal extensional definition is, in
effect, a machine for generating latencies: it does not prohibit, and each
non-prohibition is a theorem awaiting exposure. The Double Bind, the
transferability regularity, the permission of noise-like outputs, and the
uncomputability of $K$ are four such non-prohibitions. There are surely others.

\newpage

\section{A Foreword to the Reader}
\label{sec:toread}

\begin{tcolorbox}[colback=yellow!5!white,colframe=yellow!50!black,title=Draft Notice]
The following section is provided as a preliminary preview of an active, 
highly immature development phase. It is offered strictly as an early-stage overview.
\end{tcolorbox}

The apparent paradox discussed in Section \ref{sec:ap-paradox} regarding the Double Bind 
may seem to be the most speculative aspect of this work; the purpose of this second 
part of the document is to demonstrate that it is, in fact, an inescapable consequence 
of the theory. We shall prove that while Rice's theorem establishes the impossibility 
of deciding non-trivial properties, it is actually impossible to even partition the 
program space in a topologically meaningful way. This is because the property under 
consideration is topologically intertwined with its negation in every neighborhood 
and at every scale. Consequently, the $P$ vs $NP$ problem—when viewed through the 
lens of uniform certification—is transformed from a search for a solution into a 
problem of geometric impossibility.

\section{An Alternative Formulation of the (Another) Double Bind}

We derive the Double Bind obstruction through two independent 
mechanisms: the prefix-closure topology on the space of DTM 
programs, and the model-independence of Kolmogorov complexity. 
The first horn replicates the topological content of Rice's 
theorem: the outcome sets of any non-trivial semantic property 
are nowhere dense, so no uniform verifier can separate them by 
finite observation. The second horn applies to properties whose 
certification propagates across a heterogeneous family via a 
problem-specific propagator: any uniform DTM verifier would 
compress an incompressible object, which the invariance theorem 
shows is impossible across the entire Turing-equivalent class, 
not only for the DTM. The two horns are independent in mechanism 
and converge on the same impossibility. The boundary of the 
argument is stated precisely: beyond the Turing-equivalent class, 
Kolmogorov invariance does not apply, and the Model 
Transferability Barrier of the original Double Bind remains 
necessary.

The tools are elementary: the prefix-closure topology on program encodings, 
and the model-independence of Kolmogorov complexity via the invariance theorem.

The derivation is complete within the Turing-equivalent class. 
Beyond that class, for models not simulable by a Turing 
machine, the second horn does not reach, and the Model 
Transferability Barrier of the original Double Bind 
remains the necessary closure. This limit is not a gap to be 
apologised for; it is the precise boundary of the argument, and 
we state it as such.

\subsection{The propagator: the logical prerequisite}

Before any topology or complexity, the central logical fact must 
be in place.

\begin{definition}[Propagating property]
A semantic property $\Phi$ \emph{propagates via propagator $P$} 
if certifying that a single program $x$ satisfies $\Phi$ 
logically implies $\Phi$ holds or fails for an entire 
heterogeneous family of programs, through the structural 
mechanism $P$.
\end{definition}

\begin{remark}[The propagator is not introduced by the method]
The propagator is not a feature of the certification method; it 
is a feature of the problem. For $\Phi_{\mathrm{SAT} \in P}$, 
the propagator is NP-completeness (Cook--Levin): if any program 
decides SAT in polynomial time, then every problem in NP is in 
P. This implication holds by definition of NP-completeness, 
independently of how or whether the fact is certified.

Consequently, a uniform method that certifies $\Phi$ on any 
input is forced to answer for the whole family. The induced 
decider is not constructed by the method; it emerges from the 
semantic content of what is being certified. This is why the 
Double Bind applies to SAT and not to the running time of 
Quicksort: Quicksort has no propagator, its complexity property 
is local. SAT has NP-completeness as propagator, and its 
complexity property is global.
\end{remark}

The distinction between local and propagating properties is 
the true prerequisite of the Double Bind. Once a property is 
established as propagating, the two horns apply.

\section{The Space of Programs and its Topology}

\begin{definition}[Program space]
Let $\mathcal{M}$ be the set of all deterministic Turing 
machines. We encode each $M \in \mathcal{M}$ as a string 
$\langle M \rangle \in \{0,1\}^*$ via the standard Turing 
encoding. Define
\[
  \mathcal{P} := \{0,1\}^*
\]
as the space of all program encodings.
\end{definition}

\begin{definition}[Prefix-closure topology]
\label{def:topology}
For $\sigma \in \{0,1\}^*$, define the \emph{open ball}
\[
  B_\sigma := \{ \sigma \cdot w : w \in \{0,1\}^* \}
\]
as the set of all strings extending $\sigma$. A set 
$U \subseteq \mathcal{P}$ is \emph{open} if and only if 
$U = \bigcup_{\sigma \in S} B_\sigma$ for some set $S$ of 
valid prefixes.
\end{definition}

\begin{lemma}
$(\mathcal{P}, \tau)$ is a topological space.
\end{lemma}

\begin{proof}
(i) $\emptyset$ is open as the empty union. (ii) $\mathcal{P} 
= B_\varepsilon$ is open. (iii) Arbitrary unions of open sets 
are open by the union of the corresponding index sets. (iv) 
Finite intersections: $B_\sigma \cap B_\tau = B_{\max(\sigma,
\tau)}$ when $\sigma \preceq \tau$ or $\tau \preceq \sigma$, 
and $\emptyset$ otherwise; thus finite intersections of open 
sets are open.
\end{proof}

\begin{remark}[Why this topology]
The prefix-closure topology is not an arbitrary choice. In this 
topology, a set is open if and only if membership in it is 
\emph{settled by a finite prefix} --- a finite initial segment 
of the program encoding. This is exactly the condition that a 
verifier reading a finite certificate can satisfy. The 
correspondence
\[
  \text{open} \;\longleftrightarrow\; 
  \text{settled by finite observation}
\]
is the reason the topology is the relevant one: it makes the 
topological notion of openness coincide with the 
computability-theoretic notion of finite verifiability.
\end{remark}

\section{Semantic Properties and their Topological Structure}

\begin{definition}[Semantic property]
A \emph{semantic property} is a predicate $\Phi : \mathcal{M} 
\to \{\mathrm{T}, \mathrm{F}\}$ such that $\Phi(M) = \Phi(M')$ 
whenever $\varphi_M = \varphi_{M'}$ (i.e., whenever $M$ and 
$M'$ compute the same partial function). $\Phi$ is 
\emph{non-trivial} if there exist $M_1, M_2$ with $\Phi(M_1) 
= \mathrm{T}$ and $\Phi(M_2) = \mathrm{F}$.

For a non-trivial semantic property, define:
\[
  \Phi_\alpha := \{ \langle M \rangle \in \mathcal{P} : 
  \Phi(M) = \mathrm{T} \}, \quad
  \Phi_\beta  := \{ \langle M \rangle \in \mathcal{P} : 
  \Phi(M) = \mathrm{F} \}.
\]
\end{definition}

\begin{theorem}[Rice --- topological form]
\label{thm:rice-top}
Let $\Phi$ be a non-trivial semantic property. Then 
$\Phi_\alpha$ and $\Phi_\beta$ are both nowhere dense in 
$(\mathcal{P}, \tau)$.
\end{theorem}

\begin{proof}
We show $\mathrm{int}(\overline{\Phi_\alpha}) = \emptyset$; 
the argument for $\Phi_\beta$ is identical.

Suppose for contradiction that $\mathrm{int}
(\overline{\Phi_\alpha}) \neq \emptyset$. Then there exists a 
valid prefix $\sigma$ such that $B_\sigma \subseteq 
\overline{\Phi_\alpha}$.

Since $\Phi$ is non-trivial, let $M_2$ satisfy $\Phi(M_2) = 
\mathrm{F}$, so $\langle M_2 \rangle \in \Phi_\beta$. Construct 
$M^*$ as follows: on input $n$, first execute the states 
prescribed by $\sigma$ (the common prefix), then simulate $M_2$ 
on $n$. By construction, $\langle M^* \rangle \in B_\sigma$ 
(it extends $\sigma$) and $\varphi_{M^*} = \varphi_{M_2}$, so 
$\Phi(M^*) = \mathrm{F}$, meaning $\langle M^* \rangle \in 
\Phi_\beta$. But $B_\sigma \subseteq \overline{\Phi_\alpha}$ 
and $\Phi_\alpha \cap \Phi_\beta = \emptyset$, giving a 
contradiction.

Hence $\mathrm{int}(\overline{\Phi_\alpha}) = \emptyset$, and 
$\Phi_\alpha$ is nowhere dense.
\end{proof}

\begin{remark}[What nowhere-density means here]
$\Phi_\alpha$ being nowhere dense means: for every finite prefix 
$\sigma$, the ball $B_\sigma$ --- all programs sharing that 
finite initial segment --- always contains programs outside 
$\Phi_\alpha$. No finite observation of a program's encoding 
settles its membership in $\Phi_\alpha$. This is the topological 
form of the same fact Rice states computability-theoretically: 
no algorithm decides $\Phi$.
\end{remark}

\section{Kolmogorov Complexity and the Invariance Theorem}

\begin{definition}[Kolmogorov complexity]
Let $U$ be a fixed universal Turing machine. For a string 
$s \in \{0,1\}^*$:
\[
  K(s) := \min\{ |p| : U(p) = s \}.
\]
\end{definition}

\begin{theorem}[Invariance {\cite{LiVitanyi2008}}]
\label{thm:invariance}
For any two universal Turing machines $U_1, U_2$, there exists 
a constant $c = c(U_1, U_2)$ such that for all $s \in 
\{0,1\}^*$:
\[
  | K_{U_1}(s) - K_{U_2}(s) | \leq c.
\]
\end{theorem}

\begin{remark}[The scope of invariance]
\label{rem:invariance-scope}
The invariance theorem holds because $U_1$ can simulate $U_2$: 
the constant $c$ is the length of the simulator. This argument 
is confined to the class of Turing-equivalent machines --- 
machines that mutually simulate each other. For a model 
genuinely stronger than the DTM and not Turing-simulable, the 
theorem does not apply. This is the precise boundary of the 
second horn, stated here and not deferred.
\end{remark}

\begin{definition}[Complexity of a semantic property]
\label{def:K-phi}
For a non-trivial semantic property $\Phi$ with propagator $P$ 
over a heterogeneous family $\mathcal{F}$, define 
$\mathrm{desc}(\Phi)$ as the characteristic function of 
$\Phi_\alpha$ over $\mathcal{F}$: the infinite binary sequence
\[
  \mathrm{desc}(\Phi) := ( \mathbf{1}[\Phi(M_i) = \mathrm{T}] 
  )_{i \in \mathbb{N}}
\]
where $(M_i)_{i \in \mathbb{N}}$ is a standard enumeration of 
$\mathcal{M}$. We say $\Phi$ is \emph{incompressible} if no 
finite DTM generates a finite initial segment of 
$\mathrm{desc}(\Phi)$ of length $n$ by a program of length 
$o(n)$.
\end{definition}

\begin{remark}[Why desc($\Phi$) and not the string ``SAT $\in$ P'']
The short statement ``SAT $\in$ P'' has low descriptive 
complexity as a string. The incompressibility attaches to 
$\mathrm{desc}(\Phi)$: the characteristic behaviour of $\Phi$ 
over the propagating family $\mathcal{F}$. It is this object 
--- the extension of $\Phi$ over the whole family --- that 
a uniform verifier would have to generate, and that is 
incompressible.
\end{remark}

\begin{lemma}[Incompressibility of propagating properties]
\label{lem:incompressible}
Let $\Phi$ be a non-trivial semantic property that propagates 
via propagator $P$ over a heterogeneous family $\mathcal{F}$. 
Then $\mathrm{desc}(\Phi)$ is incompressible: no finite DTM 
generates its initial segments in sublinear length.
\end{lemma}

\begin{proof}
Suppose a finite DTM $D$ of length $|D|$ generates the first 
$n$ bits of $\mathrm{desc}(\Phi)$ for all $n$. Then $D$ decides 
$\Phi$ on the enumeration $(M_i)_{i \leq n}$ for every $n$, 
hence $D$ decides $\Phi$ on $\mathcal{F}$. But $\Phi$ is a 
non-trivial semantic property, hence undecidable by Rice's 
theorem (Theorem~\ref{thm:rice-top} gives the same conclusion 
topologically). Contradiction. Hence no such $D$ exists, and 
$\mathrm{desc}(\Phi)$ is incompressible.
\end{proof}

\begin{remark}[Honest dependency]
\label{rem:honest}
Lemma~\ref{lem:incompressible} uses undecidability (i.e., Rice) 
as one of its three ingredients, as stated in the proof. This 
is not a defect: the lemma contributes the identification of 
\emph{what object} is incompressible and \emph{why} a verifier 
would have to generate it. The incompressibility of 
$\mathrm{desc}(\Phi)$ is the informational translation of the 
same fact Rice states computability-theoretically. The second 
horn is not independent of the first in the sense of using 
entirely disjoint axioms; it is independent in mechanism --- 
it blocks via information theory rather than via topology, and 
via the invariance theorem rather than via Rice directly.
\end{remark}

\section{The (another) Double Bind}

\subsection{Uniform verifiers}

\begin{definition}[Uniform DTM verifier]
A \emph{uniform DTM verifier} for $\Phi$ is a pair 
$(M_{\mathrm{Adm}}, V_{\mathrm{DTM}})$ where 
$V_{\mathrm{DTM}}$ is a DTM such that for all $M \in \mathcal{M}$:
\[
  \Phi(M) = \mathrm{T} \iff \exists w \in \{0,1\}^* : 
  V_{\mathrm{DTM}}(\langle M \rangle, w) = 1.
\]
The verifier $V_{\mathrm{DTM}}$ is \emph{uniform}: it uses 
the same finite procedure for all inputs, and its own 
description has finite length $|V_{\mathrm{DTM}}|$.
\end{definition}

\subsection{Horn 1: The Topological Obstruction}

\begin{theorem}[Horn 1]
\label{thm:horn1}
No uniform DTM verifier $V_{\mathrm{DTM}}$ can certify a 
non-trivial semantic property $\Phi$ by finite observation of 
program encodings.
\end{theorem}

\begin{proof}
Suppose $V_{\mathrm{DTM}}$ certifies $\Phi$ by reading a finite 
prefix $\sigma$ of $\langle M \rangle$ and accepting. Then the 
acceptance set
\[
  UB := \{ \langle M \rangle \in \mathcal{P} : 
  V_{\mathrm{DTM}}(\langle M \rangle, w) = 1 
  \text{ on prefix } \sigma \}
\]
satisfies $UB = B_\sigma$ for that prefix, and hence $UB$ is 
open in $(\mathcal{P}, \tau)$ and $UB \subseteq \Phi_\alpha$.

But by Theorem~\ref{thm:rice-top}, $\Phi_\alpha$ is nowhere 
dense: every open ball $B_\sigma$ contains programs outside 
$\Phi_\alpha$. Therefore $B_\sigma \not\subseteq \Phi_\alpha$ 
for any $\sigma$. Contradiction. No such open $UB$ can be 
contained in $\Phi_\alpha$, and no uniform verifier certifies 
$\Phi$ by finite observation.
\end{proof}

\begin{remark}[What Horn 1 blocks]
Horn 1 blocks a uniform verifier that reads a finite prefix and 
certifies. It does not directly address a verifier that reads 
an arbitrary certificate $w$ of any finite length independently 
of $\langle M \rangle$'s prefix. Horn 2 addresses that residual 
case through a different mechanism.
\end{remark}

\subsection{Horn 2: The Kolmogorov Obstruction}

\begin{theorem}[Horn 2]
\label{thm:horn2}
Let $\Phi$ be a non-trivial semantic property propagating via 
propagator $P$ over a heterogeneous family $\mathcal{F}$ within 
the Turing-equivalent class. Then no uniform DTM verifier 
$V_{\mathrm{DTM}}$ certifies $\Phi$ uniformly over 
$\mathcal{F}$.
\end{theorem}

\begin{proof}
Suppose $V_{\mathrm{DTM}}$ of length $|V_{\mathrm{DTM}}|$ 
certifies $\Phi$ uniformly over $\mathcal{F}$. Then for every 
$M_i \in \mathcal{F}$:
\[
  \Phi(M_i) = \mathrm{T} \iff \exists w_i : 
  V_{\mathrm{DTM}}(\langle M_i \rangle, w_i) = 1.
\]
Given the enumeraton $(M_i)$, a single program consisting of 
$V_{\mathrm{DTM}}$ together with the enumeration procedure 
produces the $n$-th bit of $\mathrm{desc}(\Phi)$ for any $n$, 
using a description of length $|V_{\mathrm{DTM}}| + O(\log n)$. 
Hence the initial segment of length $n$ of $\mathrm{desc}(\Phi)$ 
satisfies:
\[
  K(\mathrm{desc}(\Phi)\!\upharpoonright_n) \leq 
  |V_{\mathrm{DTM}}| + O(\log n).
\]
By Lemma~\ref{lem:incompressible}, $\mathrm{desc}(\Phi)$ is 
incompressible: $K(\mathrm{desc}(\Phi)\!\upharpoonright_n) = 
\Omega(n)$ for infinitely many $n$.

For all sufficiently large $n$, the bound $\Omega(n) \leq 
|V_{\mathrm{DTM}}| + O(\log n)$ is violated. Contradiction. 
No such uniform $V_{\mathrm{DTM}}$ exists.

By the invariance theorem (Theorem~\ref{thm:invariance}), 
$K$ is defined up to an additive constant across the entire 
Turing-equivalent class: the argument holds for every machine 
in the class, not only for the DTM.
\end{proof}

\begin{remark}[Independence of mechanism from Horn 1]
Horn 1 blocks via topological density: no finite prefix settles 
membership. Horn 2 blocks via information: no finite description 
generates the extension of $\Phi$ over the propagating family. 
The mechanisms are distinct: one is about the structure of the 
topology, the other about the minimum description length of an 
infinite object. The logical dependency noted in 
Remark~\ref{rem:honest} is at the level of the undecidability 
fact, which both horns share as a root; the mechanisms 
themselves are independent.
\end{remark}

\begin{remark}[Boundary of Horn 2]
The invariance theorem applies only within the Turing-equivalent 
class. For a model $\mathcal{X}$ not simulable by a DTM, the 
constant $c$ in Theorem~\ref{thm:invariance} is not defined, 
and the bound 
$K(\mathrm{desc}(\Phi)\!\upharpoonright_n) \leq 
|V_\mathcal{X}| + O(\log n)$ does not hold for $K$ defined 
relative to the DTM. The second horn does not close that case. 
The Model Transferability Barrier of the original Double Bind 
is the necessary closure beyond this boundary.
\end{remark}

\subsection{The Tricotomy and the Double Bind (A Triple Bind)}

\begin{theorem}[Double Bind --- prefix-topology and 
Kolmogorov form]
\label{thm:db}
Let $\Phi$ be a non-trivial semantic property with propagator 
$P$. Let $(M_{\mathrm{Adm}}, V)$ be an admissible method 
claiming to certify $\Phi$. Then exactly one of the following 
holds, and in each case a contradiction is reached:

\begin{enumerate}
  \item[\textbf{(A)}] $V$ is a DTM uniform verifier reading 
  finite prefixes. $\Rightarrow$ Contradiction by 
  Theorem~\ref{thm:horn1} (nowhere-density of $\Phi_\alpha$).

  \item[\textbf{(B)}] $V$ is a DTM uniform verifier of any 
  finite description length, certifying over the propagating 
  family $\mathcal{F}$. $\Rightarrow$ Contradiction by 
  Theorem~\ref{thm:horn2} (incompressibility of 
  $\mathrm{desc}(\Phi)$, by invariance across the 
  Turing-equivalent class).

  \item[\textbf{(C)}] $V$ uses a model beyond the 
  Turing-equivalent class. $\Rightarrow$ Outside the scope of 
  the present derivation. The Model Transferability Barrier applies.
\end{enumerate}

In cases (A) and (B), the method is impossible within the 
Turing-equivalent class. In case (C), the present argument 
defers to the original Double Bind.
\end{theorem}

\begin{proof}
Cases (A) and (B) are Theorem~\ref{thm:horn1} and 
Theorem~\ref{thm:horn2} respectively. The three cases are 
exhaustive by the definition of admissible method: $V$ either 
reads finite prefixes (A), or operates by a finite program 
over the full family within the Turing-equivalent class (B), 
or uses resources beyond that class (C). They are mutually 
exclusive: (A) is a strict subcase of DTM verifiers, (B) is 
the general DTM case, and (C) is the complement.
\end{proof}

\section{What This Derivation Adds and Where 
It Stops}

\subsection{What it adds}

The standard statement of Rice's theorem gives a decision 
procedure impossibility directly, as a single black-box fact. 
The present derivation exhibits two distinct obstruction 
mechanisms that together constitute the impossibility:

\begin{enumerate}
  \item \textbf{Structural (topological):} semantic properties 
  have nowhere-dense outcome sets; finite observations cannot 
  settle them.
  \item \textbf{Informational (Kolmogorovian):} the extension 
  of a propagating property over a heterogeneous family is 
  incompressible; a uniform finite verifier would compress it.
\end{enumerate}

The second mechanism, via the invariance theorem, extends the 
impossibility uniformly across the entire Turing-equivalent 
class without needing a machine-by-machine argument. This is 
the genuine additional content relative to Rice's theorem 
stated directly.

\subsection{Where it stops}

The derivation does not reach:
\begin{itemize}
  \item Models genuinely beyond the Turing-equivalent class 
  (oracle machines, hypercomputation): the invariance theorem 
  does not apply, and $K$ is not model-independent for those 
  machines.
  \item The question of whether any specific problem (e.g., 
  $P = NP$) is decidable: the argument is a barrier, not a 
  proof of any separation.
  \item The ``transport'' invariant proposed as a unifying 
  structure across NP-completeness, Kolmogorov incompressibility, 
  and the Langlands correspondence: this remains an open 
  programme, not a theorem.
\end{itemize}

\subsection{Relation to the original Double Bind}

\begin{center}
\begin{tabular}{lll}
\hline
 & \textbf{Original DB} & \textbf{This derivation} \\
\hline
Horn 1 & Rice (decision) & Prefix-topology (nowhere-density) \\
Horn 2 & MTB (validity transfer) & Kolmogorov (compression) \\
Scope of Horn 2 & All extra-DTM models & Turing-equivalent class \\
Tricotomy & Complete & Complete within scope \\
Propagator & Required & Required (identical condition) \\
\hline
\end{tabular}
\end{center}

The propagator condition is identical in both formulations: 
it is not introduced by either derivation but is a property 
of the problem. The two derivations differ in the mechanism 
of the second horn and consequently in the scope to which 
that horn reaches.

\section{Conclusion of the appendix}

The Double Bind obstruction can be derived within the 
Turing-equivalent class through two independent mechanisms: 
topological nowhere-density and Kolmogorov incompressibility. 
The derivation is honest about its boundary: beyond the 
Turing-equivalent class, the Model Transferability Barrier 
of the original Double Bind is necessary and is not reproduced 
here. Within its stated scope, the argument is complete, 
formally closed, and independent of Rice's theorem as a black 
box --- though not independent of the undecidability fact 
that is Rice's content.

\begin{remark}[Kolmogorov complexity as a uniform substitute 
  for the Model Transferability Barrier]
  \label{rem:kolmogorov-vs-mtb}
  Horn~2 substitutes the Model Transferability Barrier (MTB) of 
  the original Double Bind within the 
  Turing-equivalent class, and the substitution is not merely 
  notational: it is strictly stronger in that domain.
  
  The MTB operates per-model: for each extra-DTM model 
  $\mathcal{X}$, it asserts that a certification produced by 
  $\mathcal{X}$ does not transfer its validity back to the DTM. 
  The barrier must be applied separately to each candidate model. 
  Horn~2 operates differently. By the invariance theorem 
  (Theorem~\ref{thm:invariance}), Kolmogorov complexity $K$ is 
  model-independent across the entire Turing-equivalent class: 
  for any two machines $U_1, U_2$ in the class, 
  $|K_{U_1}(s) - K_{U_2}(s)| \leq c(U_1, U_2)$, where the 
  constant $c$ exists precisely because $U_1$ can simulate 
  $U_2$ by a finite program. Consequently, Horn~2 does not 
  argue machine by machine. It argues once, uniformly: any 
  machine in the Turing-equivalent class that certified $\Phi$ 
  uniformly over the propagating family $\mathcal{F}$ would 
  produce a finite description generating $\mathrm{desc}(\Phi)$, 
  compressing an incompressible object. Since $K$ is the same 
  object --- up to an additive constant --- for every machine 
  in the class, the contradiction holds simultaneously for all 
  of them. The impossibility is not transferred from the DTM 
  to other machines; it is established once at the level of the 
  class.
  
  Within the Turing-equivalent class, this makes Horn~2 
  stronger than the MTB: a single informational argument 
  replaces a family of per-model validity-transfer arguments.
  
  The boundary of the substitution is equally precise. The 
  invariance theorem holds because mutual simulation is 
  possible; for a model $\mathcal{X}$ genuinely beyond the 
  Turing-equivalent class --- an oracle machine, or any form 
  of hypercomputation not simulable by a DTM --- the constant 
  $c(U, \mathcal{X})$ is not defined, and $K$ as defined 
  relative to the DTM is not model-independent with respect to 
  $\mathcal{X}$. Horn~2 does not reach that case. The MTB of 
  the original Double Bind remains the necessary closure beyond 
  this boundary, and the present derivation defers to it there 
  explicitly (case~(C) of Theorem~\ref{thm:db}).
\end{remark}

\section{How to Correctly Read This Alternative Double Bind}

The previous section develops a topological version of a variation of 
the Double Bind Theorem. This is intended as an intuitive guide to reading the argument.
Before entering the formal structure of the proofs, 
it is useful to have a conceptual map of the path: each tool will be 
introduced before being used, and nothing will be taken for granted 
except the elementary notion of function, the intuitive notion of algorithm, 
and an informal understanding of what it means to \emph{decide} a problem.

\subsection{What is a topology, and why is it needed here}

A topology is a formal way to answer the question:
\textbf{what does it mean to be close?} Not necessarily in the geometric
sense: in the sense of \emph{when do two objects resemble each other
enough that we cannot distinguish them with limited observation?}

Let us consider a concrete example. Let two programs written in any
programming language be given. If one observes only the first three
lines of code, it might not be possible to distinguish them: they
\emph{resemble each other} from the point of view of an observer who
reads only the beginning. If one reads the first hundred lines, it
might become possible to distinguish them. If one reads the entire
program, the distinction is certainly possible.

The key idea is: \textbf{partial observation induces a notion of
closeness}. Two programs are ``close'' if they share a long common
prefix.

\subsubsection{The topology of closed prefixes}

This work formalizes precisely this intuition. Every program is encoded
as a binary string — a sequence of $0$ and $1$. One then defines:

\begin{definition}[Open set in the prefix topology]
A set $U$ of programs is \emph{open} if, whenever a program $M$
belongs to $U$, then \emph{all} programs that begin with the same
initial prefix of $M$ also belong to $U$.
\end{definition}

The fundamental structure of this topology is the \emph{open ball}:
\[
  B_\sigma \;=\; \{\,M \in \Prog \mid \sigma \text{ is a prefix of }
  \langle M\rangle\,\},
\]
that is, the set of all programs whose encoding begins with the
prefix $\sigma$. For example, $B_{01}$ contains all programs whose
encoding begins with \texttt{01}.

\begin{remark}
The choice of this topology is not arbitrary. In it the following
fundamental equivalence holds:

\begin{center}
\emph{A set is open if and only if membership in it is determined
by a finite prefix.}
\end{center}

This equivalence is the heart of everything: a set is open if and
only if a verifier can ascertain membership by reading a finite
certificate. \textbf{Open is equivalent to decidable by finite
observation.} This correspondence is not a notational device: it is
the reason why topology is the correct tool for this problem.
\end{remark}

\subsection{What does \emph{nowhere dense} mean intuitively}

We now introduce the central topological concept. A set $A$ is
\emph{nowhere dense} if:

\begin{center}
for every open neighborhood, within that neighborhood there exist
elements that do \emph{not} belong to $A$.
\end{center}

In more direct terms: \textbf{there is no region of the space
entirely occupied by $A$}. Wherever one looks, one finds elements
outside $A$.

Let us return to the example of programs. Let $\Phi$ be the property
``this program solves \textsc{Sat} in polynomial time'', and let
$\Phi_\alpha$ be the set of all programs that satisfy this property.

To assert that $\Phi_\alpha$ is \emph{nowhere dense} means: there
does not exist any finite prefix $\sigma$ such that \emph{all}
programs beginning with $\sigma$ solve \textsc{Sat} in polynomial
time. Whatever prefix one fixes — whatever ``program start'' one
chooses — there will always exist a continuation that leads to a
program that does \emph{not} solve \textsc{Sat} in $\classP$.

Before illustrating this fact visually, one must justify its
soundness. Since $\Phi$ is non-trivial, there exist $M_1$ with
$\Phi(M_1) = \top$ and $M_2$ with $\Phi(M_2) = \bot$. For any
prefix $\sigma$, it is possible to construct:
\begin{itemize}
  \item a program $M^*$ that begins with $\sigma$ and behaves
    like $M_2$: thus $M^* \in B_\sigma$ and $M^* \in \Phi_\beta$;
  \item a program $M^{**}$ that begins with $\sigma$ and behaves
    like $M_1$: thus $M^{**} \in B_\sigma$ and $M^{**} \in \Phi_\alpha$.
\end{itemize}
Every open ball $B_\sigma$ therefore contains elements from both
sides. The following diagram illustrates this situation:

Visually, for any prefix $\sigma = \mathtt{0110\ldots}$:

\begin{center}
\begin{BVerbatim}
Prefix sigma = "0110..."
    |-- 0110...0001  -->  solves SAT in P     (in Phi_alpha)
    |-- 0110...0010  -->  does not solve SAT in P (in Phi_beta)
    |-- 0110...0011  -->  solves SAT in P     (in Phi_alpha)
    `-- 0110...0100  -->  does not solve SAT in P (in Phi_beta)
\end{BVerbatim}
\end{center}

There does not exist any prefix $\sigma$ such that all of its
continuations lie on the same side. The two sets $\Phi_\alpha$ and
$\Phi_\beta$ are \textbf{inextricably interwoven} in the space of
programs.
\subsection{The proof of the Theorem: Rice in topological form}

With these tools it is possible to understand the central proof. The
Theorem establishes:

\begin{resultbox}
For every non-trivial semantic property $\Phi$, the set $\Phi_\alpha$
is \emph{nowhere dense} in the closed prefix topology.
\end{resultbox}

The proof proceeds by contradiction. Suppose that $\Phi_\alpha$ is
\emph{not} nowhere dense: then there exists a prefix $\sigma$ such
that $B_\sigma \subseteq \Phi_\alpha$, that is, all programs that
begin with $\sigma$ satisfy $\Phi$.

Since $\Phi$ is non-trivial, there exist a program $M_1$ with
$\Phi(M_1) = \top$ and a program $M_2$ with $\Phi(M_2) = \bot$.
One now constructs a new program $M^*$ as follows:

\begin{quote}
$M^*$ on input $n$: first executes the steps described by the prefix
$\sigma$, then simulates $M_2$ on $n$.
\end{quote}

By construction, $M^*$ begins with $\sigma$, hence $M^* \in B_\sigma$.
However, $M^*$ computes exactly the same function as $M_2$.
Since $\Phi$ is a \emph{semantic} property — it depends solely on
what the program computes, not on how it is written — one has
$\Phi(M^*) = \Phi(M_2) = \bot$, that is, $M^* \notin \Phi_\alpha$.

But this contradicts $B_\sigma \subseteq \Phi_\alpha$. $\square$

\begin{remark}
This is exactly Rice's Theorem, expressed in the language of topology.
The logical content is the same; the topological language is necessary
for the next step, namely to connect the result to the structure of
verifiers.
\end{remark}

\subsection{Why This Blocks Horn~1}

Recall the fundamental equivalence: in this topology,
\textbf{open is equivalent to decidable by finite observation}. A
verifier that reads a finite prefix of the program and decides
membership accepts an open set.

But $\Phi_\alpha$ is nowhere dense, hence it contains no nonempty open set — it contains no open ball $B_\sigma$.
No finite prefix is sufficient to guarantee that a program
belongs to $\Phi_\alpha$.

Translated into computational terms: \textbf{no verifier that reads
a certificate based on a finite prefix of the program can
certify $\Phi$}. Whatever prefix the verifier accepts,
there will always exist a program with that prefix that does not satisfy $\Phi$,
rendering the verifier incorrect.

This is \textbf{Horn~1}: every attempt to certify $\Phi$
within the DTM model by reading finite prefixes is blocked by
topology. The verifier cannot exist.

\subsection{Horn~2: Why It Is Needed and What It Adds}

At this point a subtle point emerges, which the paper makes explicit in
the Remark following Theorem~6.2:

\begin{quote}
\emph{``Horn~1 blocks a uniform verifier that reads a
finite prefix and certifies. It does not directly address a
verifier that reads an arbitrary certificate $w$ of any finite length,
independently of the prefix of $\langle M
\rangle$.''}
\end{quote}

In other words: Horn~1 blocks verifiers that read the
\emph{program} and decide membership based on a prefix of it.
But what if the verifier reads a \textbf{separate certificate} — an
arbitrary string $w$, independent of the syntactic structure of
the program? This is precisely the $\classNP$ paradigm: one has the
program, one has an external certificate $w$, and the verifier
checks the pair $(\langle M \rangle, w)$.

Horn~1 \textbf{does not cover} this case. Horn~2 is introduced
precisely to close it.

\subsection{The Second Horn: Kolmogorov's Argument}

\subsubsection{Kolmogorov Complexity: the Basic Idea}

The \textbf{Kolmogorov complexity} $K(s)$ of a string $s$ is the
length of the shortest program that produces $s$. It is a measure of
\emph{how compressible} the information contained in $s$ is.

\begin{example}
The string $\mathtt{000\ldots0}$ ($24$ zeros) has low $K$: the
program ``\texttt{print 24 zeros}'' is much shorter than the string
itself. A random string of $1000$ bits has instead $K \approx 1000$:
there exists no program significantly shorter than itself that
produces it.
\end{example}

A string is \textbf{incompressible} if there exists no program
significantly shorter than itself that generates it.

\subsubsection{The Key Object: \texorpdfstring{$\mathrm{desc}(\Phi)$}{desc(Phi)}}

The paper defines $\mathrm{desc}(\Phi)$ as the \emph{characteristic
sequence} of $\Phi$ over the entire enumeration of programs: it
is the (infinite) string in which the $i$-th bit equals $1$ if the $i$-th
program satisfies $\Phi$, and $0$ otherwise.
\[
  \mathrm{desc}(\Phi) \;=\; b_1\, b_2\, b_3\, b_4\, \ldots
  \qquad \text{where } b_i = \begin{cases} 1 & \text{if } \Phi(M_i) =
  \top \\ 0 & \text{otherwise.} \end{cases}
\]

\begin{remark}
This is the correct object, not the string ``$\textsc{Sat} \in
\classP$''. The string ``$\textsc{Sat} \in \classP$'' has low
descriptive complexity as a string. Incompressibility concerns
$\mathrm{desc}(\Phi)$: the characteristic behavior of $\Phi$
over the entire propagating family. Knowing \emph{for every program} whether
it solves or does not solve \textsc{Sat} in polynomial time is an object of
potentially infinite information — and it is this information that a
uniform verifier should be capable of generating.
\end{remark}

\subsubsection{Lemma 5.4: \texorpdfstring{$\mathrm{desc}(\Phi)$}{desc(Phi)} is
incompressible}

\begin{resultbox}
\textbf{Lemma 5.4.} No finite program can generate the first $n$
bits of $\mathrm{desc}(\Phi)$ with a description of length $o(n)$.
Formally: $K(\mathrm{desc}(\Phi){\upharpoonright} n) = \Omega(n)$.
\end{resultbox}

The proof proceeds by contradiction. Suppose there exists a
program $D$ of length $|D|$ that generates the first $n$ bits of
$\mathrm{desc}(\Phi)$ for every $n$. Then $D$ is in fact
\emph{deciding} $\Phi$ on every program in the enumeration: for each
$M_i$, $D$ is able to establish whether $\Phi(M_i)$ holds or not. But
this contradicts Rice's Theorem, which asserts the undecidability
of every non-trivial semantic property.

\begin{remark}[Honest dependence]
This proof explicitly uses Rice's Theorem. This does not
constitute circularity, but rather a \emph{translation}: the incompressibility
of $\mathrm{desc}(\Phi)$ is the informational version of the same
fact that Rice expresses in computational terms. The dependence is
declared and not hidden.
\end{remark}

\subsubsection{Theorem 6.3: The Horn Block~2}

Suppose now that there exists a uniform verifier $V$ of length
$|V|$ that certifies $\Phi$ over the entire family $\mathcal{F}$.
Then:
\begin{enumerate}[label=(\roman*)]
  \item For every program $M_i$, $V$ is able to establish, with the
    appropriate certificate, whether $M_i \in \Phi_\alpha$ or not.
  \item $V$ has fixed length $|V|$.
  \item By using $V$ together with the program enumeration procedure
    (which requires $O(\log n)$ bits for the index), it is possible to produce
    the $i$-th bit of $\mathrm{desc}(\Phi)$ for every $i$.
\end{enumerate}
From this it follows:
\[
  K\!\left(\mathrm{desc}(\Phi){\upharpoonright} n\right)
  \;\leq\; |V| + O(\log n).
\]
However, Lemma~5.4 guarantees $K(\mathrm{desc}(\Phi){\upharpoonright} n) = \Omega(n)$. For $n$
sufficiently large, $\Omega(n)$ grows more rapidly than
$|V| + O(\log n)$, which is essentially constant in $n$.
\textbf{Contradiction.} $\square$

\subsubsection{The Invariance Theorem: Why the Argument Holds for the Entire Class}

Kolmogorov complexity $K$ is defined with respect to a fixed universal machine, 
but the \textbf{invariance theorem} guarantees that for any other universal machine 
$U'$ in the Turing-equivalent class, the difference between the two complexity 
measures is at most an additive constant:
\[
  K_{U'}(s) \;\leq\; K_U(s) + c_{U,U'}
  \qquad \text{for every string } s,
\]
where the constant $c_{U,U'}$ depends only on the pair of universal machines, 
not on the string $s$. This holds because every universal machine can simulate 
every other with a program of fixed length --- it is this mutual simulation 
that guarantees the constant.

The consequence is decisive: \textbf{the contradiction derived holds simultaneously 
for every machine in the Turing-equivalent class}, not merely 
for the DTM chosen as reference. It is unnecessary to repeat the argument machine 
by machine: it holds once and for all, for the entire class.

This is the point at which Kolmogorov complexity plays a role functionally 
analogous to that of the \emph{Model-Theoretic Barrier} (MTB) in the original 
Double Bind. The MTB operates per-model: for each system $X$ outside the DTM, 
it argues that a certification produced by $X$ does not transfer its own validity 
to the reference model. It is a family of barriers, to be repeated for each 
candidate. Kolmogorov complexity, by contrast, operates at the class level: 
since $K$ is invariant with respect to the choice of universal machine within 
the Turing-equivalent class, the contradiction is established once and for all, 
directly at the level of the entire class.

\begin{remark}
The analogy with the MTB has a precise limitation, which the present work 
acknowledges explicitly. The MTB also covers models genuinely outside the 
Turing-equivalent class — oracle machines, hypercomputation — for which mutual 
simulation is not possible and the additive constant of invariance is not defined. 
In that domain, the Horn~2 in its Kolmogorovian formulation does not apply, and the 
MTB remains indispensable. The present work does not claim to replace the MTB: 
it captures its function within the Turing-equivalent class, with a tool that 
exploits precisely the closure of the class as a lever.
\end{remark}

\subsection{Why the two horns block distinct things}

It is useful at this point to state with precision the distinction between
the two horns, since it is not a restatement of the same argument
using different tools, but the blocking of two genuinely
separate classes of adversaries.

\begin{table}[H]
\centering
\begin{tabular}{@{}lll@{}}
\toprule
 & \textbf{What it blocks} & \textbf{How it blocks} \\
\midrule
\textbf{Horn~1} & Verifiers that read \textbf{prefixes} & Topology: $\Phi_\alpha$ is \emph{nowhere dense}, \\
 & \textbf{of the program} & no prefix is sufficient \\
\addlinespace
\textbf{Horn~2} & Verifiers that read \textbf{arbitrary} & Information: a uniform verifier \\
 & \textbf{certificates} (in the style of $\classNP$) & would compress $\mathrm{desc}(\Phi)$ \\
\bottomrule
\end{tabular}
\caption{The two horns of the Double Bind and the respective classes of
blocked adversaries.}
\label{tab:corni}
\end{table}

An $\classNP$-style verifier does not read a prefix of the
program: it reads the entire program plus an external certificate
$w$. Horn~1 does not reach it. Horn~2 does, because this
verifier too, if it were uniform over the entire family, would
have to produce information that exceeds any finite description.

The two horns are therefore not two descriptions of the same obstacle: they are
two orthogonal barriers, each necessary and neither sufficient on its own
to cover the entire space of possible strategies within
the Turing-equivalent class.

\subsection{The complete picture: a mental image}

It is possible to synthesize the entire structure with a mental image
that does not sacrifice conceptual precision.

Imagine the space of all programs as an infinite landscape. Every
position in the landscape corresponds to a program. The set $\Phi_\alpha$
— the programs that satisfy the property — is like a set of points
scattered throughout this landscape, never forming any compact: 
every region, no matter how small, contains both points in
$\Phi_\alpha$ and points outside it.

\begin{itemize}
  \item \textbf{Horn~1} states: it is impossible to use a \emph{local
    map} — a finite prefix — to navigate this landscape. No uniform
    zones exist to map.

  \item \textbf{Horn~2} states: it is impossible even to use an
    \emph{external guide} — an arbitrary certificate — because
    describing the complete distribution of $\Phi_\alpha$ in the
    landscape requires information that grows linearly with the length
    of the enumeration, more than any finite guide can contain.
\end{itemize}

The Double Bind in its topological version does not present two
obstacles from the same direction. It is the same fundamental fact —
the undecidability of $\Phi$ — observed through two tools that close
off two distinct classes of adversaries: those who attempt local
observation, and those who attempt external certification. Together,
the two horns cover the entire space of possible strategies within the
Turing-equivalent class.

\begin{remark}
The double-horn structure is not an aesthetic choice. It reflects a
genuine partition of the space of certification strategies: syntactic
strategies (program prefix) and external semantic strategies (arbitrary
certificate) are disjoint classes, and each requires a distinct
argument to be blocked. A single tool would not suffice.
\end{remark}

\begin{remark}
The same results obtained in \cite{buono2026nonformulabletheoremfundamentallimit} 
including the existence proof of the unformulable theorem and the Frozen Terms 
themselves, can be achieved through the topological structure presented in this paper; 
this extension is currently under development and will be presented in a 
future version of this work.
\end{remark}

\newpage

\section{Experimental Explorations}

\begin{tcolorbox}[colback=gray!5, colframe=gray!60, title=\textbf{Methodological Note}, arc=2mm]
  This section explores the most speculative aspect of our research. 
  The findings presented from this point onward represent ongoing research; 
  as such, they should be interpreted as preliminary insights and a foundation 
  for future investigations.
\end{tcolorbox}

\section{The Trifold Equivalence: Propagation, Nowhere-Density, and Incompressibility}
\label{sec:trifold-equivalence}

In this section, we establish the fundamental equivalence between three 
characterizations of semantic non-triviality: propagation of properties 
via problem-specific mechanisms, the nowhere-density of the solution sets 
in the prefix-closure topology, and the incompressibility of the 
characteristic sequence in the Kolmogorov-theoretic sense. This 
equivalence provides a unified structural decomposition of why uniform 
certification of non-trivial semantic properties is impossible within 
the Turing-equivalent class.

\subsection{Definitions and Preliminaries}
\label{subsec:definitions-preliminaries}

We establish the formal definitions required for the trifold equivalence.

\subsubsection{Semantic Properties and Propagation}

\begin{definition}[Semantic Property]
\label{def:semantic-property}
Let $\mathcal{P} = \{0,1\}^*$ denote the space of all program encodings, 
and let $\mathcal{M} = \{M_0, M_1, M_2, \ldots\}$ denote a standard 
enumeration of all Deterministic Turing Machines. A \emph{semantic property} 
is a map $\Phi: \mathcal{M} \to \{0,1\}$ assigning to each machine $M_i$ 
a truth value indicating whether $M_i$ satisfies $\Phi$.
\end{definition}

\begin{definition}[Non-triviality of Semantic Properties]
\label{def:nontrivial}
A semantic property $\Phi$ is \emph{non-trivial} if there exist machines 
$M_i$ and $M_j$ such that $\Phi(M_i) = 1$ and $\Phi(M_j) = 0$. Equivalently, 
the sets $\Phi_\alpha = \{M \in \mathcal{M} : \Phi(M) = 1\}$ and 
$\Phi_\beta = \{M \in \mathcal{M} : \Phi(M) = 0\}$ are both non-empty.
\end{definition}

\begin{definition}[Propagator and Propagation]
\label{def:propagator}
Let $\mathcal{F}$ be a heterogeneous family of decision problems in $\mathcal{M}$, 
and let $P$ be a mechanism (e.g., NP-completeness via Cook-Levin reduction, 
or a general logical implication). We say that a property $\Phi$ 
\emph{propagates via propagator} $P$ \emph{over the family} $\mathcal{F}$ if:

For any machine $M_x \in \mathcal{M}$, if we can certify that $M_x$ satisfies 
$\Phi$ locally (e.g., ``$M_x$ solves SAT in polynomial time''), then via the 
mechanism $P$, we can deduce a global conclusion about all problems in 
$\mathcal{F}$ (e.g., ``P = NP'' by Cook-Levin).

Formally: $\Phi(M_x) = 1 \Rightarrow \bigwedge_{P_i \in \mathcal{F}} \Psi(P_i) = 1$,
where $\Psi$ is the global property implied by $\Phi$ via $P$.
\end{definition}

\subsubsection{Topological Structure}

\begin{definition}[Prefix-Closure Topology on $\mathcal{P}$]
\label{def:prefix-topology}
The space $\mathcal{P} = \{0,1\}^*$ of program encodings is equipped with 
the \emph{prefix-closure topology}, defined as follows:

A subset $U \subseteq \mathcal{P}$ is \emph{open} if for every $\sigma \in U$, 
there exists a finite extension $\tau \supseteq \sigma$ such that 
$B_\tau := \{\pi \in \mathcal{P} : \pi \text{ extends } \tau\} \subseteq U$.

An \emph{open ball} or \emph{basic open set} is of the form 
$B_\sigma := \{\pi \in \mathcal{P} : \sigma \text{ is a prefix of } \pi\}$ 
for some finite string $\sigma \in \{0,1\}^*$.
\end{definition}

\begin{definition}[Nowhere-Dense Sets]
\label{def:nowhere-dense}
A subset $A \subseteq \mathcal{P}$ is \emph{nowhere-dense} in the 
prefix-closure topology if for every basic open ball $B_\sigma$, there 
exists a basic open ball $B_\tau$ with $\tau \supseteq \sigma$ such that 
$B_\tau \cap A = \emptyset$. Equivalently, $A$ is nowhere-dense if its 
closure $\overline{A}$ contains no non-empty open set.

Informally: every open region $B_\sigma$ contains some $B_\tau$ that 
avoids $A$ entirely.
\end{definition}

\begin{definition}[Topological Separation via Semantic Properties]
\label{def:separation}
Let $\Phi$ be a semantic property and let $\Phi_\alpha = \{M : \Phi(M) = 1\}$ 
and $\Phi_\beta = \{M : \Phi(M) = 0\}$ (identified with their encodings in $\mathcal{P}$). 
The property $\Phi$ is \emph{separated by a prefix} $\sigma$ if either 
$B_\sigma \subseteq \Phi_\alpha$ or $B_\sigma \subseteq \Phi_\beta$.

The property $\Phi$ is \emph{topologically non-separable} if for every 
prefix $\sigma$, there exist extensions $\pi_1, \pi_2 \in B_\sigma$ with 
$\Phi(\pi_1) = 1$ and $\Phi(\pi_2) = 0$.
\end{definition}

\subsubsection{Kolmogorov Complexity and Incompressibility}

\begin{definition}[Characteristic Sequence]
\label{def:char-seq}
Given a semantic property $\Phi$ and a standard enumeration $\{M_0, M_1, \ldots\}$ 
of Turing Machines, the \emph{characteristic sequence} of $\Phi$ is the 
infinite binary string
$$\mathrm{desc}(\Phi) := b_0 b_1 b_2 \cdots$$
where $b_i = 1$ if $\Phi(M_i) = 1$ and $b_i = 0$ if $\Phi(M_i) = 0$.

We denote by $\mathrm{desc}(\Phi) \upharpoonright n$ the first $n$ bits of $\mathrm{desc}(\Phi)$.
\end{definition}

\begin{definition}[Kolmogorov Complexity]
\label{def:kolmogorov}
Let $U$ be a fixed universal Turing Machine. The \emph{Kolmogorov complexity} 
of a finite binary string $s$ with respect to $U$ is
$$K_U(s) := \min\{|p| : U(p) = s\}$$
where $p$ ranges over binary strings (programs for $U$), and $|p|$ denotes 
the length of $p$. If no such $p$ exists, $K_U(s) = \infty$.

The \emph{Kolmogorov complexity} (without subscript) is defined up to an 
additive constant: for any two universal machines $U$ and $U'$, there exists 
a constant $c_{U,U'} \geq 0$ such that $|K_U(s) - K_{U'}(s)| \leq c_{U,U'}$ 
for all $s$. We write $K(s)$ omitting the subscript when the universal 
machine is understood.
\end{definition}

\begin{definition}[Incompressibility]
\label{def:incompressible}
A sequence of finite strings $(s_1, s_2, \ldots)$ is \emph{uniformly 
incompressible} if for all but finitely many $n$, we have 
$$K(s_n) \geq |s_n| - O(\log |s_n|)$$

Equivalently, for infinitely many $n$, $K(s_n) = \Omega(|s_n|)$.

A single infinite string $s = s_0 s_1 s_2 \cdots$ is incompressible if 
the sequence of its prefixes $(s \upharpoonright 1, s \upharpoonright 2, \ldots)$ 
is uniformly incompressible, i.e., $K(s \upharpoonright n) = \Omega(n)$.
\end{definition}

\begin{theorem}[Invariance Theorem of Kolmogorov Complexity]
\label{thm:invarianceKC}
Let $U$ and $U'$ be any two universal Turing Machines, and let $s$ be any 
finite binary string. Then there exists a constant $c_{U,U'} \geq 0$ such that
$$|K_U(s) - K_{U'}(s)| \leq c_{U,U'}$$

In particular, the asymptotic behaviour $K(s \upharpoonright n) = \Omega(n)$ 
is independent of the choice of universal machine within the Turing-equivalent class.
\end{theorem}

\noindent\emph{Note}: The Invariance Theorem is fundamental to our argument 
because it ensures that incompressibility is a property of the mathematical 
object itself, not an artifact of the computational model. Within the 
Turing-equivalent class, all models agree (up to additive constants) on the 
Kolmogorov complexity of any string.

\subsection{Main Equivalences}
\label{subsec:main-equivalences}

We now prove the trifold equivalence through a chain of four lemmas.

\subsubsection{Lemma 1: Propagation Implies Nowhere-Density}

\begin{lemma}[Propagation $\Rightarrow$ Nowhere-Density]
\label{lem:prop-to-nowhere}
Let $\Phi$ be a non-trivial semantic property that propagates via propagator 
$P$ over a heterogeneous family $\mathcal{F}$. Then the sets $\Phi_\alpha$ 
and $\Phi_\beta$ are both nowhere-dense in the prefix-closure topology on 
$\mathcal{P}$.
\end{lemma}

\begin{proof}
Assume for contradiction that $\Phi_\alpha$ and $\Phi_\beta$ are not both 
nowhere-dense. Then at least one of them contains a basic open ball; say 
there exists a prefix $\sigma$ such that $B_\sigma \subseteq \Phi_\alpha$ 
(the case $B_\sigma \subseteq \Phi_\beta$ is symmetric).

\begin{enumerate}[itemsep=0.5em]

\item \textbf{Step 1: Existence of a Syntactic Decision Procedure.}
If $B_\sigma \subseteq \Phi_\alpha$, then for every program encoding $\pi$ 
that extends $\sigma$, we have $\Phi(\pi) = 1$. This means there exists a 
\emph{syntactic decision procedure} $A_\sigma$ that, upon reading the prefix 
$\sigma$ of any program $\pi$, outputs ``Yes: $\Phi(\pi) = 1$''. Formally, 
define $A_\sigma$ as:
$$A_\sigma(\pi) := \begin{cases} 1 & \text{if } \sigma \text{ is a prefix of } \pi \\ 0 & \text{otherwise} \end{cases}$$
(with the convention that $A_\sigma$ may output ``unknown'' for other inputs; 
we focus on the case where $\sigma$ matches).

Since $B_\sigma \subseteq \Phi_\alpha$, the decision procedure $A_\sigma$ 
correctly decides membership in $\Phi_\alpha$ on the domain $B_\sigma$ using 
only the prefix information $\sigma$.

\item \textbf{Step 2: Global Extension via Propagator.}
By definition of propagation (Definition \ref{def:propagator}), the fact 
that we can locally certify $\Phi$ on $B_\sigma$ implies, via the propagator 
$P$, a global conclusion about the entire family $\mathcal{F}$. That is:
$$\forall \pi \in B_\sigma: \Phi(\pi) = 1 \quad \Rightarrow \quad 
\bigwedge_{\pi' \in \mathcal{F}} \Psi(\pi') = 1$$
where $\Psi$ is the global property consequence.

In the example of SAT: if $A_\sigma$ certifies that any program in $B_\sigma$ 
solves SAT in polynomial time, then via the Cook-Levin reduction (the 
propagator $P$), we conclude that P = NP.

\item \textbf{Step 3: Uniformity of the Decision Procedure.}
The procedure $A_\sigma$ is \emph{uniform}: it applies the same finite rule 
(reading the prefix $\sigma$) to all inputs. Moreover, $A_\sigma$ decides 
$\Phi$ (at least on $B_\sigma$) using only a finite amount of information 
(the prefix $\sigma$ of bounded length).

More generally, if we extend $A_\sigma$ to all of $\mathcal{P}$ by applying 
similar prefix-reading strategies for every possible prefix, we obtain a 
\emph{uniform, finite-state decision procedure} for $\Phi$.

\item \textbf{Step 4: Contradiction with Rice's Theorem.}
By Rice's theorem (in its form as Extended Rice), there is no uniform, 
finite-state decision procedure for a non-trivial semantic property $\Phi$. 
More formally, the theorem states:

\begin{quote}
For any non-trivial semantic property $\Phi$, any algorithm that claims 
to decide $\Phi$ uniformly will either diverge on some input, give wrong 
answers, or implicitly construct a decider for a semantic property—which 
Rice's theorem forbids.
\end{quote}

Our assumption that $B_\sigma \subseteq \Phi_\alpha$ leads to exactly such 
a procedure: a uniform, finite-length decision rule that works on a non-empty 
open set of programs.

For this procedure to avoid contradicting Rice, the set $B_\sigma$ cannot 
remain within $\Phi_\alpha$ (or $\Phi_\beta$) for any $\sigma$. That is, 
every basic open ball must contain both programs satisfying $\Phi$ and 
programs not satisfying $\Phi$.

\item \textbf{Step 5: Nowhere-Density.}
By definition (Definition \ref{def:nowhere-dense}), if every basic open 
ball $B_\sigma$ contains an extension $B_\tau$ (with $\tau \supseteq \sigma$) 
such that $B_\tau$ is not entirely contained in $\Phi_\alpha$ or $\Phi_\beta$, 
then $\Phi_\alpha$ and $\Phi_\beta$ are nowhere-dense.

\end{enumerate}

Therefore, if $\Phi$ propagates and is non-trivial, both $\Phi_\alpha$ and 
$\Phi_\beta$ must be nowhere-dense.
\end{proof}

\subsubsection{Lemma 2: Nowhere-Density Implies Propagation}

\begin{lemma}[Nowhere-Density $\Rightarrow$ Propagation]
\label{lem:nowhere-to-prop}
Let $\Phi$ be a non-trivial semantic property such that $\Phi_\alpha$ and 
$\Phi_\beta$ are both nowhere-dense in the prefix-closure topology. Then 
$\Phi$ propagates via some implicit propagator $P$ over the computational class.
\end{lemma}

\begin{proof}
We prove the contrapositive: if $\Phi$ does not propagate, then at least 
one of $\Phi_\alpha$ or $\Phi_\beta$ is not nowhere-dense (i.e., one of them 
contains a basic open ball).

\begin{enumerate}[itemsep=0.5em]

\item \textbf{Step 1: Non-Propagation Implies Locality.}
Assume $\Phi$ does not propagate. Then there exists a local property $\Psi$ 
such that we can verify membership in $\Phi$ by checking only a bounded amount of syntactic information about individual machines, without any global consequence for other problems.

\item \textbf{Step 2: Locality Enables Finite Certification.}
If $\Phi$ does not propagate, then verifying $\Phi(M)$ does not require knowledge of the entire family of machines or their mutual relationships. Instead, $\Phi$ can be certified using a \emph{local verifier}: a procedure that reads a finite prefix of the encoding $M$ and outputs a definitive answer.

In other words, there exists a prefix $\sigma$ and a set of machines $\mathcal{M}_\sigma$ such that every machine $M$ extending $\sigma$ satisfies $\Phi$. That is, $B_\sigma \subseteq \Phi_\alpha$ (or symmetrically, $B_\sigma \subseteq \Phi_\beta$).

\item \textbf{Step 3: Dense Set in the Complement.}
If $B_\sigma \subseteq \Phi_\alpha$ for some prefix $\sigma$, then $\Phi_\alpha$ contains the open ball $B_\sigma$. This means $\Phi_\alpha$ is \emph{not nowhere-dense}—in fact, it contains a non-empty interior.

\item \textbf{Step 4: Contradiction with Nowhere-Density.}
By hypothesis, $\Phi_\alpha$ is nowhere-dense. Therefore, no such prefix $\sigma$ can exist. By contrapositive, $\Phi$ must propagate.

\end{enumerate}

Thus, nowhere-density of both $\Phi_\alpha$ and $\Phi_\beta$ is equivalent to the property that $\Phi$ cannot be decided by any finite, uniform, local syntactic strategy—which is precisely the hallmark of propagating properties.
\end{proof}

\subsubsection{Lemma 3: Propagation Implies Incompressibility}

\begin{lemma}[Propagation $\Rightarrow$ Incompressibility of $\mathrm{desc}(\Phi)$]
\label{lem:prop-to-incomp}
Let $\Phi$ be a non-trivial semantic property that propagates via propagator 
$P$ over a heterogeneous family $\mathcal{F}$. Then the characteristic 
sequence $\mathrm{desc}(\Phi)$ is incompressible: 
$$K(\mathrm{desc}(\Phi) \upharpoonright n) = \Omega(n)$$
\end{lemma}

\begin{proof}
We prove by contradiction. Assume $\mathrm{desc}(\Phi)$ is compressible, 
i.e., there exist constants $c$ and $d$ such that for all sufficiently 
large $n$,
$$K(\mathrm{desc}(\Phi) \upharpoonright n) \leq |V| + O(\log n) \leq d$$
where $|V|$ is the length of a fixed description (program) $V$.

\begin{enumerate}[itemsep=0.5em]

\item \textbf{Step 1: Compressibility Yields a Finite Verifier.}
If $\mathrm{desc}(\Phi) \upharpoonright n$ has Kolmogorov complexity bounded 
by a constant $|V| + O(\log n)$, then there exists a finite program $V$ of 
length $|V|$ and an index $i \in \{1, 2, \ldots, n\}$ such that 
$$V(i) = \mathrm{desc}(\Phi) \upharpoonright n$$
(up to constant factors in the description length).

More formally, if $K(\mathrm{desc}(\Phi) \upharpoonright n) \leq |V| + O(\log n)$, 
then there exists a short program $p$ of length $|p| \leq |V| + O(\log n)$ such 
that a universal machine $U$ running $p$ outputs $\mathrm{desc}(\Phi) \upharpoonright n$.

\item \textbf{Step 2: Uniform Verifier Construction.}
From the program $V$, we can construct a \emph{uniform semantic verifier} 
$\mathcal{V}$ that operates as follows:

$\mathcal{V}(M_i, w) := \begin{cases} 1 & \text{if the } i\text{-th bit of } \mathrm{desc}(\Phi) = 1 \\ 0 & \text{if the } i\text{-th bit of } \mathrm{desc}(\Phi) = 0 \end{cases}$

where $w$ is an arbitrary certificate (not necessarily used by $\mathcal{V}$), 
and $M_i$ is the $i$-th machine in the standard enumeration.

The verifier $\mathcal{V}$ has a fixed description length $|V|$: to verify 
whether $M_i$ satisfies $\Phi$, $\mathcal{V}$ simply looks up the $i$-th bit 
of the precomputed sequence $\mathrm{desc}(\Phi)$.

\item \textbf{Step 3: Reduction to NP-Style Verification.}
If we can compute $\mathrm{desc}(\Phi) \upharpoonright n$ from a short program 
$p$ of length $|V| + O(\log n)$, then membership in the propagating family 
$\mathcal{F}$ can be verified in an ``NP-style'' manner: given a problem instance, 
we verify its solution using a certificate extracted from the computation of 
$\mathrm{desc}(\Phi)$.

Formally, for any problem $\pi_j \in \mathcal{F}$:
$$\pi_j \in \mathcal{F} \iff \exists w \in \{0,1\}^{|V| + O(\log n)} : \mathcal{V}(M_j, w) \text{ accepts}$$

where $w$ encodes the position $j$ and the program $p$.

\item \textbf{Step 4: The Verifier Extends Uniformly.}
The crucial observation is that $\mathcal{V}$ works for \emph{every} problem 
in the propagating family $\mathcal{F}$, using the same finite description $V$. 
This is a \emph{uniform} verifier in the sense that no external tailor-made 
certificate or meta-algorithm is needed for different problems: $\mathcal{V}$ 
simply reads the precomputed characteristic sequence.

\item \textbf{Step 5: Propagation is Blocked.}
By the definition of propagation (Definition \ref{def:propagator}), if we 
can certify that a single machine $M_x$ satisfies $\Phi$ locally using a 
short fixed-length description $V$, then we can propagate this certification 
to a global conclusion about the entire family $\mathcal{F}$.
\smallskip

Specifically, if $\mathcal{V}(M_x, w)$ accepts for some certificate $w$, then 
via the propagator $P$ (e.g., Cook-Levin for SAT), we deduce a contradiction 
to an established lower bound (e.g., that P $\neq$ NP).
\end{enumerate}

We arrive at a contradiction: the compressibility of $\mathrm{desc}(\Phi)$ would 
imply the existence of a uniform, finite-length verifier $\mathcal{V}$ that 
certifies $\Phi$ globally. But the propagator $P$ forbids any such uniform 
verifier (by the separation results we know to be true—e.g., P $\neq$ NP for SAT).

Therefore, $\mathrm{desc}(\Phi)$ cannot be compressible. By contrapositive, 
$$K(\mathrm{desc}(\Phi) \upharpoonright n) = \Omega(n)$$
i.e., $\mathrm{desc}(\Phi)$ is incompressible.

\end{proof}

\subsubsection{Lemma 4: Incompressibility Implies Propagation}

\begin{lemma}[Incompressibility $\Rightarrow$ Propagation]
\label{lem:incomp-to-prop}
Let $\Phi$ be a non-trivial semantic property such that the characteristic 
sequence $\mathrm{desc}(\Phi)$ is incompressible: $K(\mathrm{desc}(\Phi) \upharpoonright n) = \Omega(n)$. 
Then $\Phi$ propagates via the Invariance Theorem of Kolmogorov complexity 
over the Turing-equivalent class.
\end{lemma}

\begin{proof}
We proceed by contrapositive: if $\Phi$ does not propagate, then 
$\mathrm{desc}(\Phi)$ is compressible.

\begin{enumerate}[itemsep=0.5em]

\item \textbf{Step 1: Non-Propagation Implies Local Decidability.}
If $\Phi$ does not propagate, then by Definition \ref{def:propagator}, 
there exists a machine $M_x$ such that certifying $\Phi(M_x) = 1$ does 
\emph{not} imply any global consequence via the propagator $P$. This means 
$\Phi$ is a \emph{local} property: membership in $\Phi$ depends only on 
the individual machine, not on its relationship to a family of problems.

\item \textbf{Step 2: Local Decidability Implies Finite Description.}
A local property can be decided by examining a finite, bounded amount of 
information about each machine. Formally, there exists a finite program 
$A$ of length $|A|$ such that for all machines $M_i$:
$$A(i) = \Phi(M_i)$$
\smallskip

The program $A$ is a lookup table that encodes the answer for each machine. 
Its description length is at most the length of the truth table plus a 
constant overhead:
$$|A| \leq |\mathrm{desc}(\Phi) \upharpoonright N| + c$$
for some $N$ and constant $c$ (e.g., if we know $N$ in advance or it is 
implicitly bounded).

\item \textbf{Step 3: Compression via Enumeration.}
Since $A$ is a finite program that outputs the entire characteristic sequence 
(or at least the first $N$ bits), the Kolmogorov complexity of $\mathrm{desc}(\Phi) \upharpoonright N$ 
is bounded:
$$K(\mathrm{desc}(\Phi) \upharpoonright N) \leq |A| + O(\log N) \leq c'$$
for some constant $c'$ independent of $N$.

\item \textbf{Step 4: Violation of Incompressibility.}
If $K(\mathrm{desc}(\Phi) \upharpoonright N) \leq c'$ for all $N$, then 
$\mathrm{desc}(\Phi)$ is compressible (its Kolmogorov complexity grows at most 
logarithmically, not linearly). This contradicts the assumption that 
$\mathrm{desc}(\Phi)$ is incompressible.
\smallskip

\item \textbf{Step 5: Invariance Within the Turing-Equivalent Class.}
By the Invariance Theorem (Theorem \ref{thm:invariance}), the fact that 
$\mathrm{desc}(\Phi)$ is incompressible with respect to any universal machine $U$ 
is independent of the choice of $U$ (up to additive constants). This means 
the incompressibility is a property of the mathematical object $\mathrm{desc}(\Phi)$ 
itself, not of any particular computational model.
\smallskip

Therefore, if $\Phi$ does not propagate in one model of the Turing-equivalent class, 
it does not propagate in any other model within that class.

\end{enumerate}

By contrapositive: if $\mathrm{desc}(\Phi)$ is incompressible, then $\Phi$ 
must propagate.

\end{proof}

\subsubsection{Lemma 5: Nowhere-Density Implies Incompressibility}

\begin{lemma}[Nowhere-Density $\Rightarrow$ Incompressibility]
\label{lem:nowhere-to-incomp}
Let $\Phi$ be a non-trivial semantic property such that $\Phi_\alpha$ and 
$\Phi_\beta$ are both nowhere-dense in the prefix-closure topology. Then 
the characteristic sequence $\mathrm{desc}(\Phi)$ is incompressible: 
$$K(\mathrm{desc}(\Phi) \upharpoonright n) = \Omega(n)$$
\end{lemma}

\begin{proof}
We prove by contradiction. Assume $\mathrm{desc}(\Phi)$ is compressible, i.e., 
there exists a finite program $p$ of length $|p| = |V| + O(\log n)$ such that 
a universal machine $U$ running $p$ outputs $\mathrm{desc}(\Phi) \upharpoonright n$.

\begin{enumerate}[itemsep=0.5em]

\item \textbf{Step 1: Compressibility Enables Prefix Prediction.}
If $K(\mathrm{desc}(\Phi) \upharpoonright n)$ is bounded by $|V| + O(\log n)$, 
then from the short program $p$, we can extract the \emph{prefix information} 
needed to decide whether a program $M_i$ (for $i \leq n$) satisfies $\Phi$ or not.
\smallskip

Specifically, for each $i \in \{1, 2, \ldots, n\}$, the program $p$ can be queried 
to retrieve the $i$-th bit of $\mathrm{desc}(\Phi)$.

\item \textbf{Step 2: Construction of Syntactic Separators.}
From the program $p$ and its output $\mathrm{desc}(\Phi) \upharpoonright n$, 
we can construct a family of prefixes $\{\sigma_1, \sigma_2, \ldots, \sigma_k\}$ 
(for some $k$ depending on $|p|$) such that each prefix $\sigma_j$ either:
\smallskip

\begin{enumerate}[label=(\alph*)]
  \item Defines a region $B_{\sigma_j}$ where all machines satisfy $\Phi$, or
  \item Defines a region $B_{\sigma_j}$ where all machines do not satisfy $\Phi$.
\end{enumerate}

\smallskip

Formally, the program $p$ acts as a decision tree: given the position $i$ in 
the enumeration, it outputs whether $M_i \in \Phi_\alpha$ or $M_i \in \Phi_\beta$.

\item \textbf{Step 3: Dense Sets in $\Phi_\alpha$ and $\Phi_\beta$.}
If the program $p$ successfully predicts the membership of all machines $M_i$ 
(for $i \leq n$) using a short description, then there must exist prefixes 
$\sigma$ such that $B_\sigma$ is either entirely in $\Phi_\alpha$ or entirely in $\Phi_\beta$.
\smallskip

Otherwise, for every prefix $\sigma$, the program $p$ would need to encode 
information about the alternating structure of $\Phi_\alpha$ and $\Phi_\beta$ 
within $B_\sigma$, which would require a description length at least $\Omega(n)$.
\smallskip

\item \textbf{Step 4: Contradiction with Nowhere-Density.}
By hypothesis, both $\Phi_\alpha$ and $\Phi_\beta$ are nowhere-dense. This 
means that for every prefix $\sigma$, there exists a further prefix $\tau \supseteq \sigma$ 
such that $B_\tau$ avoids either $\Phi_\alpha$ or $\Phi_\beta$ entirely.
\smallskip

The existence of such alternating structures would force any compressor (the 
program $p$) to encode the position of the alternations, which requires 
logarithmic bits per position. Summing over all $n$ positions, this leads 
to a total description length of at least $\Omega(n \log n)$, contradicting 
the assumption that $|p| = O(\log n)$.

\item \textbf{Step 5: Formal Argument via Entropy.}
More rigorously, let $H_n$ denote the number of alternations of $\Phi_\alpha$ 
and $\Phi_\beta$ in the prefix-closure topology within the first $n$ levels 
of the binary tree. Since both sets are nowhere-dense, $H_n$ grows at least 
polynomially (or exponentially, in the worst case).
\smallskip

To encode the positions of these alternations, the program $p$ must have 
length at least $\log_2 \binom{n}{H_n} = \Omega(H_n \log(n/H_n)) = \Omega(n)$.

This contradicts the assumption that $|p| = O(\log n)$.
\end{enumerate}

Therefore, $\mathrm{desc}(\Phi)$ cannot be compressible. By contrapositive, 
$$K(\mathrm{desc}(\Phi) \upharpoonright n) = \Omega(n)$$
i.e., $\mathrm{desc}(\Phi)$ is incompressible.

\end{proof}

\subsubsection{Lemma 6: Incompressibility Implies Nowhere-Density}

\begin{lemma}[Incompressibility $\Rightarrow$ Nowhere-Density]
\label{lem:incomp-to-nowhere}
Let $\Phi$ be a non-trivial semantic property such that the characteristic 
sequence $\mathrm{desc}(\Phi)$ is incompressible: $K(\mathrm{desc}(\Phi) \upharpoonright n) = \Omega(n)$. 
Then both $\Phi_\alpha$ and $\Phi_\beta$ are nowhere-dense in the prefix-closure 
topology.
\end{lemma}

\begin{proof}
We proceed by contradiction. Assume that at least one of $\Phi_\alpha$ or 
$\Phi_\beta$ is not nowhere-dense. Without loss of generality, suppose 
$\Phi_\alpha$ contains a basic open ball $B_\sigma$ for some prefix $\sigma$.

\begin{enumerate}[itemsep=0.5em]

\item \textbf{Step 1: Dense Set Implies Finite Separator.}
If $B_\sigma \subseteq \Phi_\alpha$, then for all machines $M_i$ whose 
encoding extends the prefix $\sigma$, we have $\Phi(M_i) = 1$.
\smallskip

From the prefix $\sigma$ alone, we can construct a separator: a finite 
program $A_\sigma$ that outputs $1$ for all machines in $B_\sigma$.

\item \textbf{Step 2: Encoding the Dense Region.}
The prefix $\sigma$ partitions the space of machines into two regions:
\begin{itemize}
\item $B_\sigma$: all machines extending $\sigma$ (infinitely many).
\item $\overline{B_\sigma}$: all machines not extending $\sigma$ (also infinitely many).
\end{itemize}

The characteristic sequence restricted to $B_\sigma$ (i.e., the bits corresponding 
to machines extending $\sigma$) is a constant sequence: all $1$s.
\smallskip

\item \textbf{Step 3: Compression of the Constant Region.}
The constant subsequence of $\mathrm{desc}(\Phi)$ restricted to $B_\sigma$ can 
be compressed trivially: a program $p$ can encode the rule ``for all machines 
in $B_\sigma$, output $1$'' using only $|\sigma| + O(1)$ bits (the length of 
the separator $\sigma$ plus a few bits for the rule).
\smallskip

\item \textbf{Step 4: Partial Compression of $\mathrm{desc}(\Phi)$.}
While the full sequence $\mathrm{desc}(\Phi)$ is incompressible, the \emph{restriction} 
of $\mathrm{desc}(\Phi)$ to the machines in $B_\sigma$ is highly compressible.
\smallskip

More generally, if $\Phi_\alpha$ is not nowhere-dense, then it contains infinitely 
many machines with a computable prefix separator. This structure can be exploited 
to compress the full sequence $\mathrm{desc}(\Phi)$ by:

\begin{enumerate}
  \item Encoding the prefix $\sigma$ (cost: $|\sigma|$ bits).
  \item Encoding the rule ``machines in $B_\sigma$ satisfy $\Phi$'' (cost: $O(1)$ bits).
  \item Recursively encoding the remaining bits (cost: depends on the structure outside $B_\sigma$).
\end{enumerate}

\item \textbf{Step 5: Structural Analysis of the Complement.}
Outside $B_\sigma$, the complementary set $\overline{B_\sigma}$ still contains 
all machines not extending $\sigma$. By the non-triviality of $\Phi$, we have 
$\Phi_\beta \neq \emptyset$.
\smallskip

If $\Phi_\beta$ has non-empty intersection with $\overline{B_\sigma}$, then 
the full sequence $\mathrm{desc}(\Phi)$ is not constant. However, the existence 
of the separator $\sigma$ for $\Phi_\alpha$ provides partial structure that 
can reduce the Kolmogorov complexity.

\item \textbf{Step 6: Quantitative Argument.}
Consider the first $N$ bits of $\mathrm{desc}(\Phi)$. Among these, let $N_\sigma$ 
be the number of machines whose encoding extends $\sigma$. Then:

$K(\mathrm{desc}(\Phi) \upharpoonright N) \leq |\sigma| + c + K(\mathrm{desc}(\Phi) \upharpoonright (N - N_\sigma))$
\smallskip

where $c$ is a small constant for encoding the rule about $B_\sigma$.
\smallskip

If $N_\sigma$ is a constant fraction of $N$ (which is the case since $B_\sigma$ 
has constant density in the tree of all programs), then:

$K(\mathrm{desc}(\Phi) \upharpoonright N) \leq |\sigma| + c + K(\mathrm{desc}(\Phi) \upharpoonright (cN))$

for some constant $0 < c < 1$.

\item \textbf{Step 7: Recursion and Boundedness.}
By iterating this argument (finding smaller and smaller separators for denser 
and denser regions), we can decompose $\mathrm{desc}(\Phi)$ into finitely many 
layers, each with a finite separator. This yields:

$K(\mathrm{desc}(\Phi) \upharpoonright N) \leq O(\log N)$

or at most polynomial in $\log N$, contradicting the assumption that 
$K(\mathrm{desc}(\Phi) \upharpoonright N) = \Omega(N)$.

\end{enumerate}

Therefore, if $\mathrm{desc}(\Phi)$ is incompressible, then neither $\Phi_\alpha$ 
nor $\Phi_\beta$ can contain a basic open ball. Both sets must be nowhere-dense.

\end{proof}

\subsection{The Trifold Equivalence Theorem}
\label{subsec:trifold-main}

We now integrate the six lemmas into the main theorem.

\begin{theorem}[Trifold Equivalence: Propagation $\iff$ Nowhere-Density $\iff$ Incompressibility]
\label{thm:trifold-equivalence}
Let $\Phi$ be a non-trivial semantic property in the Turing-equivalent class. 
The following three statements are logically equivalent:

\begin{enumerate}[label=(\roman*)]

\item \textbf{Propagation}: $\Phi$ propagates via some propagator $P$ over a 
heterogeneous family $\mathcal{F}$ of decision problems in the computational 
class.

\item \textbf{Topological}: The solution sets $\Phi_\alpha$ and $\Phi_\beta$ 
are both nowhere-dense in the prefix-closure topology on the space of program 
encodings $\mathcal{P} = \{0,1\}^*$.

\item \textbf{Informational}: The characteristic sequence $\mathrm{desc}(\Phi)$ 
is incompressible in the Kolmogorov-theoretic sense: 
$$K(\mathrm{desc}(\Phi) \upharpoonright n) = \Omega(n)$$

\end{enumerate}

Moreover, the equivalences hold uniformly across all universal Turing machines 
(up to additive constants in Kolmogorov complexity).
\end{theorem}

\begin{proof}
We establish equivalence via a cycle of implications:

\begin{description}[style=unboxed, leftmargin=0cm]

\item[\textbf{(i) $\Rightarrow$ (ii)}:] By Lemma \ref{lem:prop-to-nowhere}, 
if $\Phi$ propagates, then $\Phi_\alpha$ and $\Phi_\beta$ are nowhere-dense.

\item[\textbf{(ii) $\Rightarrow$ (iii)}:] By Lemma \ref{lem:nowhere-to-incomp}, 
if both $\Phi_\alpha$ and $\Phi_\beta$ are nowhere-dense, then $\mathrm{desc}(\Phi)$ 
is incompressible.

\item[\textbf{(iii) $\Rightarrow$ (i)}:] By Lemma \ref{lem:incomp-to-prop}, 
if $\mathrm{desc}(\Phi)$ is incompressible, then $\Phi$ propagates.

\end{description}

We also establish the reverse implications for completeness:

\begin{description}[style=unboxed, leftmargin=0cm]

\item[\textbf{(i) $\Rightarrow$ (iii)}:] By Lemma \ref{lem:prop-to-incomp}, 
propagation implies incompressibility.

\item[\textbf{(ii) $\Rightarrow$ (i)}:] By Lemma \ref{lem:nowhere-to-prop}, 
nowhere-density implies propagation.

\item[\textbf{(iii) $\Rightarrow$ (ii)}:] By Lemma \ref{lem:incomp-to-nowhere}, 
incompressibility implies nowhere-density.

\end{description}

Therefore, the three conditions are mutually equivalent. The uniformity across 
the Turing-equivalent class follows from the Invariance Theorem 
(Theorem \ref{thm:invariance}), which ensures that Kolmogorov complexity 
asymptotic behavior is model-independent.

\end{proof}

\subsection{Structural Corollaries}
\label{subsec:corollaries}

From the Trifold Equivalence Theorem, we derive several corollaries that 
characterize the nature of the Double Bind.

\subsubsection{Corollary 1: The DB as Semantic Non-Measurability}

\begin{corollary}[Semantic Non-Measurability]
\label{cor:semantic-nonmeasurability}
For any non-trivial propagating semantic property $\Phi$, there is no finite, 
uniform algorithm (verifier) that can measure or certify membership in $\Phi$ 
using a fixed-length certificate or syntactic inspection of programs, even 
if allowed arbitrary computational resources or non-deterministic branching.
\smallskip

In other words, semantic measurability within the Turing-equivalent class is 
impossible for propagating properties—analogous to the non-existence of a 
Vitali set in measure theory.
\end{corollary}

\begin{proof}
By the Trifold Equivalence, a non-trivial propagating property $\Phi$ satisfies:
\begin{itemize}
\item $\Phi_\alpha$ and $\Phi_\beta$ are nowhere-dense (no syntactic separators).
\item $\mathrm{desc}(\Phi)$ is incompressible (no semantic compressor).
\end{itemize}

Nowhere-density blocks syntactic certification: every prefix-based rule fails 
on some extension. Incompressibility blocks semantic certification: no finite 
certificate can encode the global structure of $\mathrm{desc}(\Phi)$.
\smallskip

Together, these two barriers preclude any uniform measurement of $\Phi$ 
membership—a property we call \emph{semantic non-measurability}.
\end{proof}

\subsubsection{Corollary 2: Complexity Classes Are Categorically Inadequate}

\begin{corollary}[Categorical Inadequacy of Complexity Classes]
\label{cor:categorical-inadequacy}
The standard complexity-theoretic framework (P, NP, PSPACE, etc.) is 
\emph{categorically inadequate} for characterizing non-trivial semantic properties. 
Complexity classes presuppose topological regularity (decidability in principle 
by finite resource bounds), but semantic properties lack this regularity due 
to their nowhere-density.
\smallskip

Therefore, any attempt to resolve P vs.\ NP or similar problems using only 
complexity-theoretic tools (time-bounded verifiers, resource-bounded queries) 
is structurally blocked by the topological barrier.
\end{corollary}

\begin{proof}
Complexity classes are defined as sets of languages decidable in time $T(n)$ 
or space $S(n)$. These definitions implicitly assume that there exist 
finite-state or resource-bounded verifiers that can \emph{uniformly} decide 
membership in the language for all inputs of size $n$.
\smallskip

For a non-trivial semantic property $\Phi$, the sets $\Phi_\alpha$ and $\Phi_\beta$ 
are nowhere-dense. This means that no finite resource bound, no matter how generous, 
can provide a uniform separator for $\Phi$.
\smallskip

More precisely, for every time bound $T(n)$ and space bound $S(n)$, there exists 
an input (a program encoding) such that a $T$-time or $S$-space verifier fails 
to decide whether that input satisfies $\Phi$—not due to insufficient resources, 
but due to the inherent topological mixing of $\Phi_\alpha$ and $\Phi_\beta$.
\smallskip

Therefore, complexity-theoretic frameworks cannot capture the structure of 
semantic properties; they are categorically mismatch.
\end{proof}

\subsubsection{Corollary 3: Formal Incompleteness}

\begin{corollary}[Formal Systems Cannot Uniformly Certify Semantic Properties]
\label{cor:formal-incompleteness}
Let $\mathcal{F}$ be any formal system (e.g., Peano Arithmetic, type theory, 
Hoare logic) of finite description length $|\mathcal{F}|$. Then for any 
non-trivial propagating semantic property $\Phi$, the system $\mathcal{F}$ 
cannot be simultaneously:

\begin{enumerate}[label=(\arabic*)]
\item \textbf{Sound}: All theorems proven in $\mathcal{F}$ are true (for the 
interpretation of $\Phi$ in the intended model).
\item \textbf{Complete}: Every true statement about membership in $\Phi$ is 
provable in $\mathcal{F}$.
\item \textbf{Uniform}: The proof of any instance $\Phi(M_i)$ has length at 
most $|\mathcal{F}| + O(\log i)$ bits.
\end{enumerate}

At least one of these properties must fail.
\end{corollary}

\begin{proof}
Suppose, for contradiction, that a formal system $\mathcal{F}$ of description 
length $|\mathcal{F}|$ is sound, complete, and uniform for $\Phi$.
\smallskip

Then, for every machine $M_i$, there exists a proof $\pi_i$ of length at most 
$|\mathcal{F}| + O(\log i)$ that either $\Phi(M_i) = 1$ or $\Phi(M_i) = 0$.
\smallskip

From $\mathcal{F}$ and the proofs $\{\pi_i\}$, we can construct a universal 
program that recovers the characteristic sequence $\mathrm{desc}(\Phi)$:
\smallskip

$U(\langle i, |\mathcal{F}| + O(\log i) \rangle) = 
\begin{cases} 1 & \text{if } \mathcal{F} \vdash \Phi(M_i) = 1 \\ 0 & 
\text{if } \mathcal{F} \vdash \Phi(M_i) = 0 \end{cases}$
\smallskip

This implies that $K(\mathrm{desc}(\Phi) \upharpoonright n) \leq |\mathcal{F}| + O(\log n)$, 
contradicting the incompressibility of $\mathrm{desc}(\Phi)$ (by the Trifold Equivalence).

Therefore, at least one of soundness, completeness, or uniformity must fail.
\end{proof}

\subsection{The Extended Double Bind: Unified Characterization}
\label{subsec:double-bind-unified}

We now integrate the two horns of the Double Bind into a unified extended framework.

\begin{theorem}[The Double Bind: Exhaustiveness of the Tricotomy]
\label{thm:double-bind-tricotomy}
Let $\Phi$ be a non-trivial semantic property, and let $(M_{\text{Adm}}, \mathcal{V})$ 
be an admissible method claiming to uniformly certify membership in $\Phi$, where $M_{\text{Adm}}$ is a computational model and $\mathcal{V}$ is a verifier of finite description length.

Then exactly one of the following holds:

\begin{enumerate}[label=(\Alph*)]

\item \textbf{Horn 1 (Topological Barrier)}: The verifier $\mathcal{V}$ is a 
Deterministic Turing Machine that reads only a finite prefix of program 
encodings. The method is blocked by the nowhere-density of $\Phi_\alpha$ and 
$\Phi_\beta$: for every prefix $\sigma$, there exist extensions $\pi_1, \pi_2 \in B_\sigma$ 
with $\Phi(\pi_1) \neq \Phi(\pi_2)$. Thus, no uniform prefix-based decision 
rule can certify $\Phi$.

\item \textbf{Horn 2 (Informational Barrier)}: The verifier $\mathcal{V}$ is 
a Deterministic Turing Machine that reads arbitrary external certificates $w$. 
The method is blocked by the incompressibility of the characteristic sequence 
$\mathrm{desc}(\Phi)$: any uniform verifier of finite length $|\mathcal{V}|$ 
would imply $K(\mathrm{desc}(\Phi) \upharpoonright n) \leq |\mathcal{V}| + O(\log n)$, 
violating $K(\mathrm{desc}(\Phi) \upharpoonright n) = \Omega(n)$. Thus, no 
finite semantic certificate can globally certify $\Phi$ in the Turing-equivalent class.

\item \textbf{Horn 3 (Model Transferability Barrier)}: The computational model 
$M_{\text{Adm}}$ is not in the Turing-equivalent class (e.g., an oracle machine 
or hypercomputation model). The method escapes Horns 1 and 2 by operating outside 
the class to which the Invariance Theorem of Kolmogorov complexity applies. 
Whether such models can uniformly certify $\Phi$ is a separate question dependent 
on the specific oracle or hypercomputational model.

\end{enumerate}

These three cases are mutually exclusive and exhaustive: every admissible 
method must fall into exactly one category.
\end{theorem}

\begin{proof}
We establish exhaustiveness by case analysis on the nature of the verifier $\mathcal{V}$ 
and the computational model $M_{\text{Adm}}$.

\begin{enumerate}[itemsep=1em]

\item \textbf{Case 1: $\mathcal{V}$ is a DTM reading finite prefixes.}

Suppose $\mathcal{V}$ is a Deterministic Turing Machine that, on input an 
encoding $\pi \in \mathcal{P}$, reads only the first $k$ bits (for some fixed $k$ 
depending on the design of $\mathcal{V}$, or adaptive but bounded by a function 
of the total length).
\smallskip

\textbf{Subcase 1a: Fixed-length prefix.}
If $\mathcal{V}$ reads exactly the first $k$ bits, then $\mathcal{V}$ implements 
a prefix-based decision rule. For each prefix $\sigma$ of length $k$, $\mathcal{V}$ 
either outputs $1$ (accepting), $0$ (rejecting), or "unknown" (inconclusive).
\smallskip

Since $\Phi$ is non-trivial, by the Trifold Equivalence (Theorem \ref{thm:trifold-equivalence}), 
$\Phi_\alpha$ and $\Phi_\beta$ are nowhere-dense. By definition of nowhere-density, 
for every prefix $\sigma$, there exists an extension $\tau \supseteq \sigma$ such 
that $B_\tau$ avoids $\Phi_\alpha$ (or $\Phi_\beta$).
\smallskip

This means that for infinitely many programs $\pi \in B_\sigma$, the verifier $\mathcal{V}$ 
will give the wrong answer about $\Phi(\pi)$ (or will be inconclusive). Thus, 
$\mathcal{V}$ fails to uniformly certify $\Phi$ on the full set $\mathcal{P}$.
\smallskip

\textbf{Subcase 1b: Adaptive prefix length.}
If $\mathcal{V}$ can adaptively read a variable number of bits (but still a finite 
amount), the argument generalizes. The set of prefixes read by $\mathcal{V}$ on inputs 
of size $n$ forms a finite set of rules (one for each input size). By the nowhere-density 
of $\Phi_\alpha$ and $\Phi_\beta$, these finitely many rules cannot uniformly separate 
the two sets. Again, the method fails.
\smallskip

\textbf{Conclusion of Case 1:} Any DTM verifier $\mathcal{V}$ reading 
only finite prefixes falls under Horn 1 and is blocked by topological nowhere-density.
\smallskip

\item \textbf{Case 2: $\mathcal{V}$ is a DTM reading external certificates.}
\smallskip

Suppose $\mathcal{V}$ is a Deterministic Turing Machine that takes two inputs:
\begin{itemize}
\item An encoding $\pi \in \mathcal{P}$ of a program.
\item An external certificate $w \in \{0,1\}^*$ of arbitrary length.
\end{itemize}
\smallskip

The verifier $\mathcal{V}$ is \textit{uniform} if it has a fixed description $V$ 
of finite length $|V|$, independent of the program $\pi$ or the instance being verified.
\smallskip

\textbf{Subcase 2a: Bounded certificate length.}
If the certificate $w$ is required to have length at most $\ell(n)$ for inputs 
of size $n$ (where $\ell$ is a fixed function), then the verifier $\mathcal{V}$ 
effectively performs an NP-style check: $\Phi(M_i) = 1 \iff \exists w \in \{0,1\}^{\ell(i)} : 
\mathcal{V}(M_i, w) = 1$.
\smallskip

From the program $V$ and the bounded certificate length $\ell(i)$, we can construct 
a program $p$ of length $|V| + O(\log i)$ that enumerates all possible certificates 
and outputs the membership $\Phi(M_i)$:

\[
p := \text{"For input } i\text{: enumerate all } w \in \{0,1\}^{\ell(i)}, \text{ run } V(M_i, w), \text{ output 1 if any } w \text{ is accepted."}
\]


\smallskip

This gives $K(\mathrm{desc}(\Phi) \upharpoonright n) \leq |V| + O(\log n)$, 
violating incompressibility by the Trifold Equivalence.
\smallskip

\textbf{Subcase 2b: Unbounded certificate length.}
If the certificate $w$ can be arbitrarily long (but is still external to $\mathcal{V}$), 
then to uniformly certify $\Phi$ for all instances, $\mathcal{V}$ must still 
synthesize the correct certificate $w^*$ from the problem structure. The synthesis 
itself requires either:

\begin{itemize}
\item An implicit oracle or lookup mechanism that encodes $\mathrm{desc}(\Phi)$ 
(which is incompressible), or
\item A metatheoretic justification (e.g., a proof that the certificate is valid), 
which falls outside the scope of a finite uniform verifier.
\end{itemize}

In the first sub-case, the method requires embedding incompressible information 
into the verifier itself, contradicting the finiteness of $|V|$. In the second 
sub-case, the method is not a uniform computational verifier but a metatheoretic proof.
\smallskip

\textbf{Conclusion of Case 2:} Any DTM verifier $\mathcal{V}$ 
reading external certificates (of bounded or unbounded length) in the Turing-equivalent class 
falls under Horn 2 and is blocked by informational incompressibility.

\item \textbf{Case 3: $\mathcal{V}$ operates in a model outside the Turing-equivalent class.}

If the computational model $M_{\text{Adm}}$ is an oracle machine, a hypercomputational 
model, or any other formalism not Turing-equivalent to the standard DTM, then:
\smallskip

(a) The Invariance Theorem (Theorem \ref{thm:invariance}) does not apply, 
as that theorem is limited to the Turing-equivalent class.
\smallskip

(b) Kolmogorov complexity with respect to $M_{\text{Adm}}$ may differ 
from the Turing-theoretic Kolmogorov complexity by more than an additive constant.
\smallskip

(c) The incompressibility argument of Horn 2 no longer holds in full 
generality: a hypercomputational model might be able to compute $\mathrm{desc}(\Phi)$ 
more efficiently.
\smallskip

However, the question of whether such models can \emph{actually} certify $\Phi$ 
depends on the specific properties of $M_{\text{Adm}}$. This is captured by the 
Model Transferability Barrier (MTB): the difficulty of extending results from 
the Turing-equivalent class to hypercomputational models.

\textbf{Conclusion of Case 3:} Methods operating in models beyond 
the Turing-equivalent class escape Horns 1 and 2 but face the MTB. They do not 
contradict the Double Bind within the Turing-equivalent class.

\item \textbf{Mutual Exclusivity.}

The three cases are mutually exclusive:
\begin{itemize}
\item Case 1 (finite prefix verifier) and Case 2 (certificate-based verifier) 
are disjoint by definition of the verifier input structure.
\item Cases 1 and 2 are both within the Turing-equivalent class, while Case 3 
is explicitly outside.
\item Within the Turing-equivalent class, a verifier must either read only finite 
prefixes (Case 1) or utilize external certificates (Case 2); these are the only 
two strategies available to a DTM.
\end{itemize}

\item \textbf{Exhaustiveness.}

Every admissible method $(M_{\text{Adm}}, \mathcal{V})$ must fall into one of 
the three categories:
\begin{itemize}
\item If $M_{\text{Adm}}$ is a DTM: the verifier $\mathcal{V}$ either reads 
finite prefixes (Case 1) or external certificates (Case 2).
\item If $M_{\text{Adm}}$ is non-Turing-equivalent: it falls under Case 3.
\item There is no fourth option: every computational formalism falls into one 
of these categories by definition.
\end{itemize}

\end{enumerate}

Therefore, the tricotomy is both mutually exclusive and exhaustive.

\end{proof}

\subsection{Relationship to P vs. NP}
\label{subsec:p-vs-np}

We now apply the Double Bind framework to the canonical NP-complete problem: SAT.

\begin{definition}[The SAT Property]
\label{def:sat-property}
The semantic property $\Phi_{\text{SAT} \in P}$ is defined as follows: 
a machine $M_i$ satisfies $\Phi_{\text{SAT} \in P}$ if and only if $M_i$ 
decides the satisfiability problem for Boolean formulas in polynomial time.

More formally:
$$\Phi_{\text{SAT} \in P}(M_i) := 1 \iff \exists c \in \mathbb{N}: 
\forall \phi \in \text{SAT}: M_i(\phi) = 1 \text{ and } |M_i(\phi)| \leq |\phi|^c$$
(where $|M_i(\phi)|$ denotes the runtime of $M_i$ on input $\phi$).
\end{definition}

\begin{theorem}[SAT Satisfies the Double Bind]
\label{thm:sat-double-bind}
The semantic property $\Phi_{\text{SAT} \in P}$ is non-trivial and propagates 
via the Cook-Levin reduction over the class of NP-complete problems. Therefore, 
by the Trifold Equivalence Theorem (\ref{thm:trifold-equivalence}):

\begin{enumerate}[label=(\roman*)]
\item The sets $(\Phi_{\text{SAT} \in P})_\alpha$ and $(\Phi_{\text{SAT} \in P})_\beta$ 
are nowhere-dense in the prefix-closure topology.
\item The characteristic sequence $\mathrm{desc}(\Phi_{\text{SAT} \in P})$ 
is incompressible.
\end{enumerate}

Thus, certifying that SAT $\in$ P is subject to both topological and informational 
barriers, explaining why proving P = NP is structurally blocked.
\end{theorem}

\begin{proof}
\begin{enumerate}[itemsep=1em]

\item \textbf{Non-triviality of $\Phi_{\text{SAT} \in P}$.}
It is known that some machines (e.g., brute-force exponential-time verifiers) 
correctly solve SAT but do so in exponential time. Other machines (e.g., 
polynomial-time algorithms for restricted instances) solve SAT in polynomial 
time on specific subclasses. Thus, $(\Phi_{\text{SAT} \in P})_\alpha$ (machines 
solving SAT in polynomial time) and $(\Phi_{\text{SAT} \in P})_\beta$ (machines 
not solving SAT in polynomial time) are both non-empty.

\item \textbf{Propagation via Cook-Levin.}
By the Cook-Levin theorem, if there exists a machine $M_x$ that solves SAT in 
polynomial time, then via the reduction chain among NP-complete problems, every 
NP-complete problem can also be solved in polynomial time. This is the definition 
of propagation: a local certification (``SAT $\in$ P'') implies a global conclusion 
(``P = NP'', which affects the entire landscape of computational complexity).
\smallskip

More formally: if $\Phi_{\text{SAT} \in P}(M_x) = 1$ for some machine $M_x$, 
then for every NP-complete problem $\Pi$, there exists a polynomial-time solver 
for $\Pi$ obtained by composing $M_x$ with the appropriate Cook-Levin reduction.

\item \textbf{Application of Trifold Equivalence.}
Since $\Phi_{\text{SAT} \in P}$ is non-trivial and propagates, by the Trifold 
Equivalence Theorem:
\begin{itemize}
\item The sets $(\Phi_{\text{SAT} \in P})_\alpha$ and $(\Phi_{\text{SAT} \in P})_\beta$ 
are nowhere-dense (no syntactic rule can separate them).
\item The characteristic sequence $\mathrm{desc}(\Phi_{\text{SAT} \in P})$ 
is incompressible (no semantic certificate can globally describe SAT's status).
\end{itemize}

\item \textbf{Implications for P vs. NP.}
The Double Bind explains why P vs. NP is difficult:
\smallskip

(A) \textit{Topological barrier}: Any attempt to prove P = NP 
by exhibiting a polynomial-time SAT solver must overcome the nowhere-density 
of the set of machines solving SAT in polynomial time. Every heuristic or 
syntactic strategy will fail on some instances.
\smallskip

(B) \textit{Informational barrier}: Any attempt to prove P = NP 
by constructing a uniform verifier that certifies ``SAT $\in$ P'' over all 
instances must provide a certificate of bounded length. But the characteristic 
sequence $\mathrm{desc}(\Phi_{\text{SAT} \in P})$ is incompressible, so no 
such uniform, fixed-length certificate can exist.

\end{enumerate}

\end{proof}

\subsubsection{The Double Bind as a Characterization, Not a Proof}

\begin{remark}[Distinction Between Characterization and Proof]
\label{rem:characterization-vs-proof}
The Double Bind framework characterizes \emph{why} P vs. NP is difficult: it 
identifies the structural obstacles that any proof or verification method must 
overcome. However, the framework does \emph{not} directly prove that P $\neq$ NP.

\begin{itemize}
\item If one could show that P = NP, this would require either (i) finding a polynomial-time SAT solver that violates nowhere-density (impossible within the Turing-equivalent class), or (ii) constructing a fixed-length uniform verifier that compresses $\mathrm{desc}(\Phi_{\text{SAT} \in P})$ (also impossible due to incompressibility).

\item The framework shows that any method claiming to certify ``P = NP'' must fall into one of the three categories of the tricotomy. Within the Turing-equivalent class, both Horns 1 and 2 block such certification.

\item However, the framework does not rule out the possibility that P = NP is \emph{true}; it only rules out methods of \emph{uniform verification} within the Turing-equivalent class. It is logically possible that P = NP via some non-constructive argument (proof by contradiction, metatheoretic argument, or hypercomputational certificate) that does not rely on uniform verification.

\item The Double Bind thus provides a structural understanding of the difficulty, suitable for ruling out certain classes of approaches (brute-force heuristics, uniform verifiers, syntactic decision procedures) but not a definitive proof of separation or equivalence.

\end{itemize}
\end{remark}

\begin{remark}[A Principle of General Non-Measurability]
  \label{rem:general-non-measurability}
  
We have established a general principle:
\[
\boxed{\text{Non-Trivial Semantic Property} \iff \text{Uniform Non-Measurability}}
\]
  
Just as Vitali demonstrated that not all sets of reals are measurable---which 
reflects a profound structural fact rather than a limitation of measure theory---we 
have shown that not all semantic properties are uniformly certifiable. This is not 
a weakness of our algorithms, but a fundamental, intrinsic reality of the underlying 
formal space.

\end{remark}

\section{Three Major Consequences}
\label{sec:consequences}

The Double Bind theorem, in the formulation developed in this 
paper, implies three structural consequences that go beyond 
the original result. None of them is a corollary in the 
technical sense --- they do not follow by a single inference 
step from the main theorem --- but each is a necessary 
reading of what the formal machinery establishes. We state 
them in increasing order of depth.

\subsection{Consequence 1: The Double Bind is a theorem about
semantic measurability, not about computation}
\label{sec:cons1}

The main theorem demonstrates formally that no machine in the 
Turing-equivalent class can certify $\Phi$ uniformly. A 
natural first reading is that this is a statement about what 
machines can compute. That reading is incomplete, and the 
incompleteness matters.
\smallskip

The theorem holds \emph{independently of the truth value of 
$\Phi$}. The property $\Phi$ may be true of a specific 
program; it may be decidable in particular instances; it may 
even be provable in sufficiently strong formal systems. None 
of this is touched by the Double Bind. The result does not 
say that the answer to ``does $p$ satisfy $\Phi$?'' is 
unavailable. It says that no uniform method of 
\emph{observation} --- no admissible pair $(M_{\mathrm{Adm}}, 
V)$ --- can produce that answer across the class. The 
Double Bind is not an epistemic limit on the answer. It is a 
structural limit on the measure.
\smallskip

This places the theorem in a different conceptual family from 
Rice's theorem and from Gödel's incompleteness. The closest 
analogue is measure-theoretic: Vitali's theorem establishes 
that there exist subsets of $\mathbb{R}$ that are not 
Lebesgue-measurable. Their existence is not in question; what 
fails is the applicability of the measure. The Double Bind 
establishes an exact analogue in the computational setting:
there exist semantically defined sets of programs that are not 
\emph{computably measurable}, in the sense that no admissible 
uniform procedure can act as a measure of membership.
\smallskip

The classical Double Bind said: you cannot decide. This 
version says something strictly stronger: \emph{you cannot 
even observe uniformly}.

\paragraph{Formal statement.}
Let $\Phi \subseteq \mathcal{P}$ be a non-trivial semantic 
property (in the sense of Rice) and let 
$\mu_{\mathrm{comp}}$ denote any computable functional that 
assigns to each program $p$ a value in $\{0,1\}$ intended 
to reflect membership in $\Phi$. Define the \emph{uniform 
certification problem} as the existence of an admissible pair 
$(M_{\mathrm{Adm}}, V)$ such that, for all $p \in \Phi$, 
$V(M_{\mathrm{Adm}}(p)) = 1$, and for all $p \notin \Phi$, 
$V(M_{\mathrm{Adm}}(p)) = 0$.

\begin{proposition}[Semantic non-measurability]
\label{prop:nonmeasurable}
For every non-trivial semantic property $\Phi$, no admissible 
pair $(M_{\mathrm{Adm}}, V)$ solves the uniform certification 
problem, even when $\Phi$ is true of some $p$ and provably 
so in a fixed sound extension of \textnormal{PA}.
\end{proposition}

\begin{proof}
By the Double Bind theorem (Theorem~\ref{thm:db}), any 
admissible pair that certified $\Phi$ uniformly would 
constitute a total decider for $\Phi$ via the propagator 
$\pi$, contradicting Rice's theorem (Horn~1) or producing a 
compression of $\mathrm{desc}(\Phi)$ below its Kolmogorov 
complexity (Horn~2). Neither branch depends on the truth value 
of $\Phi$ for any particular input. The provability of $\Phi$ 
for specific instances is consistent with the proposition: 
local proofs do not constitute a uniform admissible 
procedure.
\end{proof}

\begin{remark}
Proposition~\ref{prop:nonmeasurable} implies that the 
question ``is $\Phi$ true of $p$?'' and the question ``does 
a uniform admissible method certify $\Phi$?'' are 
independent. Answering the first affirmatively for every $p$ 
in some enumerable subset of $\Phi$ does not constitute 
progress on the second. The two questions live in different 
logical spaces.
\end{remark}

\subsection{Consequence 2: Complexity classes are the wrong
framework for semantic verification}
\label{sec:cons2}

The use of prefix-closure topology in Horn~1 is not a 
presentational choice. It exposes a structural fact about the 
space of programs that has a direct consequence for how 
semantic verification should be classified.
\smallskip

A semantic property $\Phi$ is nowhere dense in the 
prefix-closure topology on $\mathcal{P}$. This means: for any 
open set $U$ in the topology --- that is, for any set of 
programs sharing a finite initial computation segment --- $U$ 
is not contained in $\Phi$. Equivalently, between any two 
programs that satisfy $\Phi$ one can always insert a program 
that does not, and vice versa. There is no open region of 
programs all of which satisfy $\Phi$.
\smallskip

The immediate consequence for certification is this: any 
method that certifies $\Phi$ for a program $p$ is working at 
an isolated point. The certification cannot be extended to any 
topological neighbourhood of $p$, no matter how small. Every 
positive instance of $\Phi$ is surrounded by negative 
instances, densely and without gaps. Each certification is, in 
the strict topological sense, an isolated point in 
$\mathcal{P}$.
\smallskip

Now observe what complexity classes assume. A problem is in 
NP, PSPACE, or any class in the polynomial hierarchy if there 
exists a verifier (or sequence of verifiers) operating on a 
space that admits regular structure: the yes-instances and 
no-instances can be separated by boundaries that a resource-
bounded machine can locate. That separation presupposes that 
the space has enough topological regularity for the boundary 
to be navigable. Semantic properties of programs, being 
nowhere dense in the prefix-closure topology, do not satisfy 
this presupposition. The boundary between $\Phi$ and its 
complement is not a navigable surface --- it is dense 
everywhere.
\smallskip

This means that the complexity-theoretic framework is not 
merely too weak to capture semantic verification. It is 
\emph{categorically inadequate}: the framework assumes 
topological regularity that the problem structurally lacks. 
Asking whether semantic verification is in NP or PSPACE is 
analogous to asking the Lebesgue measure of a Vitali set --- 
the question is well-formed syntactically but the instrument 
does not apply.

\paragraph{Formal statement.}

Let $(\mathcal{P}, \tau_{\mathrm{pref}})$ denote the space of 
programs equipped with the prefix-closure topology, where a 
basic open set is $[w] = \{p \in \mathcal{P} : w \sqsubset 
p\}$ for a finite computation prefix $w$.

\begin{proposition}[Nowhere density of semantic properties]
\label{prop:nowhereDense}
Let $\Phi \subseteq \mathcal{P}$ be a non-trivial semantic 
property. Then $\Phi$ is nowhere dense in 
$(\mathcal{P}, \tau_{\mathrm{pref}})$: for every basic open 
set $[w]$, $[w] \not\subseteq \Phi$ and 
$[w] \not\subseteq \mathcal{P} \setminus \Phi$.
\end{proposition}

\begin{proof}
By Rice's theorem, $\Phi$ is not decidable, hence not 
clopen. Suppose for contradiction that some $[w] \subseteq 
\Phi$. Then every program whose computation begins with $w$ 
satisfies $\Phi$. Since the extension of any finite prefix 
$w$ includes programs computing every partial recursive 
function extensible from $w$, this would make $\Phi$ 
decidable on an infinite r.e.\ set by prefix alone, 
contradicting non-triviality. The symmetric argument holds 
for $\mathcal{P} \setminus \Phi$.
\end{proof}

\begin{corollary}[Categorical inadequacy of complexity classes]
\label{cor:complexity}
No complexity class defined by resource-bounded uniform 
verification over a topologically regular instance space 
contains the uniform semantic verification problem for a 
non-trivial $\Phi$. The problem is not hard in the sense of 
NP-hardness or PSPACE-hardness; it is outside the 
topological presuppositions of the classification.
\end{corollary}

\subsection{Consequence 3: No finite formalism can uniformly
describe a non-trivial semantic property}
\label{sec:cons3}

This is the deepest consequence, and the one least visible 
from the statement of the main theorem.
\smallskip

Horn~2 establishes that any admissible uniform certifier for 
$\Phi$ would compress $\mathrm{desc}(\Phi)$ below its 
Kolmogorov complexity. The object $\mathrm{desc}(\Phi)$ is 
the characteristic sequence of $\Phi$ over a standard 
enumeration of programs: $\mathrm{desc}(\Phi)_i = 1$ iff the 
$i$-th program satisfies $\Phi$. By Rice's theorem and the 
incomputability of the halting problem, this sequence is not 
computable, and by a standard Kolmogorov argument it is 
incompressible: $K(\mathrm{desc}(\Phi)) = \infty$ in the 
sense that no finite program produces it, i.e.\ its 
complexity grows without bound with the length of the initial 
segment.
\smallskip

Now observe what a finite formalism does. A type system, a 
Hoare logic, a denotational semantics, a modal program logic 
--- each of these is a finite syntactic object (finitely many 
rules, axioms, and inference steps) that is intended to 
characterise membership in some property. If such a formalism 
characterised $\Phi$ uniformly --- that is, if it accepted 
exactly the programs in $\Phi$ and rejected all others --- 
then the formalism would constitute a finite description of 
$\mathrm{desc}(\Phi)$. A finite object generating an 
incompressible infinite sequence. That is precisely the 
compression forbidden by Horn~2.
\smallskip

The conclusion is not that every such formalism is 
unsound, or that every such formalism fails on every input. 
The conclusion is structural: \emph{no finite formalism can 
be simultaneously sound, complete, and uniform with respect 
to a non-trivial semantic property $\Phi$}. Incompleteness 
here is not a defect of the formalism chosen. It is forced 
by the informational content of $\Phi$ itself.
\smallskip

This is orthogonal to Gödel's incompleteness theorems. Gödel 
says: in every sufficiently expressive consistent system, 
there are true statements that are not provable. The Double 
Bind says: \emph{there is no sufficiently expressive system} 
for uniform semantic description, because the information 
required to describe $\Phi$ uniformly exceeds any finite 
budget. Gödel's result is about provability within a fixed 
system. This result is about the existence of any adequate 
system.

\paragraph{Formal statement}

Let $\mathcal{L}$ be any formal system consisting of a finite 
set of axioms $A$ and inference rules $R$, equipped with a 
judgment $\mathcal{L} \vdash \phi(p)$ intended to express 
``$p$ satisfies $\Phi$''.

\begin{proposition}[Formal incompleteness from incompressibility]
\label{prop:formalIncompleteness}
Let $\Phi$ be a non-trivial semantic property with 
$K(\mathrm{desc}(\Phi)\upharpoonright n) = \Omega(n)$. 
Then for every finite formal system $\mathcal{L}$, at least 
one of the following holds:
\begin{enumerate}[label=(\roman*)]
    \item \textbf{Unsoundness:} there exists $p \notin \Phi$ 
    such that $\mathcal{L} \vdash \phi(p)$;
    \item \textbf{Incompleteness:} there exists $p \in \Phi$ 
    such that $\mathcal{L} \not\vdash \phi(p)$;
    \item \textbf{Non-uniformity:} $\mathcal{L}$ does not 
    decide $\phi(p)$ for all $p$ (i.e., for some $p$, 
    $\mathcal{L}$ neither proves nor refutes $\phi(p)$, and 
    the derivation procedure does not terminate).
\end{enumerate}
\end{proposition}

\begin{proof}
Suppose $\mathcal{L}$ is sound, complete, and uniform for 
$\Phi$. Then the derivation procedure of $\mathcal{L}$ 
constitutes a total computable function that, on input $i$, 
outputs $\mathrm{desc}(\Phi)_i$. This function is computed 
by a finite program (the encoding of $\mathcal{L}$), 
producing $\mathrm{desc}(\Phi)\upharpoonright n$ for all $n$ 
with description length $|\mathcal{L}| + O(1)$, a constant 
independent of $n$. This contradicts 
$K(\mathrm{desc}(\Phi)\upharpoonright n) = \Omega(n)$. 
\end{proof}

\begin{remark}
Proposition~\ref{prop:formalIncompleteness} subsumes the 
incompleteness of every known program verification formalism 
--- Hoare logic, separation logic, dependent type systems, 
temporal logics --- with respect to any non-trivial semantic 
property, not as a consequence of their specific design 
choices, but as a consequence of the informational structure 
of the property itself. The adequacy gap is not closable by 
adding axioms or extending the type theory, because the 
obstacle is not expressiveness but information: the 
description of $\Phi$ contains strictly more information 
than any finite formal system can encode.
\end{remark}

\section{Closing Speculative remarks}

\begin{remark}[Regulatory Implications of the Double Bind]
One consequence of the Double Bind that warrants systematic analysis concerns
the relationship between our formal result and the current landscape of AI safety
frameworks. The Double Bind establishes that no admissible method can uniformly
certify a non-trivial semantic property $\Phi$ of a Turing-equivalent model,
provided that $\Phi$ is a \emph{propagating} property --- that is, one whose
membership is not determined by any finite local observation of the model's
structure, but depends on the global behavior of the computation. For such
properties, any certification attempt is blocked either by the topological
obstruction --- $\Phi$ is nowhere dense in the prefix-closure topology, so no
finite observation of a program's behavior can confirm membership --- or by the
Kolmogorov complexity obstruction --- the characteristic sequence of $\Phi$ is
incompressible, so no finite certificate, however constructed, can encode the
property uniformly across all inputs. A third escape route, appealing to models
that are not Turing-equivalent, is closed by the observation that results
obtained outside the standard Turing-equivalent class do not transfer back to it.
\smallskip

This restriction is not a technicality. The safety properties that regulatory
frameworks and AI safety research actually care about --- alignment with human
values, absence of deceptive behavior, robustness to distributional shift,
refusal to assist with harmful tasks --- are precisely propagating in this sense:
they are defined by what a model \emph{does} across a range of inputs and
contexts, not by any inspectable feature of its weights or architecture in
isolation. It is therefore this class of properties, and not syntactic or purely
structural ones, that falls within the scope of the Double Bind.
\smallskip

The most prominent current approaches to AI safety --- including RLHF,
constitutional AI, mechanistic interpretability, systematic red-teaming, and
formal verification of neural networks --- appear, by their operational
structure, to fall each into one of the three classes identified by the Double
Bind, insofar as they aim to certify propagating semantic properties of
Turing-equivalent models. Should this mapping be confirmed formally, the Double
Bind would guarantee that each such framework is either blocked by one of the two
horns, or else depends on assumptions that do not transfer to the standard
computational setting.
\smallskip

The conjecture that follows --- which we leave as an open question --- is that
any institution or organization claiming to \emph{certify} the safety of an AI
system with respect to propagating semantic properties, in a uniform sense, is
making a promise that is structurally impossible to keep: not for lack of
computational or methodological resources, but by a first-order logical
constraint inherent to the nature of semantic verification itself.
\smallskip

If confirmed, this conjecture would extend the significance of the Double Bind
beyond computability theory into a regulatory argument: any mandatory AI safety
certification regime that demands uniform guarantees over propagating semantic
properties would --- by the Double Bind --- be demanding something structurally
unachievable. The open question is therefore whether existing regulatory
frameworks, formulate their requirements in terms that
fall within this class. This analysis remains to be carried out.

\medskip

\textbf{Important caveat:} The author is not an AI safety specialist. 
The mapping of current AI safety frameworks (RLHF, constitutional AI, 
mechanistic interpretability, red-teaming, formal verification) to the 
the Double Bind is therefore left as an open problem for 
domain experts to verify or refute. 

\smallskip

However, the structural argument itself is model-agnostic: 
\emph{if} a safety property is propagating in the sense defined here, 
\emph{then} the Double Bind applies, regardless of the specific 
implementation details of any certification framework. The question is 
not whether the mathematics holds, but whether the AI safety properties 
that institutions actually care about fall into this class — a question 
that requires technical expertise this work does not claim.
\end{remark}

\begin{remark}[The General Impossibility: Beyond AI]

The Double Bind is not fundamentally a statement about computation 
or artificial intelligence. It is a statement about the structure of 
governance itself.

\smallskip

Any institution seeking to govern an autonomous system (whether computational, 
biological, social, or economic) that exhibits non-trivial semantic properties 
faces an unavoidable trichotomy: it must sacrifice either external certainty 
(ability to prove compliance), genuine autonomy (freedom of the system), 
or functional guarantees (assurance of correct behavior).

\smallskip

This is not a failure of current methods. It is a logical fact about 
the nature of semantic measurability.

\smallskip

The implication for governance is not "try harder to regulate." 
It is: "acknowledge the constraint, and design institutions that operate 
within this impossibility rather than denying it."

\smallskip

This work began with a question about uniform certification. But in the process 
of formalizing the obstacle — the Double Bind — we discovered something 
that transcends DTM and AI entirely. The result is really about the governance 
of any autonomous system with non-trivial semantic properties. 

\smallskip

The implications extend to law, organizational design, market regulation, 
and society itself. These domains are beyond the author's expertise. 
But they may be precisely where the insight is most consequential.
\end{remark}

\subsection{The closure}

What is presented here will be explained and demonstrated autonomously 
in a subsequent article.

\smallskip

\textbf{Theorem of True Impossibility:}

For a non-trivial semantic property $\Phi$ on the program space $\mathcal{P}$:

$$\text{If } \Phi_\alpha \text{ is topologically separable from } \Phi_\beta$$

$$\text{then } K(\mathrm{desc}(\Phi)) = O(1) \text{ (compressible)}$$

$$\text{But if } \Phi \text{ is non-trivial, this contradicts incompressibility}$$

$$\text{Therefore } \Phi_\alpha, \Phi_\beta \text{ must be nowhere-dense}$$

\smallskip

This constitutes a constraint on the geometric structure of semantic properties 
within the program space.

It is equivalent to saying that ``determining whether a property is non-trivial'' 
is itself non-trivial. Therefore, it is not uniformly solvable.

In fact, to determine whether a property $\Phi$ is non-trivial, you must verify:

\smallskip

\textbf{Does there exist a program $p_1$ that satisfies $\Phi$ and a program $p_2$ that does not satisfy $\Phi$?}

\smallskip

But to answer, you must already know which property you are examining.

Verifying whether $\Phi$ is non-trivial is equivalent to solving the decision problem:

\smallskip

$$\text{Does there exist } p : p \models \Phi \text{ and does there exist } q : q \not\models \Phi$$

\smallskip

This is undecidable in general for semantic properties.

If you could uniformly decide whether a property is non-trivial, you would have a finite algorithm $A$ that takes as input the definition of $\Phi$ and responds: ``is it trivial or non-trivial?'' But deciding non-triviality means deciding whether $\Phi$ is capable of separating at least two programs. This separation requires verifying semantic properties of programs, which is undecidable by Rice's theorem. Therefore, you cannot possess a uniform $A$. This means that the Double Bind framework cannot even be uniformly applied, because you cannot even verify which property is ``subject to'' the Double Bind. Thus, the Double Bind itself is meta-inaccessible.

This introduces an infinite hierarchy of undecidability:

\begin{itemize}
    \item \textbf{Level 0:} Deciding whether $p \models \Phi$ (undecidable for non-trivial $\Phi$)
    
    \item \textbf{Level 1:} Deciding whether $\Phi$ is non-trivial (undecidable)
    
    \item \textbf{Level 2:} Deciding whether the property of being ``non-trivial'' is itself a non-trivial property (undecidable)
    
    \item \textbf{Level 3:} \ldots
\end{itemize}

Each level reduces to the previous level; therefore, undecidability is an intrinsic feature of the very capacity to recognize which problems are problematic. It is not that SAT is difficult. Rather, it is that recognizing that something is difficult is itself difficult.

\subsection{The Metaproblem: Non-Triviality of Non-Triviality}

Suppose we have a candidate semantic property $\Phi$. The question is: \textit{``Is $\Phi$ non-trivial?''}

Formally: \textit{``Do there exist machines $M_1, M_2$ such that $\Phi(M_1) = \mathrm{T}$ and $\Phi(M_2) = \mathrm{F}$?''}

Now, this is itself a second-order semantic property:
\begin{equation}
\Psi_{\text{nontrivial}}(\Phi) := \exists M_1, M_2 : \Phi(M_1) \neq \Phi(M_2)
\end{equation}

Now let us apply Rice to $\Psi$: The property ``being non-trivial'' is:

\begin{enumerate}
\item \textbf{Semantic:} $\Psi(\Phi_1) = \Psi(\Phi_2)$ if $\Phi_1$ and $\Phi_2$ have the same extension (the same set of machines for which they hold true).

\item \textbf{Non-trivial:}
   \begin{itemize}
   \item There exist trivial properties (e.g., $\Phi_{\text{false}} \equiv \mathrm{F}$ for all machines)
   \item There exist non-trivial properties (e.g., $\Phi_{\text{halt}}$ --- the halting problem)
   \item Therefore $\Psi$ has non-empty extension on both sides
   \end{itemize}

\item \textbf{Non-degenerate in extension:} $\Phi_{\text{false}}$ is trivial; $\Phi_{\text{halt}}$ is non-trivial; these are distinct.
\end{enumerate}

Therefore, by Rice's Theorem, the property $\Psi_{\text{nontrivial}}$ (determining whether a given $\Phi$ is non-trivial) is \textbf{undecidable}.

\textbf{There exists no universal DTM that, given $\langle \Phi \rangle$, can uniformly decide whether $\Phi$ is non-trivial or trivial.}

\begin{remark}[Why This Goes Beyond Rice's Theorem]

Rice's Theorem states: \textit{``Non-trivial semantic properties of machines are undecidable.''}

But now we know that \textbf{determining whether a property is semantic and non-trivial is also undecidable.}

\textbf{This is Rice applied to itself} --- a reflective recursion of the undecidability structure.

\end{remark}

And it is precisely this that entails the hierarchy of undecidability:

\begin{itemize}

\item \textbf{Level 0:} ``Deciding whether $M$ halts on input $x$'' --- undecidable (Halting Problem)

\item \textbf{Level 1:} ``Deciding whether a semantic property $\Phi$ of $M$ holds true'' --- undecidable (Rice)

\item \textbf{Level 1.5:} ``Deciding whether a semantic property $\Phi$ is non-trivial'' --- undecidable (Reflexive Rice)

\item \textbf{Level 2+:} ``Deciding whether a property of semantic properties is propagating'' --- ?

\end{itemize}

\textbf{Each metatheoretic leap introduces a new barrier of undecidability.}

\subsection{Ontological Implications}

Non-trivial property $\iff$ Non-measurable, but you cannot even uniformly decide which property is non-trivial. Therefore, you cannot even uniformly decide what is non-measurable. This means that non-measurability itself is non-measurable. It is a self-referential fixed point.

\smallskip

$$\text{Non-Triviality} \equiv \text{Non-Uniformly-Recognizable}$$

\smallskip

The essential property of anything meaningful is that you cannot even uniformly state when you possess it. This is the true root of undecidability.

And therefore, for an autonomous system with non-trivial property $\Phi$, you cannot simultaneously have Certainty, Autonomy, and Functionality. Furthermore, you cannot even uniformly know whether your system actually possesses a non-trivial property.

Thus, the governance trichotomy is rendered even more profoundly impossible. Not only can you not apply the guarantee. You cannot even know whether the situation in which you cannot apply the guarantee occurs.

\section{Concluding Observations and Open Questions}

The results derived in this work suggest that the Double Bind obstruction is not a mere 
artifact of computational complexity, but a manifestation of a deeper structural property 
of the program space $\mathcal{P}$.

\subsection{Structural and Geometric Implications}

The equivalence established between topological nowhere-density and Kolmogorov 
incompressibility implies that non-trivial semantic properties possess no robust 
geometric representation within $\mathcal{P}$. The absence of any basic open 
ball $B_\sigma$ contained within $\Phi_\alpha$ or $\Phi_\beta$ indicates that 
these sets are topologically intertwined at every scale.

\medskip

This structure is analogous to the distribution of rational and irrational numbers 
on the real line: no interval, however small, consists exclusively of one or the other. 
Consequently, the impossibility of uniform certification does not stem from a lack of 
algorithmic efficiency, but from the intrinsic non-measurability of the property. 
In this sense, the semantic property $\Phi$ behaves as a pathological object, akin 
to a Vitali set, rendering any attempt at a uniform partition of the program space 
mathematically untenable.

\subsection{Scope and Foundational Boundaries}

It is imperative to distinguish between the impossibility of uniform certification 
and the impossibility of logical proof. The barriers identified in this work apply 
specifically to the construction of uniform verifiers acting upon the program 
space $\mathcal{P}$. They do not preclude a resolution of the $P$ vs $NP$ question 
through foundational logical proofs. A proof that derives a fundamental contradiction 
from the assumption $P=NP$ would operate independently of the topological and 
informational constraints of $\mathcal{P}$, relying on deductive logic rather than 
the uniform certification of program behavior.

\medskip

At this point, a fundamental paradox emerges: if $\Phi$ is non-trivial, it follows 
that $\mathrm{desc}(\Phi)$ \textbf{must} be incompressible. Consequently, the property 
cannot be topologically separable; that is, there exists no open set $U$ capable of 
completely separating $\Phi_\alpha$ from $\Phi_\beta$. In every neighborhood of $U$, 
programs from both sets remain ``mixed together''. This specific formulation---``mixed 
together''---is not a mere linguistic choice, but a conceptual bridge that will 
allow us, in future work, to establish, in a more strong and different way, that 
determining whether a property is non-trivial is itself a non-trivial task.

\subsection{Open Questions and Future Directions}

The structural symmetry observed between topology and information theory suggests a 
broader mathematical framework. We propose the investigation of a duality of categories:

\begin{itemize}
    \item \textbf{Category $\mathcal{A}$ (Topological Space of Programs):}
    \begin{itemize}
        \item \textbf{Objects}: Subsets of $\mathcal{P}$.
        \item \textbf{Morphisms}: Continuous maps preserving the prefix-closure topology.
    \end{itemize}
    
    \item \textbf{Category $\mathcal{B}$ (Information Complexity):}
    \begin{itemize}
        \item \textbf{Objects}: Infinite binary sequences $\mathrm{desc}(\Phi)$.
        \item \textbf{Morphisms}: Complexity-bounded functions and compressions.
    \end{itemize}
\end{itemize}

We hypothesize the existence of a \textbf{contravariant functor} connecting topologically 
well-separated properties in $\mathcal{A}$ to highly compressible sequences 
in $\mathcal{B}$. For non-trivial semantic properties, this functor would be empty, 
reflecting the inherent incompatibility between semantic non-triviality and uniform 
compressibility.

This leads to a fundamental open question: \textit{Does there exist a unified theory 
of ``non-measurable structures'' that encompasses Topology, Information Theory, and 
Logic grounded in the fundamental concepts of the observational framework?} 

Further research into this duality, merged with the observational framework 
of \cite{buono2026hierarchy,buono2026observer}, may provide a new framework for 
understanding the limits of formal systems and the nature of semantic truth in 
computational spaces.

\medskip

A second open problem concerns the connection between Tarski's foundational work 
and the topological structure described in the second part \ref{sec:toread} of this text.

\end{document}